\documentclass[11pt]{article}
\usepackage[utf8]{inputenc}
\usepackage{standalone}
\usepackage{float}
\pdfoutput=1
\usepackage[margin=1.1in,letterpaper]{geometry}
\usepackage[numbers]{natbib}
\usepackage{physics}
\usepackage{amsmath,amssymb,amsthm}
\usepackage{mathtools}
\usepackage{bm}
\usepackage{enumitem}
\usepackage{graphicx}
\usepackage{multirow}
\usepackage{subcaption}
\usepackage[colorlinks=true, urlcolor=blue, citecolor=black]{hyperref}
\usepackage{authblk}
\usepackage[dvipsnames]{xcolor}
\usepackage{algorithm}
\usepackage{algpseudocode}

\newtheorem{theorem}{Theorem}[section]
\newtheorem{lemma}[theorem]{Lemma}
\newtheorem{corollary}[theorem]{Corollary}
\newtheorem{proposition}[theorem]{Proposition}

\theoremstyle{definition}
\newtheorem{definition}[theorem]{Definition}
\newtheorem{remark}[theorem]{Remark}

\numberwithin{equation}{section}
\newcommand{\iu}{\mathrm{i}}
\newcommand{\hl}[1]{\textcolor{magenta}{#1}}

\title{Ollivier's Ricci Curvature on Complex-weighted Graphs}

\author[1,2,3,4]{Yu Tian}
\author[5]{Eleanor Wiesler} 
\author[6]{Melanie Weber} 
\affil[1]{\textit{Center for Systems Biology Dresden, 01307 Dresden, Germany}}
\affil[2]{\textit{Max Planck Institute for the Physics of Complex Systems, 01187 Dresden, Germany}}
\affil[3]{\textit{Max-Planck Institute of Molecular Cell Biology and Genetics,
01307 Dresden, Germany}}
\affil[4]{\textit{Cluster of Excellence, Physics of Life, TU Dresden, 01307 Dresden, Germany}}
\affil[5]{\textit{Harvard College, Cambridge, MA, US}}
\affil[6]{\textit{Harvard University, School of Engineering and Applied Sciences, Cambridge, MA, US}}

\date{}

\begin{document}

\maketitle
\begin{abstract}
  Understanding the geometry of complex networks is critical for effective modeling and analysis across domains. While discrete notions of Ricci curvature have emerged as powerful tools for characterizing both local and global network structure, existing formulations are largely confined to undirected networks with real-valued weights. This limits the use of curvature-based analysis of directional and complex-weighted relations that arise naturally in many applications, from social and biological systems to quantum and signal-processing networks. In this work, we introduce a principled extension of Ollivier’s Ricci curvature to complex-weighted graphs, which encompasses directed graphs as a special case. We establish fundamental theoretical properties of this new notion, including relations to the magnetic Laplacian and combinatorial upper and lower bounds that relate curvature to cycle structure in local neighborhoods. We further develop computational methods for curvature estimation and demonstrate their utility in community detection on directed networks.
\end{abstract}



\section{Introduction}
Understanding the geometric structure of graphs is central to effective learning and inference on relational data, both directed and undirected. \emph{Directed} relational data arises naturally in a broad range of domains, including following–follower relationships in social networks, regulatory relations in genetic networks, and synaptic connections in the brain. Characterizing the geometric structure of such data can provide insight into both local and global properties of the underlying complex systems, with direct implications for modeling, analysis, and prediction tasks.

Among the tools developed for analyzing relational data, discrete notions of Ricci curvature have proven particularly valuable. Drawing on analogies to its continuous counterpart, discrete Ricci curvature, when computed at the edge level, encodes both local structural properties, such as density and clustering coefficients, as well as characterizations of global structure, such as information routing efficiency and the presence of communities. These dual capabilities have motivated a wide range of curvature-based approaches in both Network Science~\citep{weber2017characterizing,ni2019community,sia2019ollivier,Znaidi2023,tian2025curvature,fesser2023augmentations} and Graph Machine Learning~\citep{topping2022understanding,fesser2023mitigating,farzam2024from,feng2024graph,fesser2024effective}.

Despite these advances, existing notions of discrete Ricci curvature are predominantly limited to undirected graphs and graphs with real-valued weights. This restriction is significant, as directionality frequently arises in practical applications, such as the aforementioned complex systems. Moreover, graphs with complex-valued weights, a generalization of directed graphs, arise commonly in domains such as quantum physics, signal processing, and systems biology~\citep{bottcher2024complex,tian2026generalizing}. Recent literature has begun to extend classical graph characteristics and algorithmic methods to settings that involve directed or complex-weighted structures~\citep{bottcher2024complex,tian2024structural,fanuel2017magnetic,zhang2021magnet}. However, little is known about discrete Ricci curvature in these settings. 

Naïve strategies for extending curvature notions, such as discarding edge directionality or reducing complex weights to their absolute values, result in the loss of crucial structural information. Such simplifications can dramatically distort both local and global geometric properties. While recent efforts have attempted to generalize Forman’s~\citep{forman} and Ollivier’s~\citep{Ol} curvature to directed graphs~\citep{sreejith2016forman,saucan2019discrete}, these approaches remain limited in scope. To the best of our knowledge, no extensions to complex-weighted networks have been proposed.

In this paper, we begin to address this gap by introducing a principled notion of Ollivier’s Ricci curvature for complex-weighted networks. Our formulation naturally subsumes directed networks as a special instance. We investigate both theoretical and computational aspects of this new notion. This includes combinatorial upper and lower bounds that link curvature to the effective length of cycles in an edge’s two-hop neighborhood, as well as connections to the magnetic Laplacian.  Finally, we explore the utility of these curvature notions in applications, with a focus on community detection in directed graphs.\footnote{The code is publicly available at:~\url{https://github.com/Weber-GeoML/ComplexGraphCurvature}.}

\subsection{Related Work}
Complex-weighted networks have recently received interest as a natural generalization of directed graphs, with applications in quantum physics, signal processing, systems biology, and control theory~\cite{bottcher2024complex}. Existing network analyses in this setting primarily rely on spectral methods that leverage generalizations of the Graph Laplacian~\citep{berkolaiko2013nodal,colin2013magnetic} and stochastic tools such as random walks~\cite{tian2024structural}. The special case of directed graphs has been studied extensively, with several classical graph characteristics~\citep{babul2025strong,gallo2025patterns} and algorithms~\citep{fanuel2017magnetic,he2022msgnn,zhang2021magnet} extended to this setting. However, while several curvature-based methods have been proposed in network science~\citep{weber2017characterizing,ni2019community,sia2019ollivier,Znaidi2023,tian2025curvature} and graph machine learning~\citep{topping2022understanding,fesser2023mitigating,farzam2024from,feng2024graph,fesser2024effective}, they leverage undirected and real-weighted notions. 

Prior literature has begun to investigate discrete notions of curvature on directed graphs, albeit with limitations.
\cite{sreejith2016forman} introduce a notion of Forman's Ricci curvature on the vertices of directed graphs that captures flow through the respective directed edges. An extension of this notion that incorporates curvature contributions of higher-order structures was proposed in~\citep{saucan2019discrete}. The authors further introduce an analogous flow-based notion of Ollivier's Ricci curvature on vertices with measures defined on the in- and out-neighbors respectively, albeit without accounting for degeneracy. 

To the best of our knowledge, this paper introduces the first notion of discrete Ricci curvature on complex-weighted graphs and the first curvature-based algorithm for community detection on directed networks. 

\subsection{Overview of the paper}
The paper is structured as follows: In section~\ref{sec:background} we introduce background on complex-weighted networks and Ollivier's notion of discrete Ricci curvature, which forms the basis for our curvature notion.
Section~\ref{sec:curvature} defines extensions of Ollivier's Ricci curvature on complex-weighted graphs. We discuss different parametrizations and simplifications in special cases.
Section~\ref{sec:theory} investigates the theoretical properties of the new notion, including combinatorial upper and lower bounds and relations to the magnetic Laplacian and structural balancedness. Section~\ref{sec:applications} explores applications of the new notion for community detection in directed networks. Specifically, we introduce a community detection algorithm and investigate its performance in numerical experiments. Conclusions and directions for future work are given in Section~\ref{sec:conclusions}.

\section{Background and Notation}\label{sec:background}

\subsection{Directed Graphs}
We define a \textit{directed graph }or \textit{digraph} $D = (V, E)$ as a non-empty finite set $V$ of \textit{vertices} and a finite set $E$ of ordered pairs of distinct vertices called \textit{edges}. An edge $e \in E$ connecting vertices $u$ and $v$ will be denoted $\overrightarrow{uv}$.  It is permissible for a digraph to have an edge $\overrightarrow{uv}$ and an arc $\overrightarrow{vu}$ in which two vertices are connected by two edges of opposite direction. Just as with undirected graphs, it is crucial to specify neighborhoods of vertices when working with digraphs. However, if a vertex $v$ is an \textit{end-vertex} of two edges $\overrightarrow{uv}$ and $\overrightarrow{vw}$, then the relationships between its two neighbors $u$ and $w$ are fundamentally different. Here, $u$ is an \textit{in-neighbor} of $v$ because the edge is pointing from the \emph{start node} $u$ to the \emph{end node} $v$. We call $w$ an \textit{out-neighbor} of $v$ because the edge is pointing from $v$ to $w$. For a digraph $D$, the \textit{in-neighborhood} of a vertex $v$ is the set of all \textit{in-neighbors} of $v$ and is denoted as $\mathcal{N}^{in}(v)$. Similarly, the \textit{out-neighborhood} of $v$ is the set of all \textit{out-neighbors} of $v$ and is denoted as $\mathcal{N}^{out}(v)$.
We further specify the degree of a vertex based on in-going edges and out-going edges. We define the \textit{in-degree} $d^{in}(v)$ of a vertex $v$ as the number of its in-neighbors, or in other words, the number of edges for which $v$ is the end node. We define the \textit{out-degree} $d^{out}(v)$ of a vertex $v$ analogously as the number of its out-neighbors or the number of edges for which $v$ is the start node. 

\subsection{Complex-Weighted Graphs}
We consider the more general case of graphs with \textit{complex edge weights} where each edge $\overrightarrow{uv}\in E$ is characterized by a complex value $W_{uv}\in\mathbb{C}$. We decompose the complex weight as $W_{uv} = r_{uv}e^{\iu \varphi_{uv}}$, where $r_{uv} \ge 0$ indicates the \emph{magnitude} of the value while $\varphi_{uv}\in [0, 2\pi)$ is the \emph{phase} (or the \emph{argument}). For example, if each node represents an image and the edge weights characterize the similarity between them, $r_{uv}$ can store the maximum similarity value between the two images subject to rotations, while the extra freedom introduced by $\varphi_{uv}$ can be used to record the rotation value associated with the maximum. The in-degree of a vertex $v$ sums over the magnitudes of incoming edges to $v$, $d^{in}(v) = \sum_u r_{uv}$, and the out-degree of a vertex $v$ sums over the magnitudes of outgoing edges from $v$, $d^{out}(v) = \sum_u r_{vu}$. Clearly, classic graphs with only positive edge weights can be recovered if all phases are $0$, while signed graphs, which allow negative edge weights, can be obtained if all phases can only be $0$ or $\pi$. 
Furthermore, we define the \textit{phase of a path (cycle)} as the sum of phases of composing edges, and we still restrict the phase to lie in $[0,2\pi)$. In the special case when the complex edge weights satisfy that $W_{vu} = \bar{W}_{uv} = r_{uv}e^{-i\varphi_{uv}}$, the important notion of \textit{structural balance} has been established, where all cycles have phase $0$. 

\subsection{Ollivier's Ricci curvature on unweighted graphs}
Ollivier's notion of Ricci curvature~\citep{Ol} (short: ORC) is one of the most widely considered discrete notions of curvature.  
To compute ORC on undirected graphs, one endows the 1-hop neighborhoods of neighboring vertices $u,v$ with uniform measures $m_{i}(z) := \frac{1}{{\rm deg}(i)}$, where $z$ is a neighbor of $i$ and   $i \in \{u,v\}$.  Such measures are induced by the transition probabilities of uniform random walks starting at $u,v$. 
The transportation cost between the two measures is characterized by the Wasserstein-1 distance
\begin{equation}
    W_1(m_{u}, m_{v}) = \inf_{m \in \Gamma(m_{u},m_{v})} \int_{(z,z') \in V \times V} d(z,z') m(z,z') \; dz \; dz' \; .
    \label{equ:W1}
\end{equation}
Here, $\Gamma(m_{u},m_{v})$ denotes the set of all measures over $V \times V$ with marginals $m_{u},m_{v}$.  Ollivier proposes a curvature of the form~\citep{Ol}
\begin{equation}\label{eq:orc}
	\kappa (u,v) := 1 - \frac{W_1 (m_{u}, m_{v})}{d_G(u,v)} \; ;
\end{equation}
here, $d_G(u, v)$ denotes the shortest-path distance between $u$ and $v$ in the graph $G$. 
This notion is motivated by a fundamental connection between Ricci curvature and the behavior of random walks on smooth manifolds: Two random walks initiated at nearby points will stay near each other with high probability, if the manifold has locally positive Ricci curvature, and they will diverge with high probability, if it has locally negative Ricci curvature. Eq.~\ref{eq:orc} captures this intuition with respect to measures $m_u,m_v$ induced by uniform random walks initiated at neighboring nodes $u,v$.

\section{Ollivier's Ricci curvature on complex-weighted graphs}\label{sec:curvature}
Our definition of a principled Ricci curvature on complex-weighted graphs will follow Ollivier's notions (Eq.~\ref{eq:orc}). In order to define this extension, we need to define (1) measures $m_u, m_v$, and (2) a shortest-path distance $d_G(\cdot,\cdot)$ that are consistent with the complex weighting scheme of the graph. In the following we assume for simplicity magnitude $r_{uv} = 1$ for all $\overrightarrow{uv} \in E$; all definitions and results can be naturally extended to edge weights of arbitrary magnitude.

\subsection{Warm-up: Directed graphs}
We begin with the special case of directed graphs as a warm-up. For any given directed edge we define measures on the neighborhoods of its \textit{start} and \textit{end} nodes induced by uniform random walks. Let $u \in V$ be the \textit{start} node of a directed edge $\overrightarrow{uv} \in E$. Then we define a uniform measure on the in-neighbors of $u$: 
\[
m_u(x) =
\begin{cases}
\frac{1}{d_{\text{in}}(u)} & \text{if } \overrightarrow{xu} \in E, d_{in}(u) > 0\\
1 & \text{if } d_{\text{in}}(u) = 0, x = u  \\
\end{cases}
\]
Observe that if the node $u$ has no incoming edges, then with probability 1 the random walk will stay at node $u$, i.e., $d_{\text{in}}(u) = 0$ implies that $x = u$. 
If we instead consider the \textit{end} node $v$, then we can define a uniform measure on its out-neighbors:
\[
m_v(x) =
\begin{cases}
\frac{1}{d_{\text{out}}(v)} & \text{if } \overrightarrow{vx} \in E, d_{out}(v) > 0  \\
1 & \text{if } d_{\text{out}}(v) = 0, x = v \\
\end{cases}
\]
Similarly to the above, observe that if a node $v$ has no outgoing edges, then with probability 1 the random walk will stay at node $u$, i.e., if $d_{\text{out}}(v) = 0$, then $x=v$.

Note that the definition of the measures ensures that if a node has no incoming or outgoing edges, the measure places all mass on itself, resulting in a degenerate walk. This construction reflects directionality, allowing curvature to capture asymmetric structural features.

\begin{remark}
    A similar construction allows for extending the related Lin-Lu-Yau to directed graphs~\cite{lin-lu-yau}. Here, the measures on node neighborhoods are induced by $\alpha-$random walks: For \( \alpha \in [0, 1] \) and any vertex \( x \), they define the probability measure  $m_x^\alpha(v) = \frac{(1 - \alpha)}{d_x}$ if $v$ shares an undirected edge with $x$, $m_x^\alpha(v) = \alpha$ if $v=x$ and $m_x^\alpha(v) = 0$ otherwise. In analogy to our construction above, we can define in the directed case: 
\[
\begin{array}{c@{\quad\quad}c}
\displaystyle
m_u^\alpha(x) =
\begin{cases}
\frac{(1 - \alpha)}{d_{\text{in}}(u)} & \text{if } \overrightarrow{xu} \in E \\
\alpha & \text{if } d_{\text{in}}(u) > 0, x=u  \\
1 & \text{if }d_{\text{in}}(u) = 0, x=u
\end{cases}
&
\displaystyle
m_v^\alpha(x) =
\begin{cases}
\frac{(1 - \alpha)}{d_{\text{out}}(v)} & \text{if } \overrightarrow{vx} \in E  \\
\alpha & \text{if } d_{\text{out}}(v) > 0, x=v \\
1 & \text{if } d_{\text{out}}(v) = 0, x=v
\end{cases}
\end{array}
\]
\end{remark}

In the directed case, $d_G(\cdot,\cdot)$ is the usual (directed) shortest-path distance. For instance, if the graph is unweighted, then $d_G(u,v)=1$ for all $\overrightarrow{vx} \in E$.

\begin{figure}[htbp]
    \centering
    \hspace*{-1em}
    \includegraphics[width=.8\textwidth]{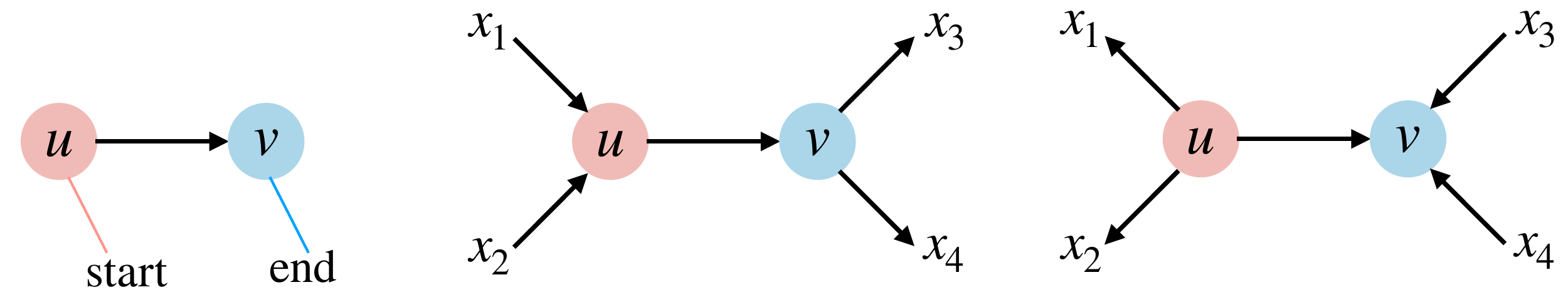}
    \caption{(Left) Example directed edge. (Middle) Standard example. (Right) Degenerate example.}
    \label{fig:digraph}

\end{figure}

\begin{figure}[htbp]
    \centering
    \begin{tabular}{cc}
        \includegraphics[width=.35\textwidth]{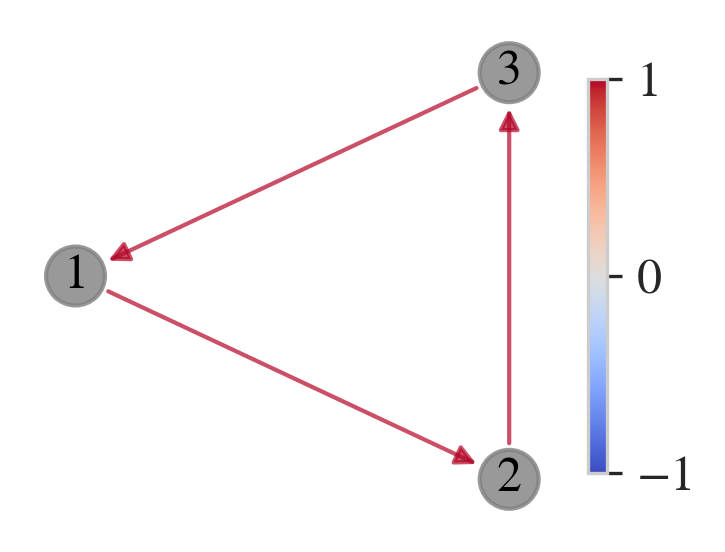} & \includegraphics[width=.35\textwidth]{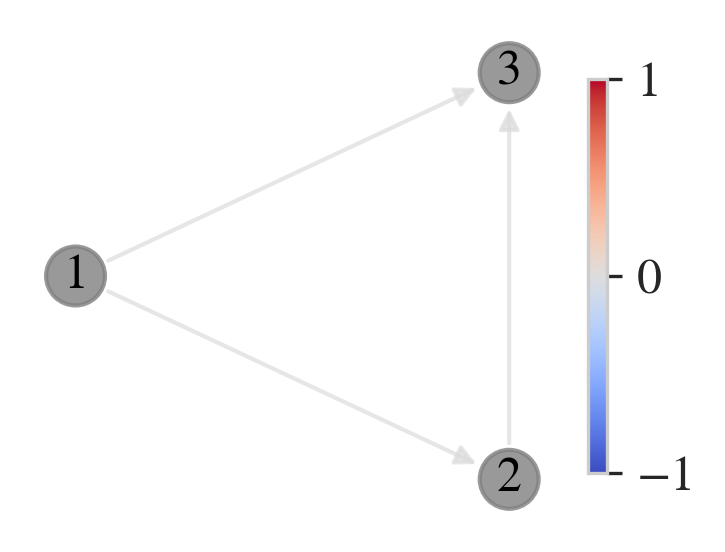}
    \end{tabular}
    \caption{ORC for the directed cycle graph (left) and the directed flow graph (right).}
    \label{fig:dir-G-eg-tri}
\end{figure}

\subsection{Complex-weighted graphs}
We now consider the general case, where graphs have complex-valued edge weights. The key question is how to incorporate the phase information appropriately into the formulation. In this section, we will extend both the probability measures and the distance measures for complex weights, thus defining a principled notion of ORC for complex-weighted graphs.

\paragraph{Construction of measure}
Recall that in undirected graphs, the measures on the neighborhoods are induced by the transition probabilities of random walks starting at $u,v$; for directed graphs, we defined separate measures on the start and end vertices over the in-neighborhood and out-neighborhood, respectively. In directed complex-weighted graphs, we consider a complex random walk with a natural relation to the connection Laplacian operator on the corresponding manifold \cite{singer2012vdm,tian2024structural}. Specifically, each random walker carries a phase; at each time step, it chooses one of its out-neighbors by probabilities proportional to the magnitude of the edge weights. As it traverses an edge, its phase becomes the sum of its original phase and the phase of the edge weight. Formally, this random walk induces the following neighborhood measures:
\begin{definition}[Measures induced by complex random walk]
    Let $u \in V$ be the start node of a directed edge $\overrightarrow{uv} \in E$. Then we define a uniform measure on the in-neighbors of $u$: 
    \[
    m_{\text{in},u}(x) =
    \begin{cases}
    \frac{r_{xu}}{d_{\text{in}}(u)} & \text{if } \overrightarrow{xu} \in E, d_{in}(u) > 0\\
    1 & \text{if } d_{\text{in}}(u) = 0, x = u  \\
    \end{cases}
    \]
    For the end node $v$, we define a uniform measure on its out-neighbors:
    \[
    m_{\text{out}, v}(x) =
    \begin{cases}
    \frac{r_{vx}}{d_{\text{out}}(v)} & \text{if } \overrightarrow{vx} \in E, d_{out}(v) > 0  \\
    1 & \text{if } d_{\text{out}}(v) = 0, x = v \\
    \end{cases}
    \]
\end{definition}
\begin{remark}
    Note that we assume $r_{uv} = 1$ for all $\overrightarrow{uv}\in E$, hence, we do not separate different meanings of edge weights, affinity versus distance, explicitly in our definitions. In the case of heterogeneous weight magnitudes, we can also consider measures via Gaussian kernels since edge weights are commonly assumed to represent affinity. Our results can also be generalized to the heterogeneous case.
\end{remark}

\paragraph{Shortest-walk distance}
Complex numbers are not fully ordered; therefore, the shortest-path distances are not directly applicable to complex-weighted graphs, and to develop appropriate distance measures on complex-weighted graphs is still an active area of research; we refer the reader to \cite{bottcher2024complex} for a recent review on complex-weighted graphs. Inspired by the complex random walk discussed above, instead of directly incorporating the phase information into the calculation, we lift it to constrain the valid set of directed walks connecting the two vertices. 
Specifically, we relax the requirement of ``paths'' to be general ``walks'' where repetition of vertices is allowed: for each pair of vertices $u,v\in V$, we define the \textit{shortest-walk distance} from $u$ to $v$ subject to the phase constraint with $\varphi\in [0, 2\pi)$ as 
\begin{align}
    d_C^{\varphi}(u,v) =& \min_{\mathcal{P}_{uv}: \theta(\mathcal{P}_{uv}) = \varphi}\sum_{\overrightarrow{ij}\in \mathcal{P}_{uv}} r_{ij},
\end{align}
where $\mathcal{P}_{uv}$ is a directed walk from $u$ to $v$, and the function $\theta(\cdot)$ returns the phase of a walk, i.e., the sum of the phases of its composing edges, modulo $2\pi$. {We use the convention $d_C^0(u,u)=0$; if $u=v$ and $\varphi\ne 0$, the minimum is taken over nontrivial closed walks from $u$ to itself with phase $\varphi$, and is $+\infty$ if no such walk exists.} 
Then, in particular, to define the ORC for a directed edge $\overrightarrow{uv}$ with complex weights, we only consider a walk as a shortcut from an in-neighbour of $u$, $u'$, to an out-neighbour of $v$, $v'$, if the random walker ends up with the same phase no matter following the walk or going through $u,v$ to $v'$. That is, the walk has the same phase as the directed path $u'uvv'$; see Definition \ref{def:orc-swd}. 
\begin{definition}[Shortest-walk distance for ORC]
    For each directed edge $\overrightarrow{uv}\in E$, we consider the shortest-walk distance between an in-neighbor of $u$, $u'$, and an out-neighbor of $v$, $v'$, with the phase constraint $\varphi = \theta(u'uvv')$: 
    \begin{align}
        d_C^{\theta(u'uvv')}(u',v') =& \min_{\mathcal{P}_{u'v'}: \theta(\mathcal{P}_{u'v'}) = \theta(u'uvv')}\sum_{\overrightarrow{ij}\in \mathcal{P}_{u'v'}} r_{ij},
    \end{align}
     where $u'uvv'$ is the specific directed paths from $u'$ going through $u,v$ to $v'$, and the function $\theta(\cdot)$ returns the phase of a walk. 
     \label{def:orc-swd}
\end{definition}
With this convention, the shortest-walk distance is an extended directed distance: $d_C^0(u,u)=0$, all other finite values are positive, and it satisfies the phase-compatible triangle inequality
\begin{align*}
    d_C^{\theta(u'uvv')}(u',v') + d_C^{\theta(v'xww')}(v',w') \ge d_C^{\theta(u'uvv'xww')}(u',w'),
\end{align*}
It is not symmetric in general.
We note that $u'uvv'$ is also a feasible walk for the minimization in the distance, thus $d_C^{\theta(u'uvv')}(u',v') \le r_{u'u} + r_{uv} + r_{vv'}$ (i.e., $3$, since we assume $r_{\cdot\cdot} = 1$).
Also, in the degenerate case when the phase of each edge is $0$, the shortest-walk distance returns the same value as the shortest-path distance defined on graphs with positive weights.  
\begin{figure}[htbp]
    \centering
    \includegraphics[width=.9\linewidth]{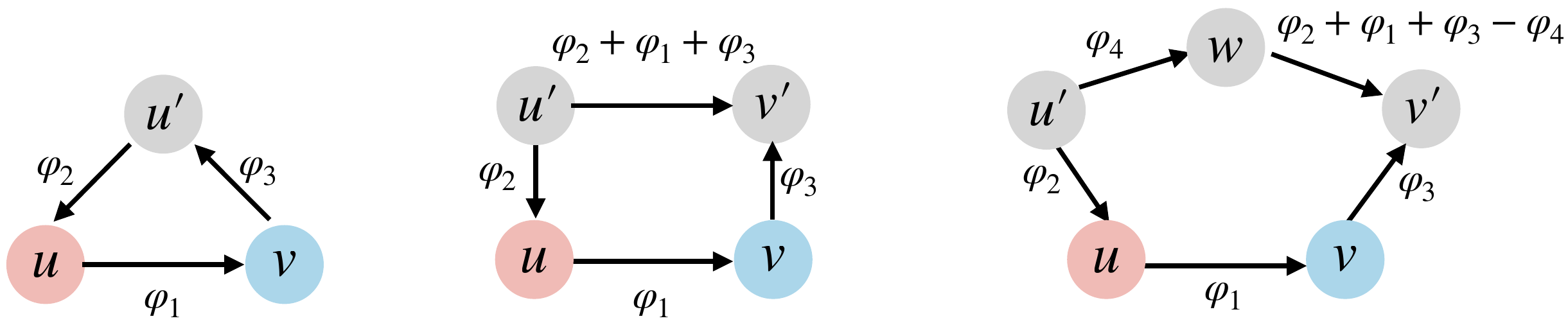}
    \caption{Three different cases to provide shortcuts between the in-neighbor of the start vertex $u$ and the out-neighbor of the end vertex $v$, where the phases of edge weights are shown next to the edges ($\varphi_i\in [0, 2\pi), \, i\in\{1,2,3,4\}$). }
    \label{fig:complexorc-eg}
\end{figure}

With the two key ingredients of the measures and the shortest-walk distance, we define the ORC in complex-weighted graphs as in Definition \ref{def:orc-complex}. 
\begin{definition}[ORC for complex weights]
    We define the ORC in complex-weighted graphs as 
    \begin{equation}\label{eq:orc-complex}
    	\kappa (u,v) := 1 - \frac{W_1 (m_{in, u}, m_{out, v})}{d_C^{\theta(uv)}(u,v)} \; ,
    \end{equation}
    where
    \begin{equation}
        W_1(m_{in, u}, m_{out, v}) = \inf_{m \in \Gamma(m_{in, u},m_{out, v})} \int_{(z,z') \in V \times V} {d_C^{\theta(zuvz')}(z,z')} m(z,z') \; dz \; dz' \; .
        \label{equ:W1-complex}
    \end{equation}
    {Here $z\in \operatorname{supp}(m_{in,u})$ and $z'\in \operatorname{supp}(m_{out,v})$, so $\theta(zuvz')$ denotes the phase accumulated along the directed walk from $z$ through $u,v$ to $z'$.}
    {Furthermore, the direct edge is feasible, so $d_C^{\theta(uv)}(u,v)\le r_{uv}$, with equality when $\overrightarrow{uv}$ is phase-geodesic.}
    \label{def:orc-complex}
\end{definition}

\begin{figure}[htbp]
    \centering
    \begin{tabular}{cc}
        \includegraphics[width=.35\textwidth]{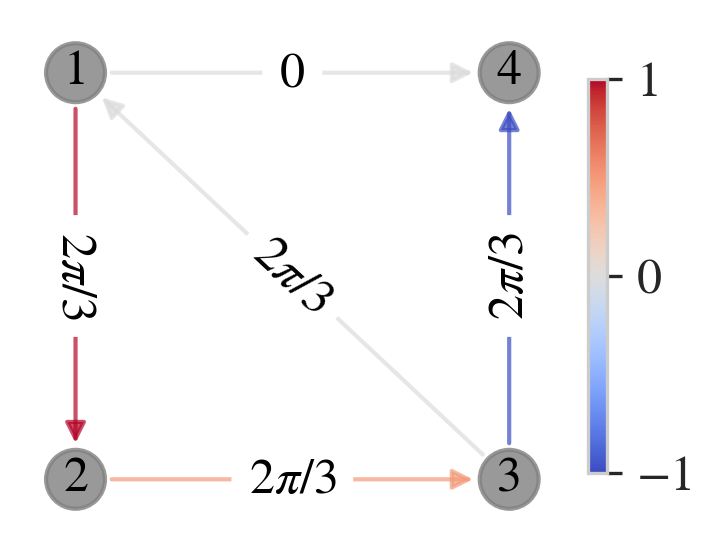} & \includegraphics[width=.35\textwidth]{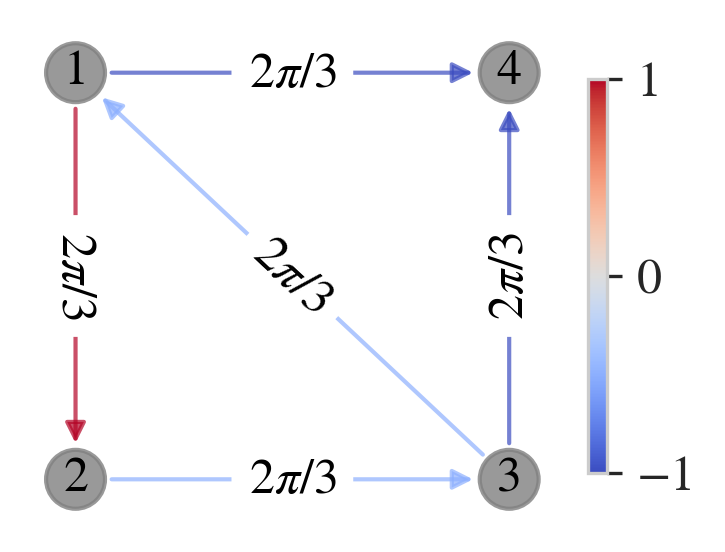}
    \end{tabular}
    \caption{ORC for the complex-weighted graphs with two different phase distributions, where (left) edge $(1,4)$ has the same phase as the path $\{(1,2), (2,3), (3,4)\}$ while (right) edge $(1,4)$ does not provide such a shortcut, with the phases of edges shown on the edges and ORC shown as edge colors.}
    \label{fig:dir-G-eg-quad}
\end{figure}

\subsection{Multi-layer Parametrization} 
The set of walks connecting two vertices can be infinitely large, even with the phase constraint, but not all of them are of equal importance to compute the shortest-walk distance. In this section, we propose an approach to simplify the computation based on a \textit{multi-layer parametrization}. 
We assume that there is some value $\varphi_0\in [0, 2\pi)$ such that for each edge $\overrightarrow{uv}\in E$, $\varphi_{uv}$ is an integer multiple of $\varphi_0$, and $\varphi_0 = 2\pi/k$, i.e., the edge phases are in the set $\{z\varphi_0\}_{z=0}^{k-1}$, where $k\varphi_0 = 2\pi$, or the complex weights (ignoring the magnitude) lie in a cyclic group $S_k^1 := \{\xi^j|j = 0, 1, \dots, k-1\}$ where $\xi = e^{2\pi\iu/k}$ is the primitive $k$-th root of unity. As we will show later, $k$ is the number of layers in the multi-layer parametrization.  

Specifically, we transform the phase information in the complex-weighted graph to the following \textit{$k$-layer representation}: 
\begin{itemize}
    \item[(i)] all vertices are copied inside each layer $1, \dots, k$;
    \item[(ii)] edges are placed according to their original phases, where edges of phase $z\varphi_0$ are placed between layers $1$ and $1+z$, $2$ and $2+z$, up to $k$ and $k+z$ (convention: $k+z = z$ for all integer $z < k$);
    \item[(iii)] all edges in the parameterization only carry the magnitude of the original edge weights, thus have weight $r_{uv}=1,\, \forall \overrightarrow{uv}\in E$.
\end{itemize}
Therefore, the adjacency matrix of the $k$-layer graph $G^{(k)}$ is 
\begin{align*}
    \mathbf{A}^{(k)} = 
    \begin{pmatrix}
        \mathbf{A}^{0}  & \mathbf{A}^{\varphi_0} & \cdots & \mathbf{A}^{(k-1)\varphi_0}\\
        \mathbf{A}^{(k-1)\varphi_0} & \mathbf{A}^{0} & \cdots & \mathbf{A}^{(k-2)\varphi_0} \\
        \vdots & \vdots & \ddots & \vdots \\
        \mathbf{A}^{\varphi_0} & \mathbf{A}^{2\varphi_0} & \cdots & \mathbf{A}^{0}
    \end{pmatrix},
\end{align*}
where $(\mathbf{A}^{z\varphi_0})_{uv} = 1$ if $z\varphi_0 = \varphi_{uv}$, the phase of edge $\overrightarrow{uv}$, and $0$ otherwise. 
We note that the adjacency matrix of the complex-weighted graph can be decomposed as $\mathbf{A} = \sum_{z=0}^{k-1}e^{z\varphi_0\iu}\mathbf{A}^{z\varphi_0}$.

\begin{figure}[htbp]
    \centering
    \includegraphics[width=0.9\linewidth]{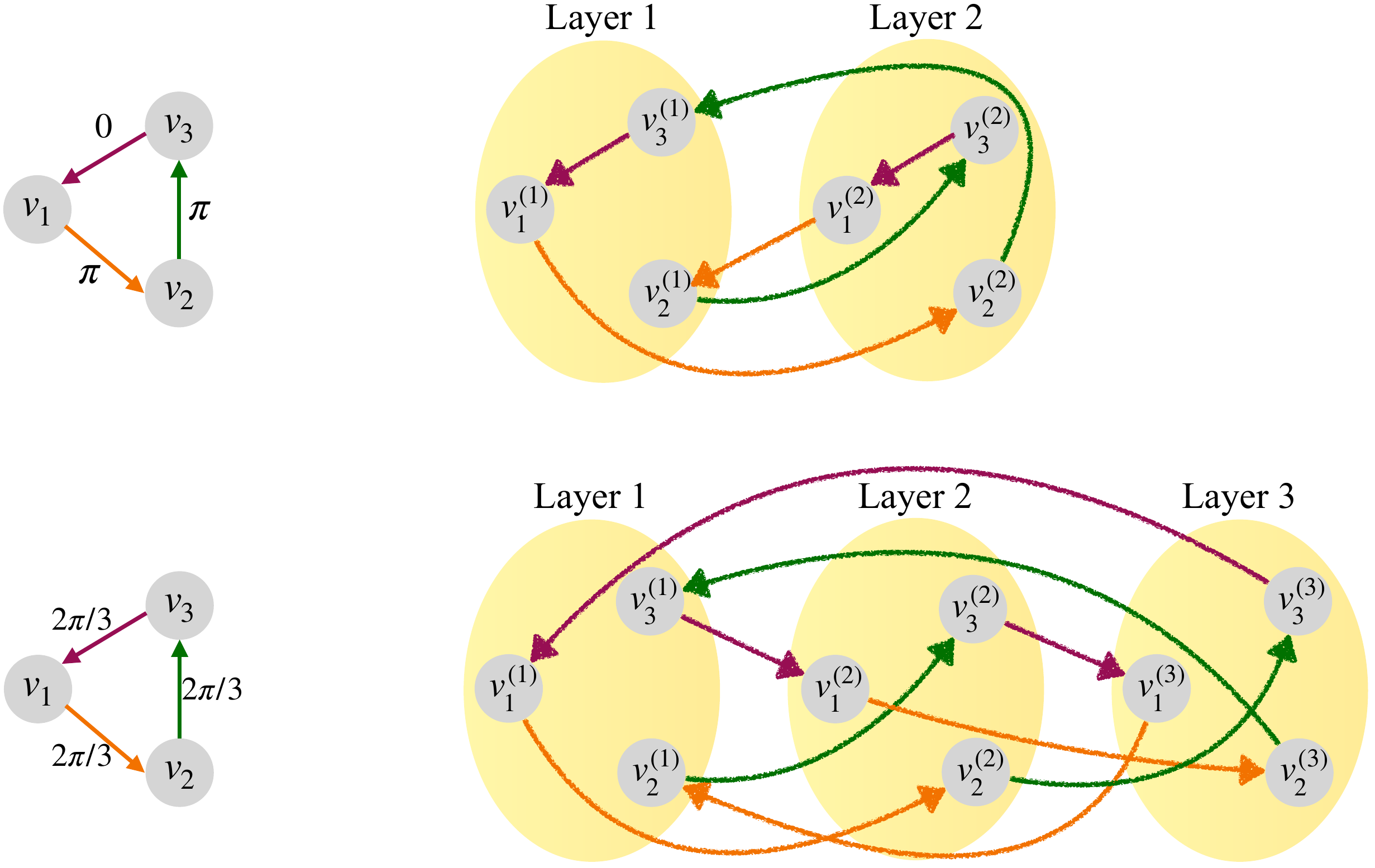}
    \caption{Example of complex-weighted graphs (left) and their multi-layer parametrization (right), where the phases of complex edge weights are shown next to the edges in the left figure (top: $\varphi_0 = \pi$, $k=2$; bottom: $\varphi_0 = 2\pi/3$, $k=3$), the edges in the right figure all have phase $0$, and copies of the same edge in the multi-layer parameterization are shown in the same color as the original edges.
    }
    \label{fig:multi-eg}
\end{figure}

\subsubsection{Equivalence of multi-layer parametrization and original formulation}
We show that the shortest-walk distance in the complex-weighted graph is equal to the shortest-path distance between the vertices in the corresponding layers. More importantly, the ORC of an edge $\overrightarrow{uv}$ in the complex-weighted graph is equal to the ORC of the corresponding edges in the $k$-layer graph.  {We present proof sketches for the formal results below; full proofs can be found in SM.}

\begin{lemma}
    For each $u,v\in V$, the shortest-walk distance from $u$ to $v$ of target phase $\theta = z\varphi_0$ in the complex-weighted graph, is equal to the shortest-path distance from $u$ in layer $i$ to $v$ in layer $i+z$ where $i=1,..., k$ with convention $k+i = i,\, \forall i$.
    {In particular, when $u=v$ and $z\not\equiv 0 \mod k$, this is the distance between distinct copies of $u$ in layers $i$ and $i+z$.}
    \label{lem:klayer-dist}
\end{lemma}
\begin{proof}[Proof sketch]
    We can show that: (i) for each walk from $u$ to $v$ in the complex-weighted graph of phase $\theta = z\varphi_0$, we can find a corresponding walk from $u$ in layer $i$ to $v$ in layer $i+z$ in the $k$-layer graph, and (ii) for each walk from $u$ in layer $i$ to $v$ in layer $i+z$ in the $k$-layer graph, we can also find a corresponding walk of phase $\theta = z\varphi_0$ from $u$ to $v$ in the complex-weighted graph. Therefore, the shortest-walk distance from $u$ to $v$ in the complex-weighted graph of phase $\theta = z\varphi_0$ is equal to the shortest-walk distance from $u$ in layer $i$ to $v$ in layer $i+z$ in the $k$-layer graph of phase $0$, which is equal to the shortest-path distance.
\end{proof}

\begin{theorem}
    For each $\overrightarrow{uv}\in E$, the ORC of $\overrightarrow{uv}$ in the complex-weighted graph with $\varphi_{uv} = z_0\varphi_0$ is equal to the ORC of $\overrightarrow{u^{(i)}v^{(i+z_0)}}$ in the (directed) $k$-layer graph, where $i=1,..., k$ with convention $k+i = i,\, \forall i$.
    \label{the:klayer-orc}
\end{theorem}
\begin{proof}[Proof sketch]
    We can first find the one-to-one correspondence between the in-neighbors of $u$ in the complex-weighted graph and the in-neighbors of $u$ in layer $i$ in the $k$-layer graph, and between the out-neighbors of $v$ in the complex-weighted graph and the out-neighbors of $v$ in layer $i+z_0$ in the $k$-layer graph. Then from Lemma \ref{lem:klayer-dist}, the shortest walk distance in the complex-weighted graph and the shortest path distance in the $k$-layer graph return the same value. With the two key components, we can finally prove that the ORCs return the same value.
\end{proof}

\begin{corollary}
    If for each $\overrightarrow{uv}\in E$, $\overrightarrow{vu}\in E$ with $r_{vu}=r_{uv}$ and $(\varphi_{uv} + \varphi_{vu}) \mod 2\pi = 0$, then the ORC of $\overrightarrow{uv}$ in the complex-weighted graph with $\varphi_{uv} = z\varphi_0$ is equal to the ORC of $\overrightarrow{u^{(1)}v^{(1+z)}}$ in the $k$-layer graph ignoring the direction.
    \label{cor:ORC-k-undi}
\end{corollary}
\begin{proof}[Proof sketch]
    The proof follows naturally from the construction.
\end{proof}
    
\subsubsection{Special case: complex weights from the magnetic Laplacian of digraphs}
In this section, we consider the specific case of a complex-weighted graph from the magnetic Laplacian of a digraph, for further illustration.

The magnetic Laplacian characterizes each directed network as a Hermitian matrix where the imaginary part carries the directional information. 
It is motivated by the dynamics of a free quantum mechanical particle on a graph under the influence of magnetic fluxes passing through the cycles in the network. Specifically, for a given digraph $H = (V,E_H)$, where $E_H$ is the edge set with $w_{ij}>0$ encoding the weight of each directed edge $(v_i,v_j)$ ($w_{ij} = 0$ if there is no edge), the magnetic Laplacian $\mathbf{L}^\theta = (L_{ij}^\theta)$ is
\begin{align*}
    L_{ij}^\theta = 
    \begin{cases}
        \sum_{h}w_{s}(i,h),\quad &\text{if } i=j\\
        - w_{s}(i,j)T_{i\to j}^\theta,\quad &\text{if } \{(v_i,v_j), (v_j,v_i)\} \cap E\ne \emptyset\\
        0,\quad &\text{otherwise}
    \end{cases}\; .
\end{align*}
where $w_{s}(i,j) = (w_{ij} + w_{ji})/2$ is the symmetrized weight, and $T_{i\to j}^\theta = \exp{(\iu \theta a(i,j))}$. $\theta\in [0, 2\pi)$ is an extra parameter introduced and
\begin{align}
    a(i,j) = 
    \begin{cases}
        1,\quad &\text{if } (v_i,v_j)\in E, (v_j,v_i)\notin E\\
        -1,\quad &\text{if } (v_i,v_j)\notin E, (v_j,v_i)\in E\\
        0,\quad &\text{if } (v_i,v_j)\in E, (v_j,v_i)\in E
    \end{cases}\; .
    \label{equ:magL-a}
\end{align}
The magnetic Laplacian is Hermitian by construction. For illustrative purposes, we choose $\theta = 2\pi/k$. 

\begin{figure}[htbp]
    \centering
    \includegraphics[width=0.98\linewidth]{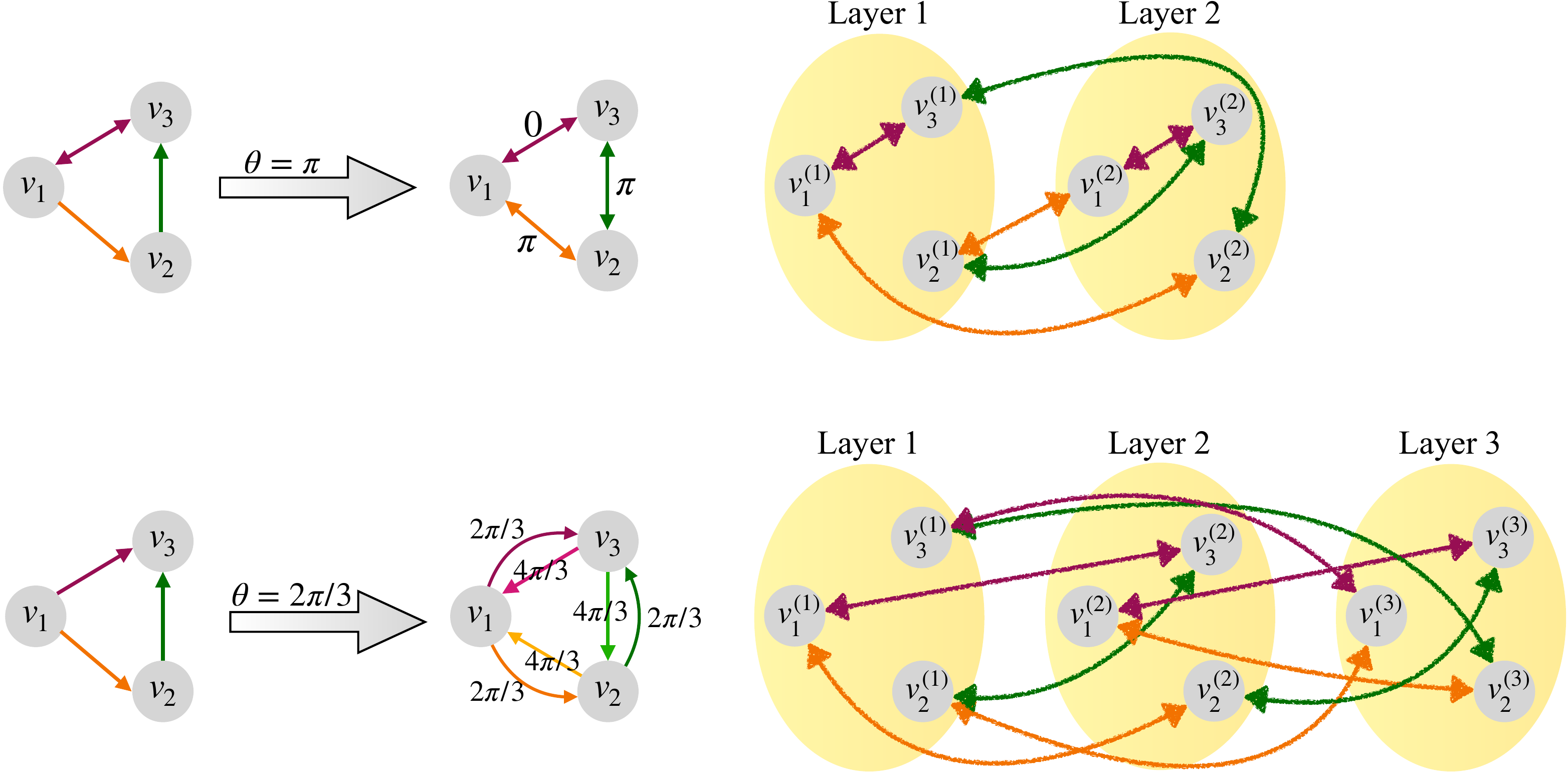}
    \caption{Example of digraphs $H$ (left), the complex-weighted graph corresponding to the magnetic Laplacian with $\theta$ (middle), and their $3$-layer graphs $G^{(3)}$ (right), of the directed cycle graph with one reciprocal edge (top) and the directed flow graph (bottom).}
    \label{fig:multi-dir-eg}
\end{figure}

In the complex-weighted graph from the magnetic Laplacian, denoted by $G$, the phases of edges can only be $\{0, 2\pi/k, 2(k-1)\pi/k\}$, where phase $0$ corresponds to bidirectional edges in $H$ ($a(i,j)=0$), phase $2\pi/k$ corresponds to unidirectional edges in $H$ ($a(i,j) = 1$), and phase $ 2(k-1)\pi/k$ corresponds to node pairs of the opposite direction to the unidirectional edges in $H$ ($a(i,j) = -1$). Therefore, in the corresponding $k$-layer parametrization, $\varphi_0 = 2\pi/k$, $\mathbf{A}^0$ encodes the adjacency for bidirectional edges, denoted by $\mathbf{A}^{(\text{bi})}$, $\mathbf{A}^{\varphi_0}$ encodes the adjacency for unidirectional edges, denoted by $\mathbf{A}^{(\text{uni})}$, $\mathbf{A}^{(k-1)\varphi_0} = \mathbf{A}^{(\text{uni})T}$, and $\mathbf{A}^{z\varphi_0} = \mathbf{0}$ for $z=2,\dots,k-2$.
The adjacency matrix of the $k$-layer graph $G^{(k)}$ is then:
\begin{align*}
    \mathbf{A}^{(k)} = 
    \begin{pmatrix}
        \mathbf{A}^{(\text{bi})} & \mathbf{A}^{(\text{uni})} & \mathbf{0} & \dots & \mathbf{A}^{(\text{uni})T} \\ 
        \mathbf{A}^{(\text{uni})T} & \mathbf{A}^{(\text{bi})} & \mathbf{A}^{(\text{uni})} & \dots & \mathbf{0}\\
        \mathbf{0} & \mathbf{A}^{(\text{uni})T} & \mathbf{A}^{(\text{bi})} & \dots & \mathbf{0}\\
        \vdots & \vdots & \vdots & \ddots & \vdots \\
        \mathbf{A}^{(\text{uni})} & \mathbf{0} & \mathbf{0} & \dots & \mathbf{A}^{(\text{bi})}
    \end{pmatrix}.
\end{align*}
Hence, $\mathbf{A}^{(k)}$ is symmetric, and we can ignore the direction of edges (also from Corollary \ref{cor:ORC-k-undi}). 

As an example, suppose that the digraph $H$ is a directed cycle of length $3$ with edges $\{(v_1, v_2), (v_2, v_3), (v_3, v_1)\}$, and we consider the magnetic Laplacian with $\theta=2\pi/3$. Then $k=3$, and if we denote the adjacency matrix of $H$ as $\mathbf{A}_H$, $\mathbf{A}^{(\text{uni})} = \mathbf{A}_H$ and $\mathbf{A}^{(\text{bi})} = \mathbf{0}$. Hence, the adjacency matrix of the $3$-layer graph $G^{(3)}$ is then: 
\begin{align*}
    \mathbf{A}^{(3)} = 
    \begin{pmatrix}
        \mathbf{0} & \mathbf{A}_H & \mathbf{A}_H^T \\
        \mathbf{A}_H^T & \mathbf{0} & \mathbf{A}_H \\
        \mathbf{A}_H & \mathbf{A}_H^T & \mathbf{0}
    \end{pmatrix}
    \text{, where }
    \mathbf{A}_H = 
    \begin{pmatrix}
        0 & 1 & 0\\
        0 & 0 & 1\\
        1 & 0 & 0
    \end{pmatrix}.
\end{align*}
$G^{(3)}$ consists of $3$ disconnected cycles.
Now, if we reverse the direction of one edge and consider the directed flow graph where $E_H = \{(v_1, v_2), (v_2, v_3), (v_1, v_3)\}$, while maintaining $\theta=2\pi/3$, then the $3$-layer graph $G^{(3)}$ becomes a cycle of length $9$; see Fig.~\ref{fig:multi-dir-eg}.
The ORC of edges in the former of a directed cycle is larger than that in the latter; see Fig.~\ref{fig:multi-dir-L-k3} for details. Hence, the ORC can capture shortcuts between the in-neighbors of the start and out-neighbors of the end for each directed edge. 

\begin{figure}[htbp]
    \centering
    \begin{tabular}{cc}
        \includegraphics[width=.35\textwidth]{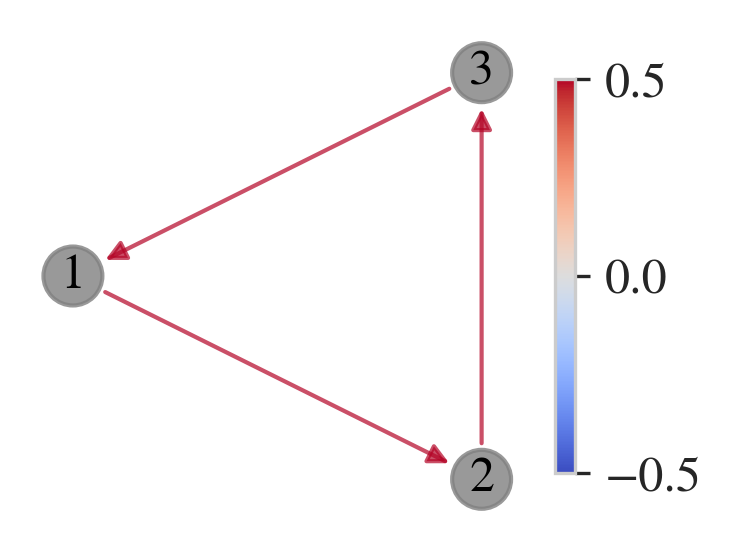} &  \includegraphics[width=.35\textwidth]{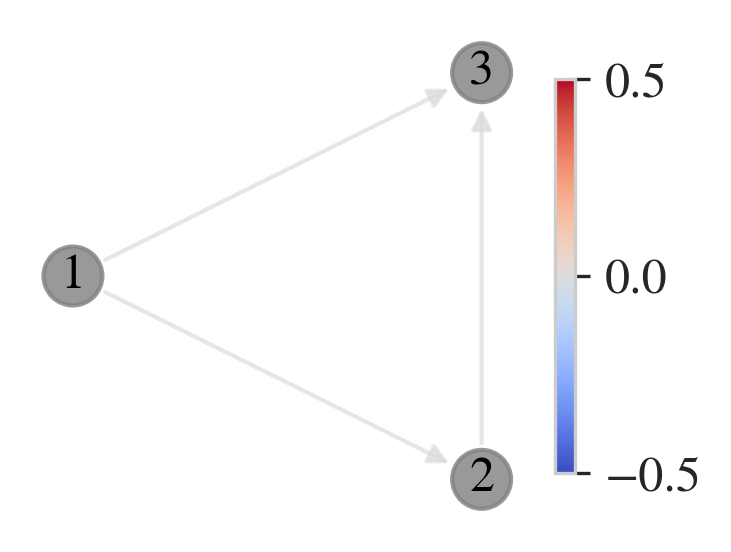}
    \end{tabular}
    \caption{ORC for the directed cycle graph (left) and the directed flow graph (right), computed from their complex-weighted graphs induced by the magnetic Laplacian with $\theta=2\pi/3$.}
    \label{fig:multi-dir-L-k3}
\end{figure}

\subsubsection{ORC for complex-weighted graphs induced from digraphs versus ORC for digraphs} 
Through the magnetic Laplacian of a digraph $H$, we can characterize the ORC of each edge from the ORC in the induced complex-weighted graph $G$. This, in turn, proposes a new notion of ORC for digraphs. How this new notion characterizes the ORC in digraphs also depends on the adequacy of the magnetic Laplacian in reflecting the underlying structural features. 
In the proposed ORC for complex-weighted graphs, we require walks with desired phases in computing the shortest-walk distance. (i) Since each directional edge in $H$ corresponds to an edge of a nonzero phase $\theta$ in $G$, only a subset of all walks connecting the vertices is eligible. However, in digraph $H$, all edges have phase $0$, thus all walks connecting the vertices satisfy the condition. (ii) The magnetic Laplacian is Hermitian, thus for all $\overrightarrow{uv}\in E_H$, $\overrightarrow{uv},\overrightarrow{vu}\in E(G)$, and then some combinations of edges from one vertex to another that are excluded in $H$ because of direction conflicts are now eligible in $G$, providing more walks connecting the same pair of vertices. 
In summary, the magnetic Laplacian transforms the directional information to phase information, while the proposed ORC can capture both. 
\begin{remark}
    Directed triangles can create shortcuts in both notions. However, while the complex-weighted ORC notion accounts for undirected triangles, this structural information is not captured by the directed ORC notion. We illustrate this below in the numerical results.
    \label{rem:triangle}
\end{remark}

\section{Theoretical Results}\label{sec:theory}
In this section, we present several results that characterize relations between the proposed ORC notion and the Magnetic Laplacian, as well as structural balance. We further present combinatorial upper and lower curvature bounds. We again present proof
sketches only in the main text and defer all details to SM.

\subsection{Magnetic Laplacian}
The magnetic Laplacian provides a natural connection between digraphs and complex-weighted graphs, both of which the proposed ORC can characterize. In this section, we provide further characterization of the ORC of digraphs via the magnetic Laplacian. Specifically, we consider the $k$-layer representation $G^{(k)}$ of the complex-weighted graph induced by the magnetic Laplacian of a digraph $H$, and connect the structures in $G^{(k)}$ with those in the digraph $H$. Recall that copies in $G^{(k)}$, of the same edge in the complex-weighted graph induced by the digraph, have the same connectivity, thus the same ORC values.

From the construction of the magnetic Laplacian, each directed edge in $H$ also becomes bidirectional in the complex-weighted graph, but different directions have different phases. Therefore, each undirected cycle in $H$ (which is a cycle if we ignore the direction of edges) can lead to two directed cycles, with opposite orientations and different phases of their compositing edges, in the complex-weighted graph. Such a pair of directed cycles leads to one or more pairs of directed cycles in the $k$-layer graph $G^{(k)}$, where cycles of different phases in the complex-weighted graph lead to cycles of different lengths in $G^{(k)}$. Therefore, the direction information in $H$ is transferred into the phase dimension in the complex-weighted graph and finally into the distance in $G^{(k)}$. 
To see how the magnetic Laplacian characterizes the digraph $H$ via $G^{(k)}$, we consider the \textit{effective length} of each cycle in $H$, as defined in Definition \ref{def:mag-eff-len}, and connect this notion with the length of cycles in $G^{(k)}$.
\begin{definition}[Effective length]
    For an undirected cycle $C$ in a digraph $H$ with edges of any direction, {fix an arbitrary traversal of the underlying undirected cycle.} The effective length of $C$ is defined as the absolute difference between the number of edges oriented in one direction and the number of edges oriented in the opposite direction {relative to this traversal, with bidirectional edges contributing $0$; the absolute value makes the traversal choice irrelevant.}
    \label{def:mag-eff-len}
\end{definition}
\begin{proposition}
    For an undirected cycle $C$ in a digraph $H$ with edges of any direction, suppose $C$ has the actual length $|C|$ in $H$ and the effective length $|C|_e$, then the length of the corresponding cycle(s) in $G^{(k)}$ is 
    \begin{align*}
        |C|^{(k)} = 
        \begin{cases}
            |C|, &\quad \text{if } \quad  (|C|_e \hspace{-.5em}\mod k) = 0\\
            \frac{lcm((|C|_e \hspace{-.5em}\mod k), \hspace{.5em}k)}{(|C|_e \hspace{-.5em}\mod k)} |C|, &\quad \text{otherwise }
        \end{cases}
    \end{align*}
    where $lcm(\cdot,\cdot)$ returns the least common multiple of the two elements. 
    \label{pro:effect-cycle-len}
\end{proposition}
\begin{proof}[Proof sketch]
    This can be shown from the construction of the magnetic Laplacian and the multilayer parametrization. 
\end{proof}

As an example, we consider two triangles: one is a directed cycle and the other is a directed flow graph as in the bottom of Figure \ref{fig:multi-dir-L-k3}, thus $|C| = 3$; we set $k = 3$. For the former, $|C|_e = 3$, hence $(|C|_e \hspace{-.5em}\mod k) = 0$ and $|C|^{(k)} = 3$. While for the latter, $|C|_e = 1$, hence $(|C|_e \hspace{-.5em}\mod k) = 1$, then $lcm((|C|_e \hspace{-.5em}\mod k), \hspace{.5em}k) = 3$ and finally $|C|^{(k)} = {9}$, as in Figure \ref{fig:multi-dir-L-k3}.
Therefore, when some edges do not follow the same orientation in an undirected cycle, this will change the effective length of the cycle in $H$, and may cause the length of the corresponding cycle(s) in $G^{(k)}$ to change (if $(|C|_e \hspace{-.5em}\mod k)$ and $lcm((|C|_e \hspace{-.5em}\mod k), \hspace{.5em}k)$ changes differently), which can cause changes in ORC. 

\begin{proposition}
    For $k > 1$, for every undirected cycle $C$ in $H$, the length of the corresponding cycles in $G^{(k)}$ is also $|C|$, if and only if the corresponding complex-weighted graph (induced by the magnetic Laplacian with arguments $2\pi/k, -2\pi/k$) is structurally balanced. 
    \label{pro:mag-k-bal}
\end{proposition}
\begin{proof}[Proof sketch]
    This follows from Proposition \ref{pro:effect-cycle-len} and the definition of structural balance. 
\end{proof}

\begin{remark}
    If all cycles in $G^{(k)}$ maintain the same lengths as cycles in $H$, then the ORC of edges in $G^{(k)}$ retrieves the ORC of edges in $\bar{H}$ where the edge direction is ignored. Indeed, for the ORC of edges in $G^{(k)}$ to have the same values as the ORC of edges in $\bar{H}$, it is sufficient to have cycles of actual lengths $3,4$ and $5$ in $\bar{H}$ maintain the same lengths in $G^{(k)}$ (in the unweighted case).
\end{remark}

The choice of $k=3$ is of particular interest, because (i) the induced complex-weighted graphs can have any multiples of $2\pi/3$ as the edge phase, i.e., the complex weights lie in the whole cyclic group $S_3^1 := \{\xi^j|j = 0, 1, 2\}$ where $\xi = e^{2\pi\iu/3}$ is the primitive third root of unity, and (ii) the triangle structure is of particular importance in characterizing the curvature. Therefore, we will choose $\theta = 2\pi/3$ in constructing the magnetic Laplacian for digraphs in the following experiments. We can choose more customized values for $\theta$ if additional characteristics of the graph are known. For example, in the case of bipartite graphs, we may want to choose $\theta = 2\pi/4$.

\subsection{Structural balance}
The notion of \emph{structural balance} has been developed for complex-weighted graphs with Hermitian adjacency matrices \cite{Lange2015MagL,tian2024structural} (see \cite{harary_1953_balance,zaslavsky_1982_signed} for the original definition in signed networks). 
Here, we extend the notion to general directed complex-weighted networks in a Hermitian or directed manner, as in Definitions \ref{def:balance-herm} and \ref{def:balance-direct}, respectively. 
\begin{definition}[Hermitian structural balance]
    A complex-weighted digraph $G$ is Hermitian structural balanced if (i) for all $\overrightarrow{uv}\in E$, if $\overrightarrow{vu}\in E$, then $(\varphi_{uv} + \varphi_{vu}) \mod 2\pi = 0$, and (ii) for all $\overrightarrow{uv}\in E$, $\overrightarrow{vu}\notin E$, we add a dummy edge $\overrightarrow{vu}$ with phase $(2\pi - \varphi_{uv}) \mod 2\pi$, then all directed cycles, composed of edges and dummy edges, have phase $0$. 
    \label{def:balance-herm}
\end{definition}
\begin{definition}[Directed structural balance]
    A complex-weighted digraph $G$ is directed structural balanced if (i) all directed cycles have phase $0$ and (ii) all walks from $u$ to $v$ have the same phase for all $u,v\in V$. 
    \label{def:balance-direct}
\end{definition}

The Hermitian structural balance imposes stricter conditions on the directed complex-weighted graphs than the directed structural balance. The Hermitian one is closely related to the structural balance for Hermitian complex-weighted graphs, where one can obtain the complex weight matrix of the former by setting some elements in the complex weight matrix of the latter to zero. Hence, many desired properties of the structural balance for Hermitian complex-weighted graphs, as shown in \cite{tian2024structural}, still hold in the case of Hermitian structural balance. 

The directed structural balance emphasizes the direction information, and only requires well-defined walks to satisfy the phase condition. Although imposing a weaker condition on the phases, the directed structural balance provides a sufficient condition for the ORC in complex-weighted graphs to be equivalent to the ORC in the corresponding digraphs by ignoring the phase information. 
\begin{lemma}
    The ORC for complex weights gives the same value as the ORC for digraphs by setting all phases to be zero if the complex-weighted graph satisfies the directed structural balance. 
\end{lemma} 
\begin{proof}[Proof sketch]
    We can show it through the definitions of ORC for complex weights in Definition \ref{def:orc-complex} and of directed structural balance in Definition \ref{def:balance-direct}.
\end{proof}

\subsection{Combinatorial bounds}
Jost and Liu~\citep{jurgen_jost_olliviers_2013} bound Ollivier-Ricci curvature on unweighted undirected graphs by constructing explicit transport plans between the two one-step neighborhood measures. The key observation is that local cycles create shortcuts: common neighbors can be matched at zero cost and paths along triangles reduce the cost of moving the remaining mass. We adapt this perspective to the directed and complex-weighted ORC notions introduced above. Throughout this section, we restrict attention to local structures induced by cycles of actual length $3$. In the directed case, these are the cyclic triangles of effective length $3$ and the acyclic triangles of effective length $1$.

\subsubsection{Combinatorial curvature bounds on directed graphs}

We first consider the unweighted directed setting of Section~\ref{sec:curvature}. Fix a directed edge $\overrightarrow{uv}\in E$. The curvature is computed by transporting the uniform measure on the in-neighbors $X$ of $u$ to the uniform measure on the out-neighbors $Y$ of $v$, and since $d_G(u,v)=1$, we have
\[
    \kappa(u,v)=1-W_1(m_{in,u},m_{out,v}).
\]
In the following, we assume  that $a:=d^{in}(u)>0$ and $b:=d^{out}(v)>0$.
The cyclic directed triangles (effective length $3$) around $\overrightarrow{uv}$ are formed with the vertices
$C_3(u,v)=X \cap Y$; let $c=|C_3(u,v)|$. There are two types of acyclic triangle, each with effective length $1$, formed by
\[
    A^-(u,v)=\{x\in X\setminus C_3(u,v):\overrightarrow{xv}\in E\}
\]
and
\[
    A^+(u,v)=\{y\in Y\setminus C_3(u,v):\overrightarrow{uy}\in E\}.
\]
Let $p=|A^-(u,v)|$ and $q=|A^+(u,v)|$.

\begin{theorem}[Directed three-cycle curvature bounds]
    Let $Z=\frac{c}{a\vee b}$ and $M=\min\left\{1-Z,\frac{p}{a}+\frac{q}{b}\right\}$,
    where $c,p,q$ are the local three-cycle counts defined above and $a\vee b := \max\{a,b\}$. Then
    \[
        -2+3Z+M \leq \kappa(u,v)\leq Z.
    \]
    \label{thm:directed-combinatorial-bound}
\end{theorem}

\begin{proof}[Proof sketch]
    The proof is an adaptation of the argument in Jost-Liu~\cite{jurgen_jost_olliviers_2013}. The upper bound follows from the same zero-cost overlap argument, the lower bound from an explicit transport plan that leverages the cycle-induced shortcuts. The shortcut mass is coupled through the source mass on $A^-(u,v)$ and the target mass on $A^+(u,v)$. We present a detailed proof in Apx.~\ref{apx.proofs}.
\end{proof}

\begin{remark} We discuss two special cases in which the combinatorial bounds are attained.
    \begin{enumerate}
        \item Consider an isolated cyclic triangle. Then $a=b=c=1, p=q=0$ for every edge $\overrightarrow{uv}$, which implies $Z=1$ and $M=0$ and hence $1\leq \kappa(u,v)\leq 1$. Indeed, the two neighborhood measures are supported on the same neighborhood, so $W_1=0$ and $\kappa(u,v)=1$.
    \item Consider a directed edge $\overrightarrow{uv}$ without directed shortcuts, i.e., $c=p=q=0$. Then $Z=M=0$ and the lower bound gives $\kappa(u,v)\geq -2$. If, in addition, every residual source-target pair has shortest directed distance exactly $3$, moving every unit of mass costs exactly $3$, so $W_1=3$ and $\kappa(u,v)=-2$; i.e., the lower bound is attained. 
    \end{enumerate}
\end{remark}

\subsubsection{Combinatorial curvature bounds on complex-weighted graphs}

We now derive analogous bounds for complex-weighted graphs. The main difference is that a triangle contributes only when the shortcut has the same accumulated phase as the reference walk through $\overrightarrow{uv}$. Moreover, the overlap and shortcut terms are weighted by the neighborhood measures induced by the magnitudes of the complex edge weights.

Fix $\overrightarrow{uv}\in E$ with $W_{uv}=r_{uv}e^{\iu\varphi_{uv}}$ and $r_{uv}>0$, and set $\delta_{uv}:=d_C^{\varphi_{uv}}(u,v)$. Let again
$X,Y$ denote the neighborhoods of $u,v$, and write $d_{\text{in}}(u)=\sum_{x\in X}r_{xu}$ and $d_{\text{out}}(v)=\sum_{y\in Y}r_{vy}$; we assume again that
$d_{\text{in}}(u),d_{\text{out}}(v)>0$. The neighborhood measures are given by
\[
    m_{\text{in},u}(x)=\frac{r_{xu}}{d_{\text{in}}(u)},\qquad x\in X,\qquad
    m_{\text{out},v}(y)=\frac{r_{vy}}{d_{\text{out}}(v)},\qquad y\in Y.
\]

The phase-compatible cyclic overlap is
\[
    C_0(u,v)=
    \left\{
    z\in X\cap Y:
    \varphi_{zu}+\varphi_{uv}+\varphi_{vz}\equiv 0 \mod 2\pi
    \right\}.
\]
Only vertices in $C_0(u,v)$ can be matched to themselves at zero cost under the convention for $d_C^\varphi(z,z)$. Define the weighted overlap
\[
    Z=\sum_{z\in C_0(u,v)}\min\{m_{\text{in},u}(z),m_{\text{out},v}(z)\}.
\]
After matching this mass at zero cost, let $\widetilde{m}_{\text{in},u}$ and $\widetilde{m}_{\text{out},v}$ denote the residual source and target measures.

The phase-compatible acyclic shortcut sets are
\[
    A^-(u,v)=
    \left\{
    x\in X\setminus C_0(u,v):\overrightarrow{xv}\in E,\ 
    \varphi_{xv}\equiv \varphi_{xu}+\varphi_{uv}\mod 2\pi
    \right\}
\]
and
\[
    A^+(u,v)=
    \left\{
    y\in Y\setminus C_0(u,v):\overrightarrow{uy}\in E,\ 
    \varphi_{uy}\equiv \varphi_{uv}+\varphi_{vy}\mod 2\pi
    \right\}.
\]
Set
\[
    P=\sum_{x\in A^-(u,v)}\widetilde{m}_{\text{in},u}(x),\qquad
    Q=\sum_{y\in A^+(u,v)}\widetilde{m}_{\text{out},v}(y),\qquad
    M=\min\{1-Z,P+Q\}.
\]
Finally, define the baseline cost
\[
    L_3=\max_{x\in X,\ y\in Y}(r_{xu}+r_{uv}+r_{vy}).
\]
For pairs with $x\in A^-(u,v)$ or $y\in A^+(u,v)$, define
\[
    \ell_2(x,y)=
    \begin{cases}
        \min\{r_{xv}+r_{vy},r_{xu}+r_{uy}\}, & x\in A^-(u,v),\ y\in A^+(u,v),\\
        r_{xv}+r_{vy}, & x\in A^-(u,v),\ y\notin A^+(u,v),\\
        r_{xu}+r_{uy}, & x\notin A^-(u,v),\ y\in A^+(u,v),
    \end{cases}
\]
and let $L_2$ be the maximum of $\ell_2(x,y)$ over such pairs. If no such pair exists, the term involving $L_2$ is ignored. We also put $\overline{L}_2=\min\{L_2,L_3\}$ and
\[
    \ell_+=
    \min\left\{
    d_C^{\theta(xuvy)}(x,y):
    x\in X,\ y\in Y,\ d_C^{\theta(xuvy)}(x,y)>0
    \right\},
\]
with the convention that the term $\ell_+(1-Z)$ is $0$ if $Z=1$.

\begin{theorem}[Complex-weighted three-cycle curvature bounds]
    With the notation above,
    \[
        1-\frac{M\overline{L}_2+(1-Z-M)L_3}{\delta_{uv}}
        \leq
        \kappa(u,v)
        \leq
        1-\frac{\ell_+(1-Z)}{\delta_{uv}}.
\]
    In particular, if $\ell_+\geq \delta_{uv}$, then $\kappa(u,v)\leq Z$.
    \label{thm:complex-combinatorial-bound}
\end{theorem}

\begin{proof}[Proof sketch]
    The proof is a generalization of the arguments in the directed case and can be found in Apx.~\ref{apx.proofs}.
\end{proof}

\subsubsection{Combinatorial curvature approximation}
Finally, we leverage the derived bounds to define scalable combinatorial curvature approximations. We follow the approximation scheme proposed in~\citep{tian2025curvature}, which leverages the arithmetic mean of the upper and lower bounds.

\begin{definition}[Combinatorial curvature approximation]\label{def:comb-approx}
    Suppose that, for an edge $\overrightarrow{uv}$, we have computable bounds $\kappa_L(u,v)\leq \kappa(u,v)\leq \kappa_U(u,v)$.
    Then we define the combinatorial curvature approximation as
    \[
        \widehat{\kappa}_{\mathrm{comb}}(u,v)
        =
        \frac{1}{2}\left(\kappa_L(u,v)+\kappa_U(u,v)\right).
    \]
\end{definition}

In the directed setting, Theorem~\ref{thm:directed-combinatorial-bound} provides upper and lower bounds of the form
\[
    \kappa^{dir}_{L}(u,v)=-2+3Z+M,
    \qquad
    \kappa^{dir}_{U}(u,v)=Z \; ,
\]
i.e., $\widehat{\kappa}^{dir}(u,v)
    =-1+2Z+\frac{M}{2}$.
This quantity depends only on the in-degree of $u$, the out-degree of $v$, and the local effective-length of triangles that include $\overrightarrow{uv}$.

In the complex-weighted case, Theorem~\ref{thm:complex-combinatorial-bound} gives the bounds
\[
    \kappa^{cw}_{L}(u,v)
    =
    1-\frac{M\overline{L}_2+(1-Z-M)L_3}{\delta_{uv}},
    \qquad
    \kappa^{cw}_{U}(u,v)
    =
    1-\frac{\ell_+(1-Z)}{\delta_{uv}} \; .
\]
The resulting notion, albeit a more complicated one than in the directed case, reduces the computational cost of the curvature computation, as we see below.
When the positive phase-compatible distances satisfy $\ell_+\ge \delta_{uv}$, Theorem~\ref{thm:complex-combinatorial-bound} also gives the local upper bound $\kappa^{cw}_{U,loc}(u,v)=Z$. In that case, one may use the fully local approximation
\[
    \widehat{\kappa}^{cw}_{loc}(u,v)
    =
    \frac{1}{2}\left(\kappa^{cw}_{L}(u,v)+Z\right)
\]
as a more scalable proxy.

\begin{remark}[Computational complexity]
    The exact curvature computation requires solving an optimal transport problem together with the relevant directed or phase-constrained shortest-walk distances for every edge. The directed approximation $\widehat{\kappa}^{dir}$ replaces this by local adjacency queries for $C_3(u,v)$, $A^-(u,v)$, and $A^+(u,v)$, reducing complexity from cubic to linear in $(d_{\mathrm{in}}(u)+d_{\mathrm{out}}(v))$. 
    The same applies to the local complex-weighted approximation \(\widehat{\kappa}^{cw}_{loc}\). The nonlocal version is not generally linear, since computing \(\ell_+\) may require phase-constrained shortest-walk computations.
\end{remark}

\begin{figure}[tb]
    \centering
    \begin{tabular}{cc}
        \includegraphics[width=.35\textwidth]{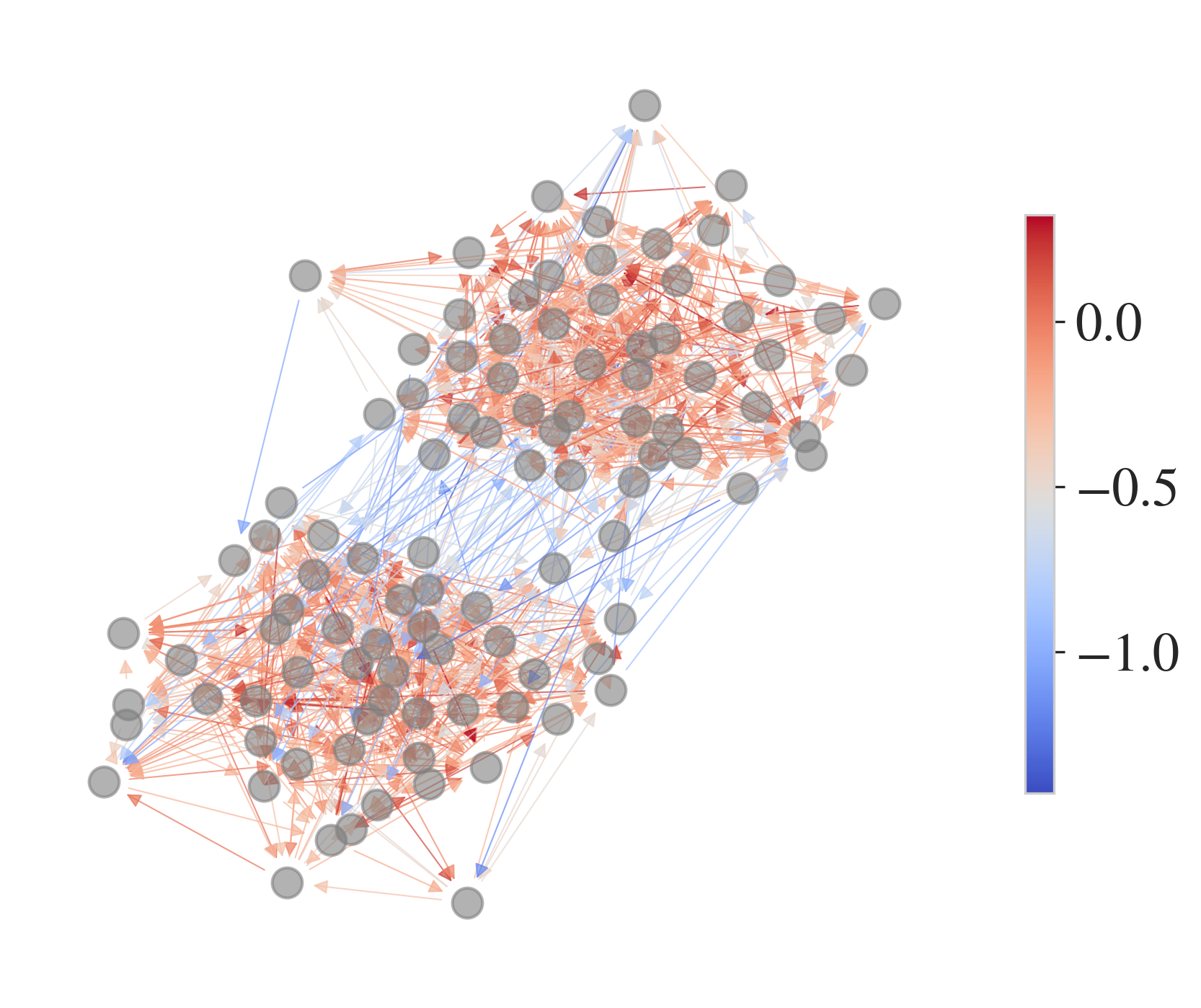} & \includegraphics[width=.35\textwidth]{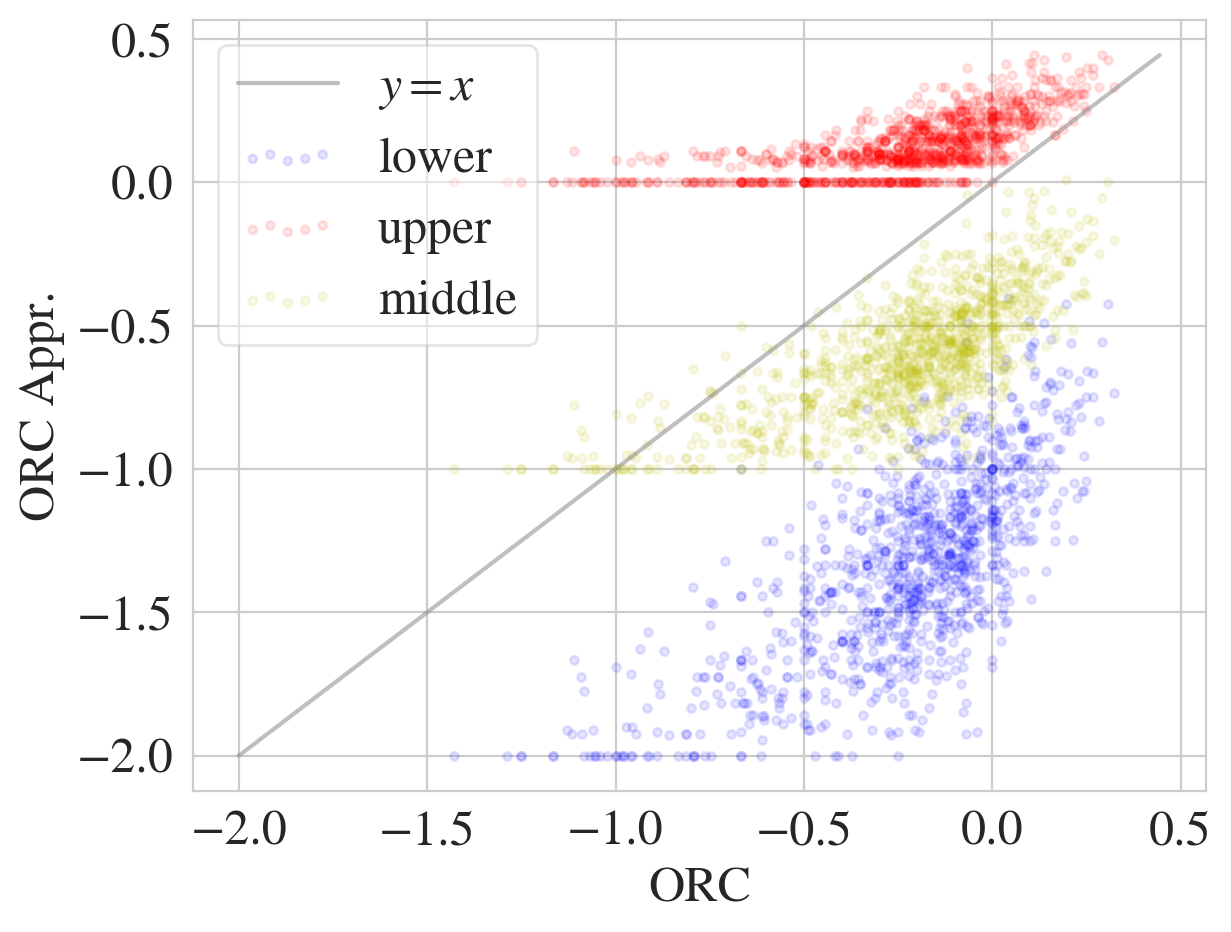} \\
        \includegraphics[width=.35\textwidth]{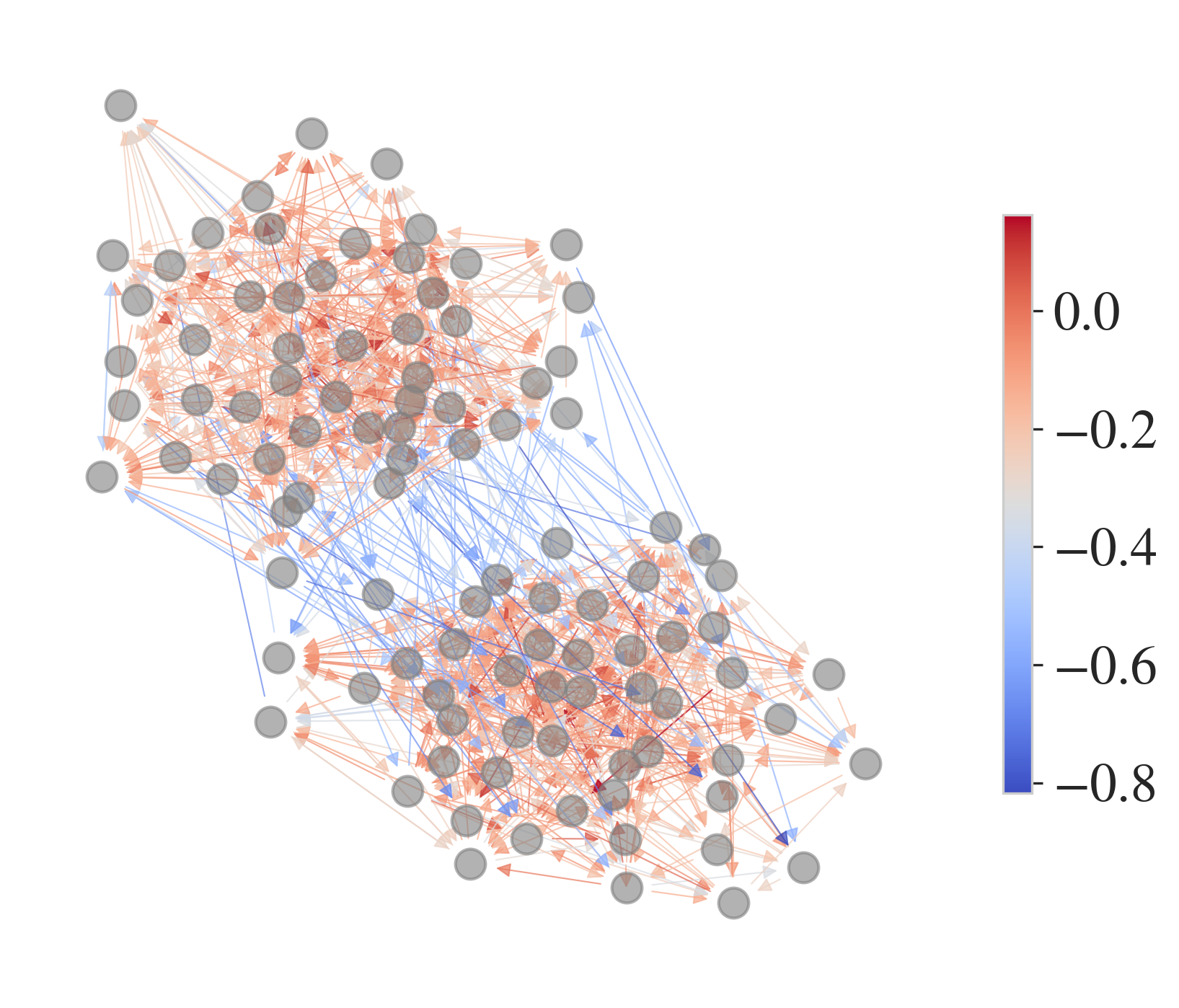} & \includegraphics[width=.35\textwidth]{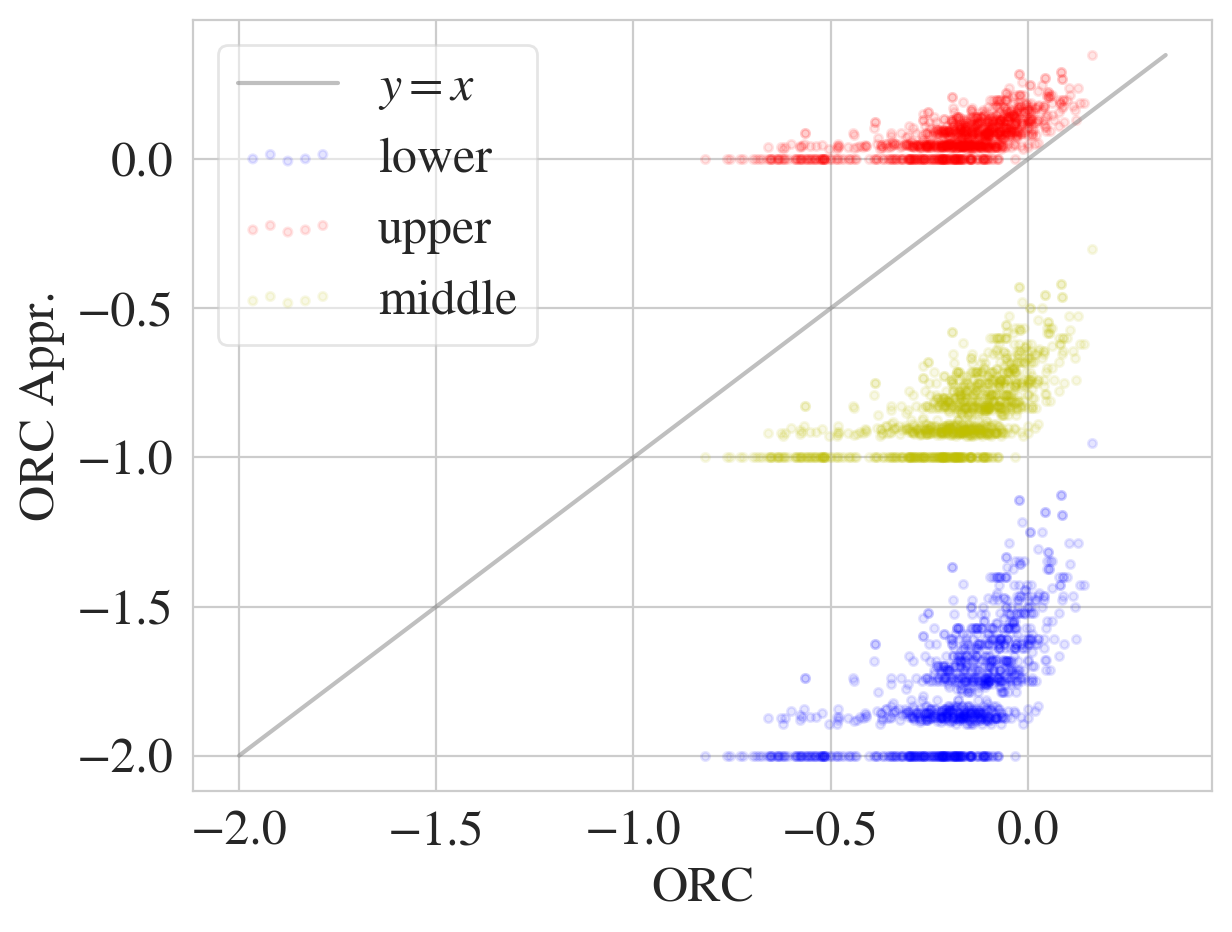}
    \end{tabular}
    \caption{Scatter plot of the exact directed (top) and complex-weighted (bottom) ORC and the approximation (green; upper bound in red and lower bound in blue), where the graph $G$ is obtained from the planted directed SBM of size $n = 100$, two equally-sized blocks, $p_{\mathrm{out}} = 0.02$ and $p_{\mathrm{in}} = 0.2$, with uniform weights $1$, and the complex weights are obtained from the magnetic Laplacian with parameter $\theta=2\pi/3$. $G$ is visualised in the left column, with the color indicating the exact ORC.}
    \label{fig:bounds-sbm}
\end{figure}

\section{Application in Community Detection}\label{sec:applications}

\subsection{Curvature gaps}

    



Following~\citep{gosztolai2021unfolding}, we
define the \emph{curvature gap} of a two-block model as the separation between the distributions of the curvature values on within-community edges and on between-community edges.
Formally, assume that the empirical curvature distribution can be modeled as a two-component Gaussian mixture,
$\left(\mathcal N(\mu_1,\sigma_1^2),\mathcal N(\mu_2,\sigma_2^2)\right)$. We measure the curvature gap as the Hellinger distance between these distributions:
\begin{equation}
\Delta \kappa = \sqrt{1-\sqrt{\frac{2\sigma_1\sigma_2}{\sigma_1^2+\sigma_2^2}}\
e^{-\frac{(\mu_1-\mu_2)^2}{4(\sigma_1^2+\sigma_2^2)}}}.
\end{equation}
The larger this gap, the more reliably an edge's curvature identifies it as internal or inter-community, which is directly linked to the effectiveness of curvature-based community detection; see Figures \ref{fig:sbm-dirtri} and \ref{fig:sbm-dirtri-reci} for examples. 

The behavior of the gap across the undirected, directed and complex-weighted curvatures is governed by the role of triangles (Proposition~\ref{pro:effect-cycle-len}). The difference in the resulting curvature gaps can be partially understood from the distinct ways in which the three notions incorporate triangle information: directed ORC only takes into account triangles whose orientations are compatible with the flow imposed by the edge under consideration, while the complex-weighted ORC for the magnetic Laplacian, triangles with all orientations will make some contributions, although one may take a more significant role than others; see Remark \ref{rem:triangle}.

\begin{figure}[htbp]
    \centering
    \hspace*{-1em}
    \begin{tabular}{ccc}
        \includegraphics[width=.32\textwidth]{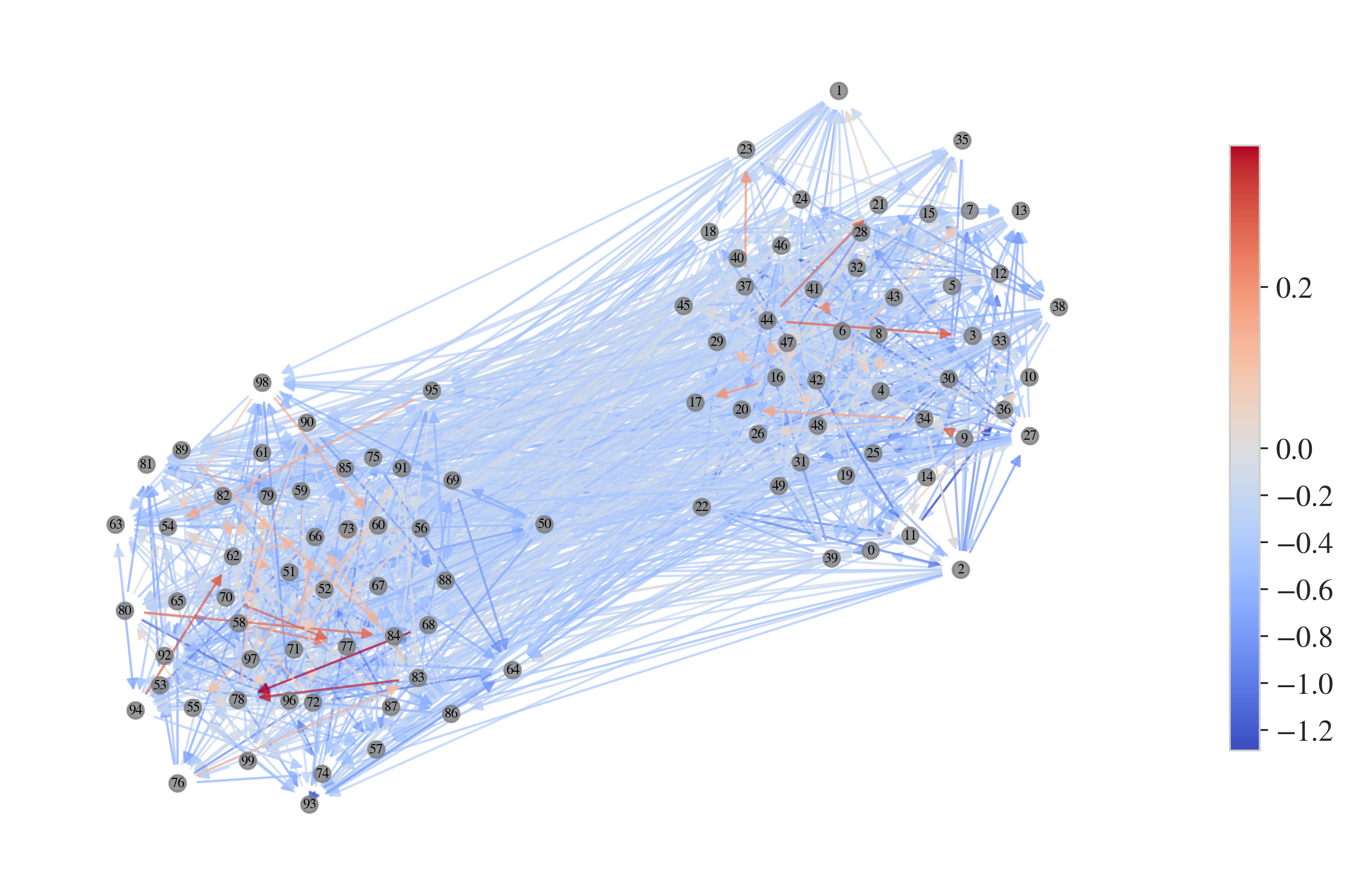} & \includegraphics[width=.32\textwidth]{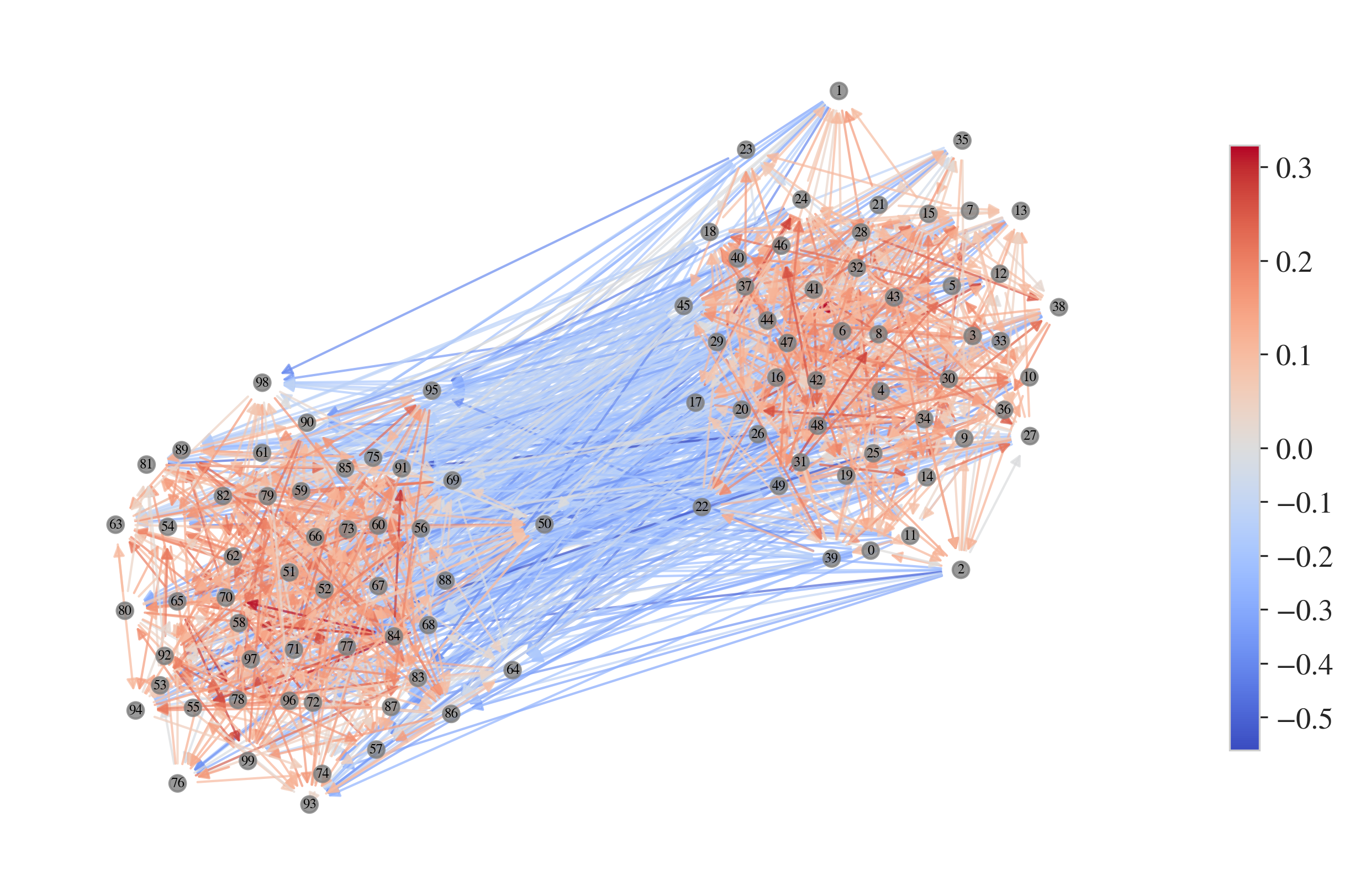} & \includegraphics[width=.32\textwidth]{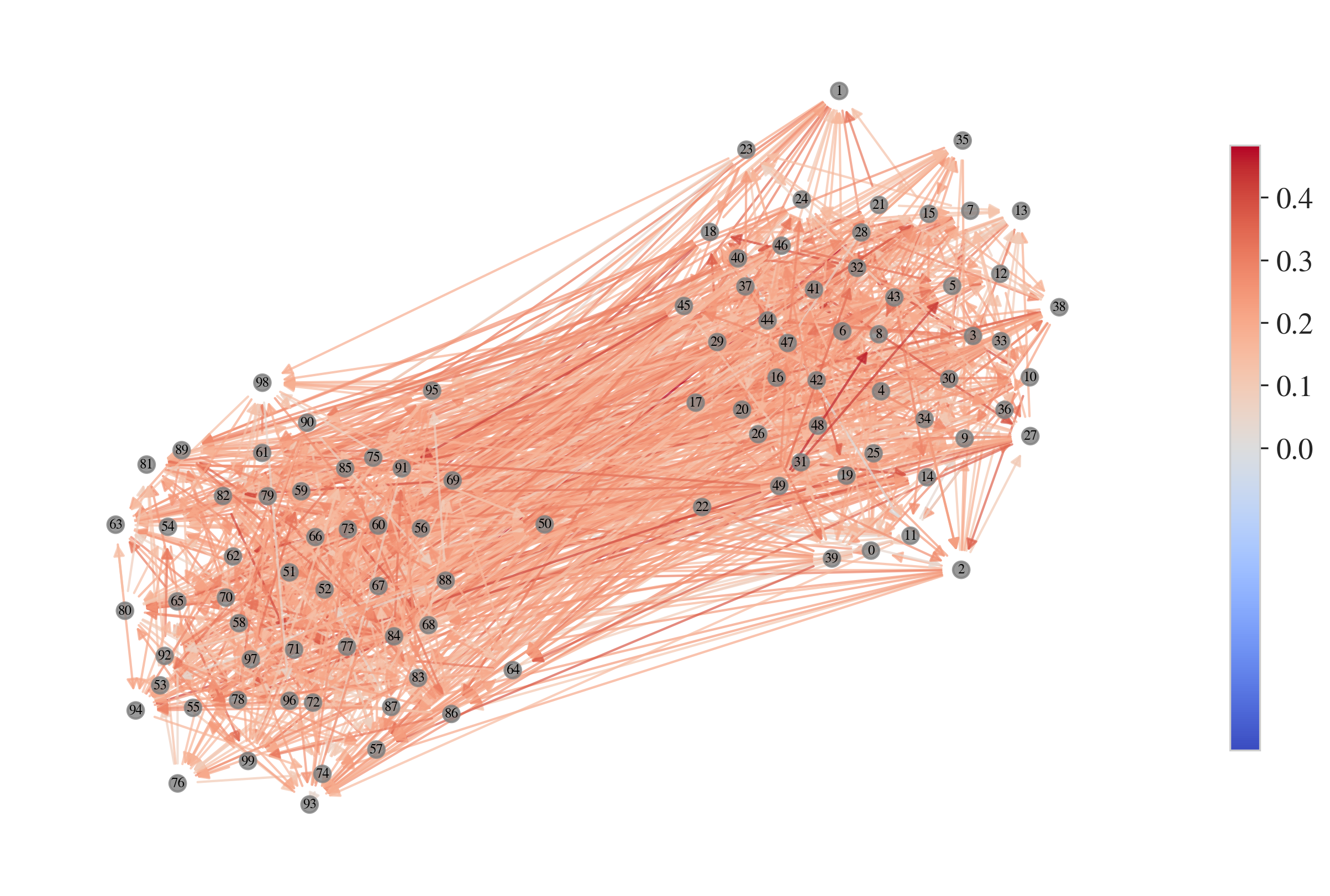} \\
        \includegraphics[width=.3\textwidth]{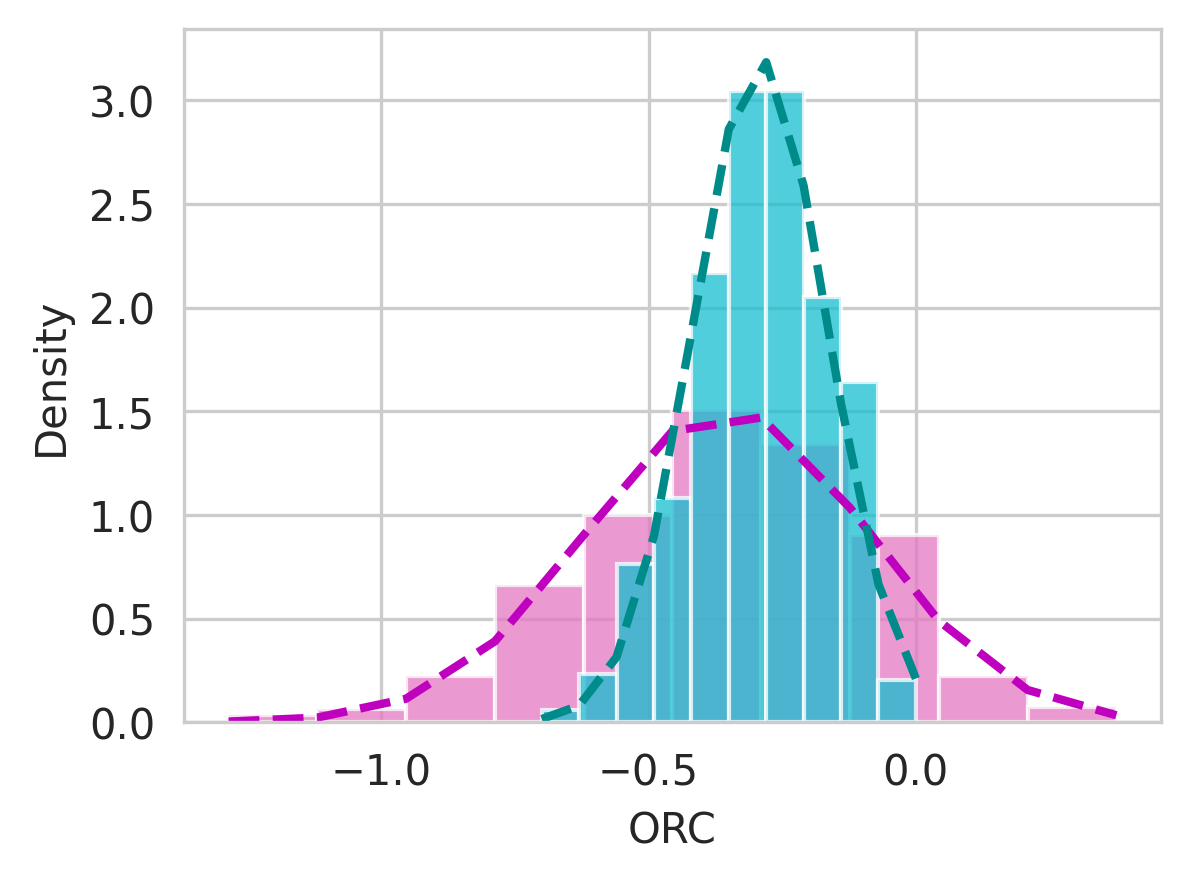} & \includegraphics[width=.3\textwidth]{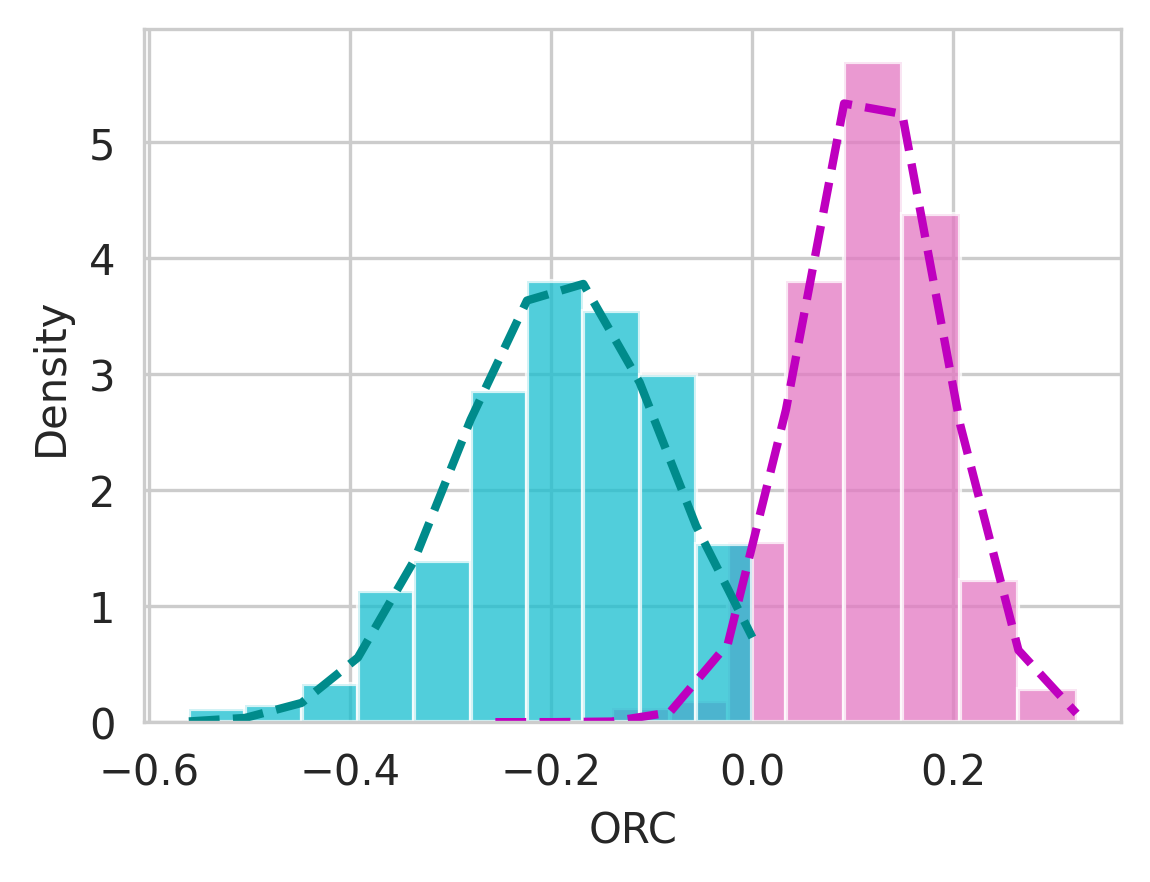} & \includegraphics[width=.3\textwidth]{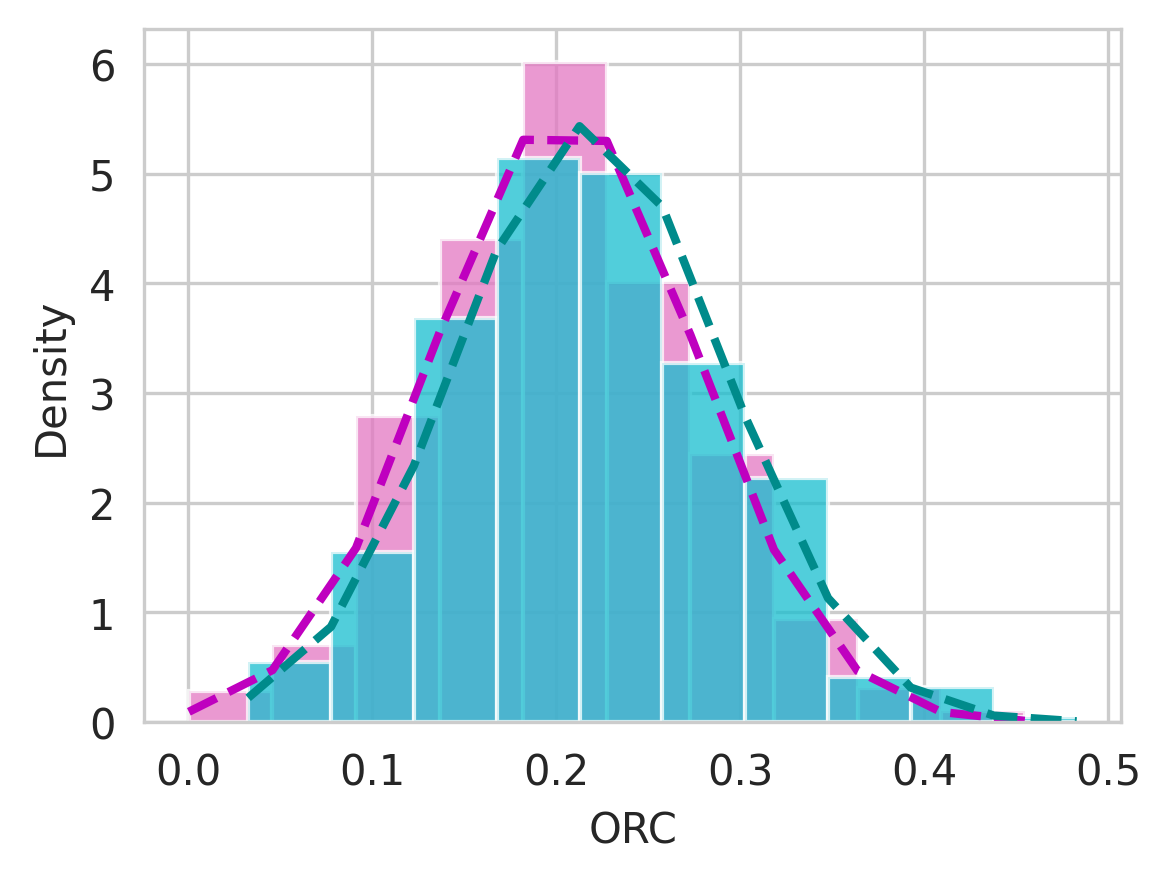}
    \end{tabular}
    \caption{(Top) Visualization of the two-block directed SBMs with size $100$, $p_{\text{in}} = 0.33$ and $p_{\text{out}} = 0.2$, where edges are colored by the values of directed ORC (left), complex-weighted ORC for the magnetic Laplacian with $\theta = 2\pi/3$ (middle) and undirected ORC (right); (Bottom) Distributions of the corresponding ORC values, where within-community edges (pink) are directed (to have as many directed triangles as possible) and between-community edges (cyan) are from community 1 to 2. Curvature gaps measured by Hellinger distance are $0.130, 0.791, 0.04$ from left to right.}
    \label{fig:sbm-dirtri}
\end{figure}

\begin{figure}[htbp]
    \centering
    \hspace*{-1em}
    \begin{tabular}{ccc}
        \includegraphics[width=.32\textwidth]{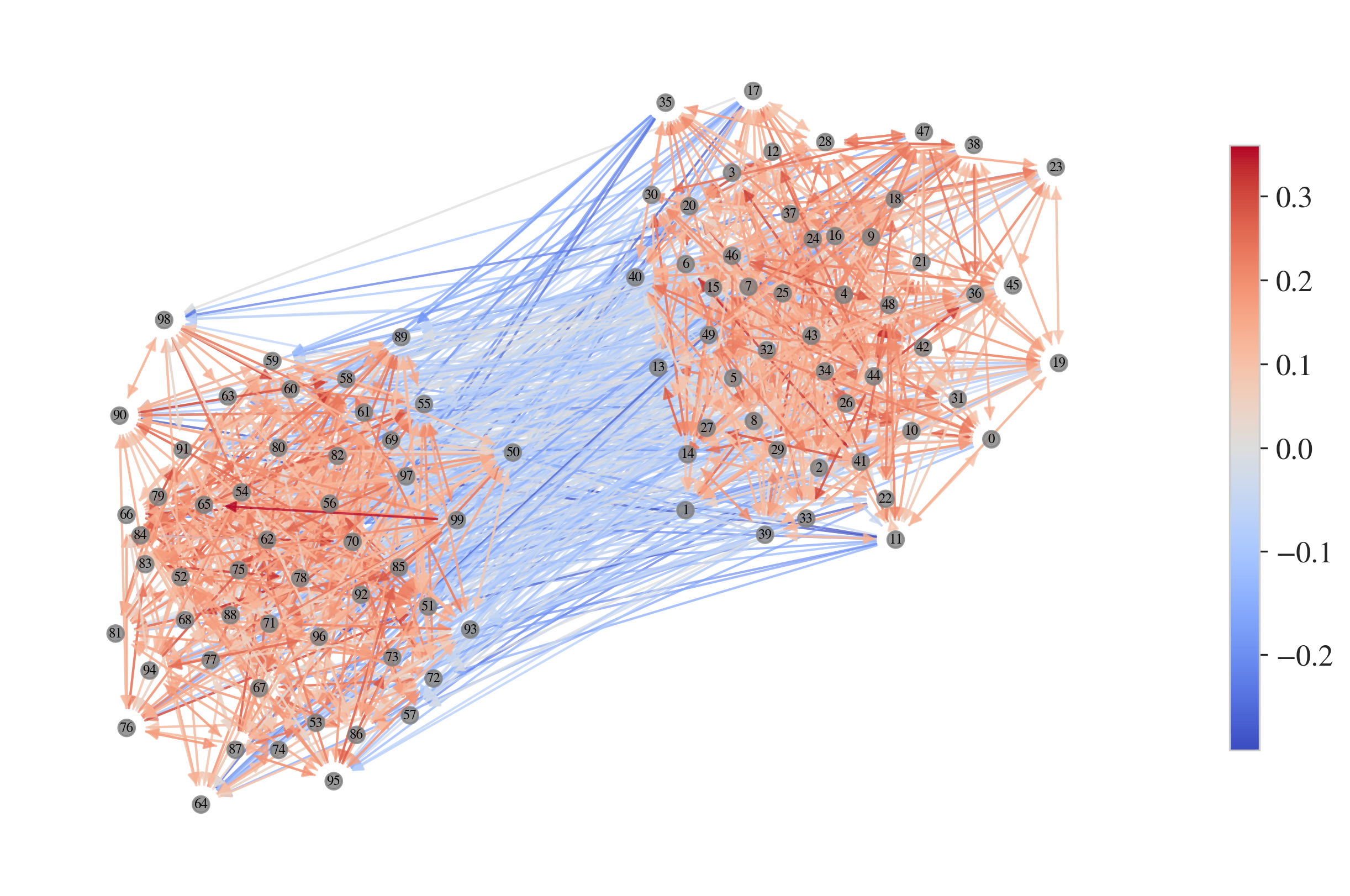} & \includegraphics[width=.32\textwidth]{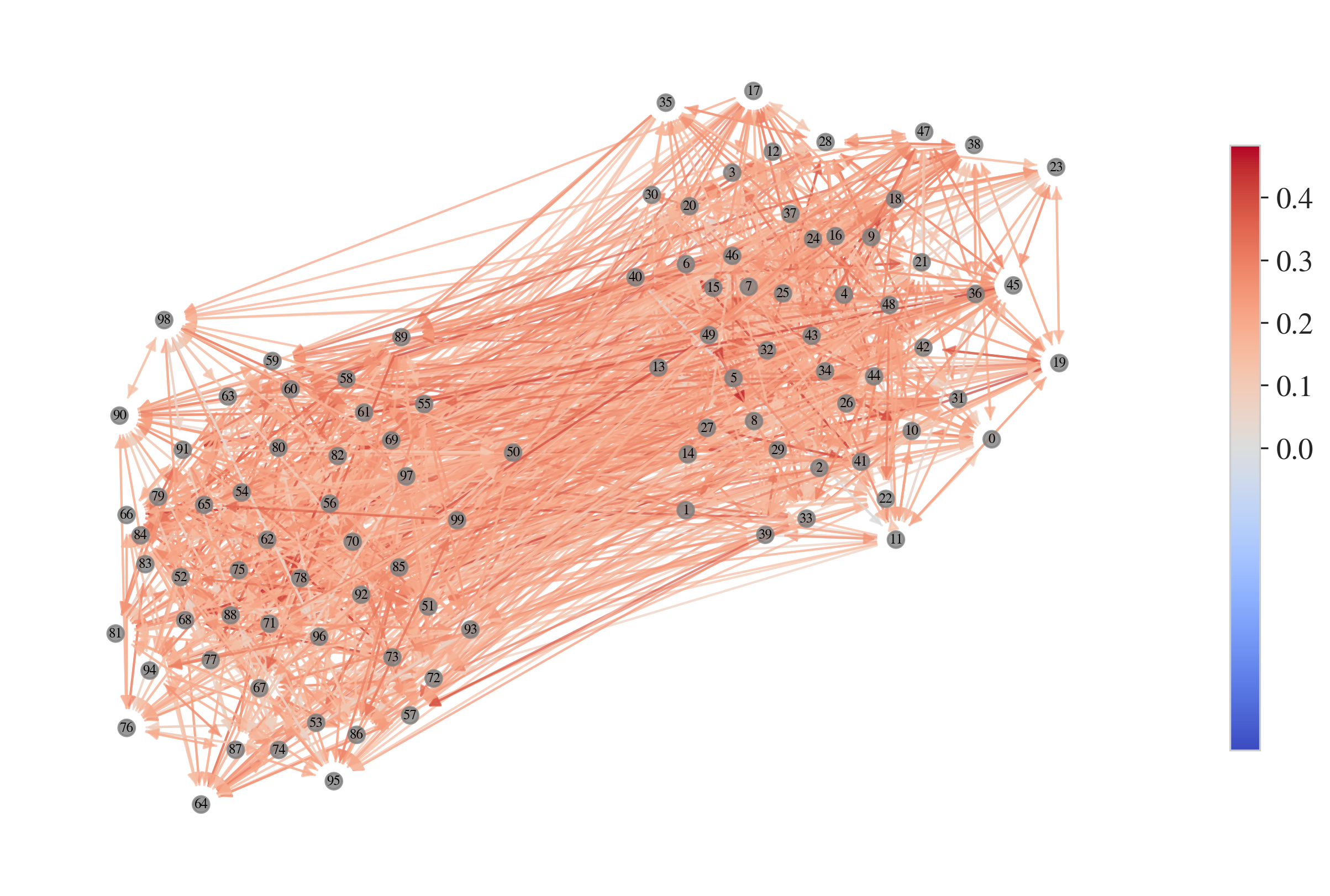} & \includegraphics[width=.32\textwidth]{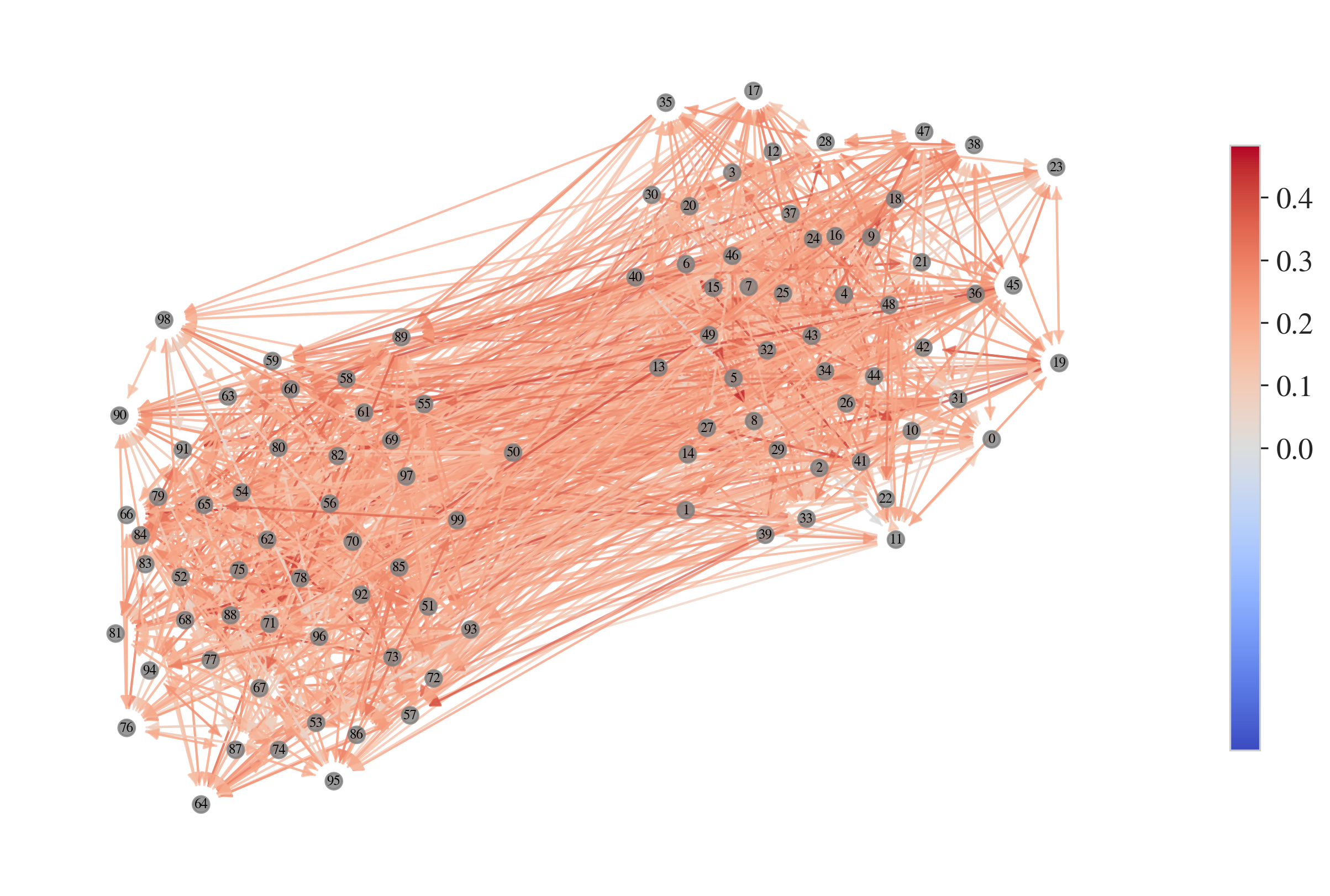} \\
        \includegraphics[width=.3\textwidth]{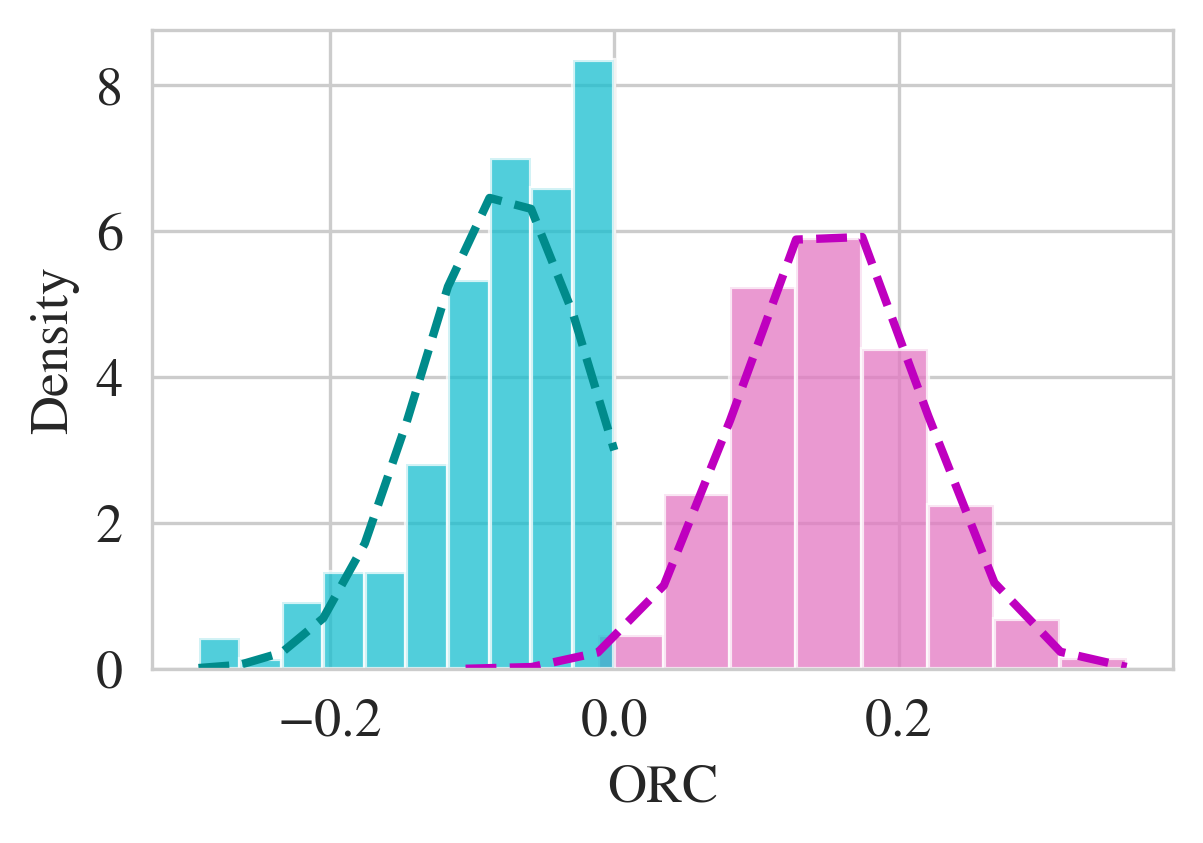} & \includegraphics[width=.3\textwidth]{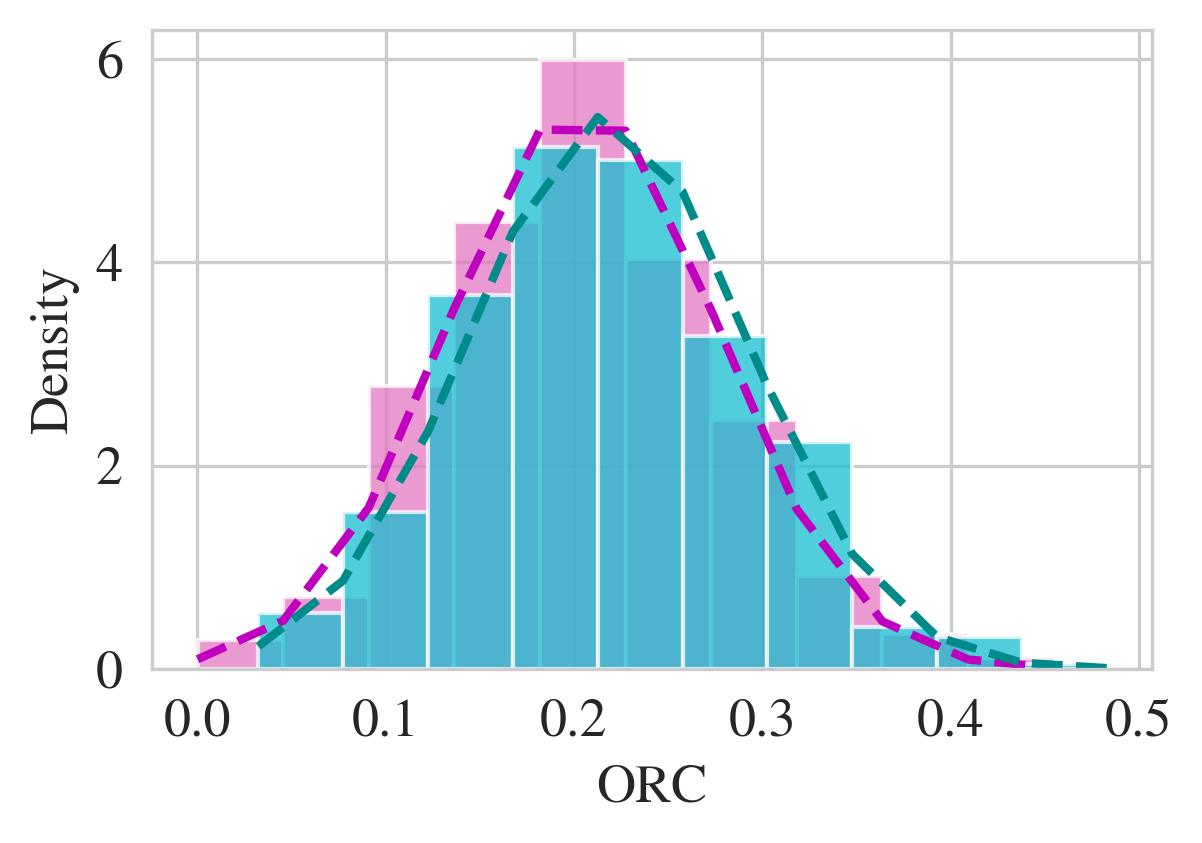} & \includegraphics[width=.3\textwidth]{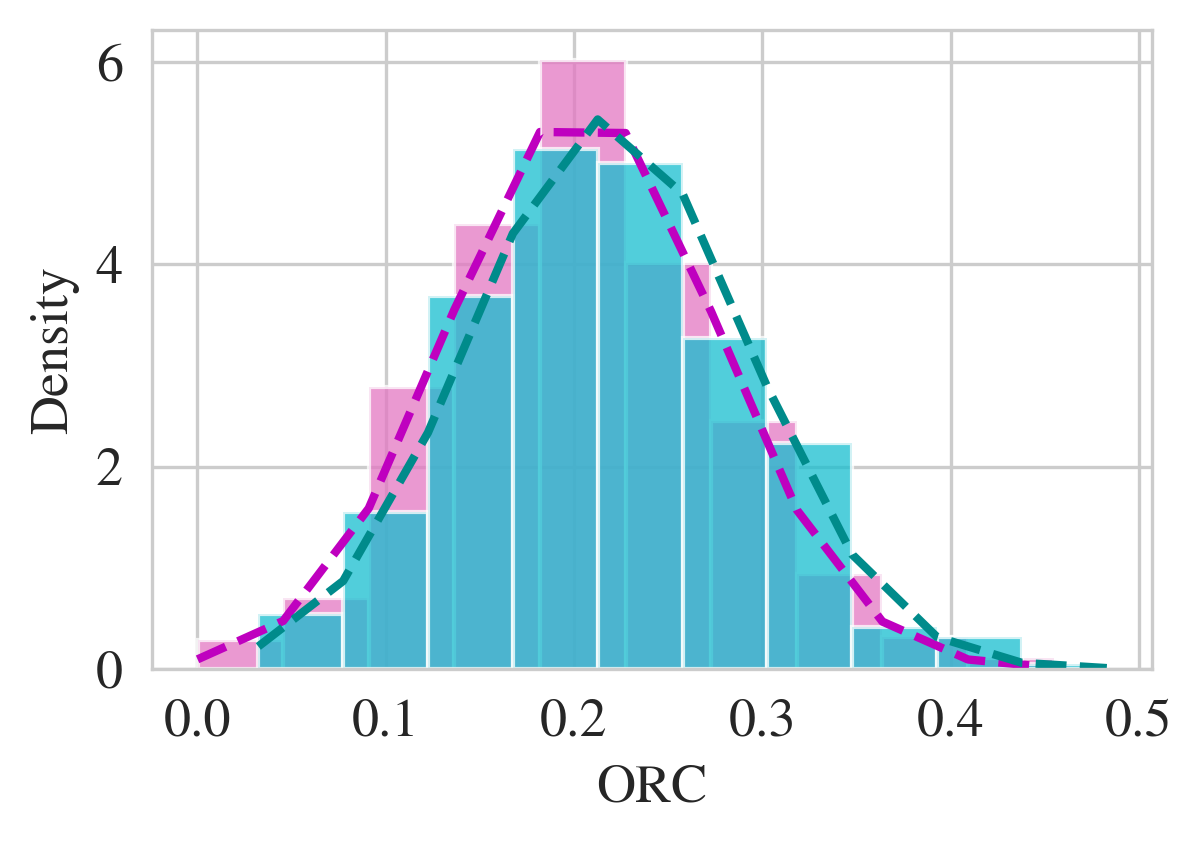}
    \end{tabular}
    \caption{(Top) Visualization of the two-block directed SBMs with size $100$, $p_{\text{in}} = 0.33$ and $p_{\text{out}} = 0.2$, where edges are colored by the values of directed ORC (left), complex-weighted ORC for the magnetic Laplacian with $\theta = 2\pi/3$ (middle) and undirected ORC (right); (Bottom) Distributions of the corresponding ORC values, where within-community edges (pink) are reciprocal (otherwise same as Figure \ref{fig:sbm-dirtri}) and between-community edges (cyan) are from community 1 to 2. Curvature gaps measured by Hellinger distance are $0.816, 0.004, 0.004$ from left to right.}
    \label{fig:sbm-dirtri-reci}
\end{figure}

\subsection{Algorithm: Community Detection on Complex-Weighted Graphs}
Using the proposed curvature notions, we introduce the first curvature-based community detection algorithm for directed and complex-weighted graphs, summarized in Algorithm~\ref{alg:orc-sbm-directed}. Our approach resembles prior algorithms for undirected graphs, that leveraged the curvature gap~\citep{gosztolai2021unfolding,fesser2023augmentations}, where we iteratively delete the lowest-curvature edges to sever the inter-community bridges, and recover the communities as the strongly connected components of the remaining graph. In Algorithm~\ref{alg:orc-sbm-directed}, the threshold $\Delta$ controls how aggressively edges are removed. Rather than fixing it a priori, we propose to derive it from the data using a quality function, such as modularity: we remove edges greedily in increasing order of curvature and, among the resulting nested partitions, retain the one whose components attain the highest (directed) modularity with respect to the input graph; $\Delta$ is then the curvature of the last edge removed at this modularity-optimal cut. This eliminates a sensitive hyperparameter and lets the number of detected communities adapt to the graph, at the cost of one modularity evaluation per component split.

\begin{algorithm}[htbp]
\caption{Ollivier-Ricci Curvature-based Community Detection for a Given Complex-Weighted Graph}
\label{alg:orc-sbm-directed}
\begin{algorithmic}[1]
    \State{ \textbf{Input:} A complex-weighted graph \( G = (V, E) \), and a threshold $\Delta$}
    \State{ \textbf{Output:} Set of detected communities \( \mathcal{C} \)}

    \medskip

    \State{ \textbf{Step 1: Compute Ollivier-Ricci Curvature (ORC)}}
    \For{each directed edge \( \overrightarrow{uv} \in E \)}
        \State{ Compute the complex-weighted ORC \( \kappa(\overrightarrow{uv}) \).}
    \EndFor

    \medskip

    \State{ \textbf{Step 2: Iterative Edge Removal Based on ORC}}
    \State{ Sort edges \( E \) in ascending order of \( \kappa(\overrightarrow{uv}) \) (smallest first).}
    \While{there exists an edge \( \overrightarrow{uv} \) with \( \kappa(\overrightarrow{uv}) < \Delta \)}
        \State{ Remove the edge \( \overrightarrow{uv} \) with the smallest curvature if unique, otherwise remove one with the smallest curvature at random.}
        \State{ Recompute ORC for affected edges.}
    \EndWhile

    \medskip

    \State{ \textbf{Step 3: Identify Communities}}
    \State{ Extract the strongly connected components (SCCs) of the remaining graph.}
    \State{ These SCCs represent the detected communities $\mathcal{C}$.}

    \medskip
    \If{small isolated components exist}
        \State{ Merge them back based on inter-community ORC.}
    \EndIf
    
    \State{ \Return \( \mathcal{C} \)}
\end{algorithmic}
\end{algorithm}

\subsection{Experiments}
We evaluate Algorithm~\ref{alg:orc-sbm-directed} on synthetic directed stochastic block models and on the real-world datasets, comparing the undirected (un-ORC), directed (di-ORC), and complex-weighted (complex-ORC, from the magnetic Laplacian with $\theta = 2\pi/3$) curvatures. 
Detection quality is measured by clustering accuracy via the normalized mutual information (NMI).
The synthetic benchmark is composed of two different cases of two-block directed SBMs (size $100$, $p_{\text{in}} = 0.33$, $p_{\text{out}} = 0.2$): the first one (SBM-1) has directed within-community edges (with many directed triangles), while the second one (SBM-2) has reciprocal within-community edges, and both have one-way between-community edges (community $1\to 2$); see Figures~\ref{fig:sbm-dirtri} and~\ref{fig:sbm-dirtri-reci} for examples.

We examine the following two real-world datasets: C.elegans \cite{cionca2025celegans} and EU email network (Email-EU) \cite{leskovec2007email,yin2017email}. The C.elegans dataset contains $279$ neurons and $2194$ directed synaptic edges to describe the connectome of \textit{C.elegans}, together with the real functional classification of the neurons. The dataset has a moderate proportion of reciprocal edges ($0.21$) and a moderate portion of directed triangles ($0.16$), consistent with the known abundance of directional information flow. We chose the sensory and sex-specific classes in our experiments. The EU-email dataset contains $1005$ people and $24,929$ email communications between them, together with their departments as ground-truth classification. The dataset has high portion of reciprocal edges ($0.71$) and high portion of directed triangles ($0.78$), consistent with the known high reciprocity of email exchange. We should note here that the high proportion of directed triangles is from high proportion of reciprocal edges, and if we restrict to pure directed triangles without reciprocal edges, C.elegans has a higher proportion, $2\%$ vs $0.5\%$. We select departments 13 and 22 in our experiments.
\begin{table}[htbp]
    \centering
    \begin{tabular}{c|ccc}
         & un-ORC & di-ORC & complex-ORC \\ \hline
        SBM-1 & $0.259$ & $0.415$ & \hl{$1.00$} \\
        SBM-2 & $0.259$ & \hl{$1.00$} & $0.215$ \\
        C.elegans & $0.339$ & $0.182$ & \hl{$0.400$} \\
        Email-EU & $0.415$ & \hl{$0.441$} & $0.410$ \\
    \end{tabular}
    \caption{Community detection evaluation (NMI) of the algorithm with different versions of ORCs, where the top-performance one is marked in magenta.}
    \label{tab:community-all}
\end{table}

Our results demonstrate that both complex-ORC and di-ORC can effectively recover community structure in digraphs; see Table \ref{tab:community-all}. In particular, when directional edges and directed triangles are abundant, as in SBM-1 and C.elegans, complex-ORC achieves better performance, whereas when reciprocal edges are prevalent, as in SBM-2 and Email-EU, di-ORC performs better. Notably, for both SBMs, the undirected ORC exhibits almost no curvature gap, since the communities are mainly distinguished by the directional information rather than density. In contrast, di-ORC and complex-ORC incorporate this directional information and consequently produce a substantial curvature gap, leading to better community detection performance.

\section{Conclusions}\label{sec:conclusions}
In this paper, we proposed an extension of Ollivier's discrete Ricci curvature to complex-weighted graphs. We further derived relations to other graph characteristics and explored the use of the proposed notion in community detection.

Directions for future work include extending additional notions of discrete Ricci curvature, such as Forman~\citep{forman}, Bakry-Émery~\citep{bakry2013analysis}, and Lin-Lu-Yau~\citep{lin-lu-yau}, to complex-weighted networks. Another important avenue of investigation is the study of Ricci flow in this setting, which could provide insight into the dynamic evolution of network geometry and the analysis of network dynamics. In both the directed and complex-weighted setting, analyzing connections between our proposed notion of Ollivier-Ricci curvature and other graph characteristics could give rise to valuable theoretical and empirical insights. Lastly, leveraging the new notion for the design of algorithms and models for directed and complex-weighted networks, such as graph rewiring strategies and graph encodings, could lead to practical applications in Network Science and Graph Machine Learning.

\section*{Acknowledgments}
YT acknowledges support from ELBE postdoctoral fellowship. EW acknowledges funding from the Harvard College Research Program. MW was partially supported by
NSF awards CBET-2112085 and DMS-2406905, and an Alfred P. Sloan Research Fellowship in Mathematics.
This research was further developed with funding from the Defense Advanced Research Projects Agency (DARPA) under agreement no. HR0011-25-3-0205. The views, opinions, and/or findings expressed are those of the authors and should not be interpreted as representing the official views or policies of the Department of Defense or the U.S. Government.

\newpage
\bibliography{ref}

\appendix
\section{Detailed proofs}
\label{apx.proofs}

We include the detailed proofs of the theoretical results in the main text. 

\begin{proof}[Proof of Lemma 3.6] 
	Suppose there is a walk from $u$ to $v$ in the complex-weighted graph of phase $\theta = z\varphi_0$,  and without loss of generality, we denote the intermediate vertices in this walk by $j_1, j_2, \dots, j_{l-1}$, where the walk has length $l$. Since the walk has phase $\theta = z\varphi_0$, 
	\begin{align}
		\label{eq:lem-walk-phase}
		\theta(uj_1) + \sum_{i=1}^{l-2}\theta(j_ij_{i+1}) +\theta(j_{l-1}v) = 2k\pi + z\varphi_0,
	\end{align}
	where $k\in \mathbb{Z}_+$. Let's recall that in Section 3.3, we assume that the phases of all edges are integer multiples of $\varphi_0$, therefore, we can write $\theta(uj_1) = z_1\varphi_0,\, \theta(j_1j_2) = z_2\varphi_0,\, \dots,\, \theta(j_{l-1}v) = z_l\varphi_0$, where $z_1, \dots, z_l\in \mathbb{Z}_{\ge 0}$, and $(\sum_{i=1}^{l}z_i) \mod k = z$. 
	In the construction of the multilayer parametrization, each vertex has one copy in each layer, and edges are placed between the copies of the vertices that are connected in the complex-weighted graphs but their layers are selected according to the phases.
	In particular, since there is an edge from vertex $u$ to $j_1$ with phase $z_1\varphi_0$, we can find an edge from vertex $u$ in an arbitrary layer $i$ to vertex $j_1$ in layer $((i+z_1)\mod k)$. Similarly, we can find an edge from vertex $j_1$ to $j_2$, $\dots$, $j_{l-1}$ to $v$ in the corresponding layers. That is to say, we find a walk from vertex $u$ in layer $i$ to vertex $v$ in layer $(i+((\sum_{i=1}^{l}z_i) \mod k)) = i+z$. 
	
	Now, suppose there is a walk from $u$ in an arbitrary layer $i$ to $v$ in layer $i+z$ in the $k$-layer graph, and without loss of generality, we denote the intermediate vertices in this walk by $j_1, j_2, \dots, j_{l-1}$ in layers $z_1, z_2, \dots, z_{l-1}$, respectively, where the walk has length $l$.  By construction, since there is an edge from $u$ in layer $i$ to $j_1$ in layer $z_1$ in the $k$-layer graph, there is an edge from $u$ to $j_1$ in the complex-weighted graph with phase $((z_1 - i)\mod k)\varphi_0$.  Similarly, we can find edges from $j_1$ to $j_2$ with phase $((z_2 - z_1)\mod k)\varphi_0$, $\dots$, from $j_{l-1}$ to $v$ with phase $((i+z - z_{l-1})\mod k)\varphi_0$ in the complex-weighted graph. That is to say, we find a walk from vertex $u$ to vertex $v$ where the sum of edge phases is $((z_1 - i)\mod k)\varphi_0 + ((z_2 - z_1)\mod k)\varphi_0 + \dots + ((i+z - z_{l-1})\mod k)\varphi_0 = z\varphi_0$, thus the walk has phase $z\varphi_0$. 
	
		Therefore, the shortest-walk distance from $u$ to $v$ in the complex-weighted graph of phase $\theta = z\varphi_0$ is equal to the {shortest-path distance between the copy of $u$ in layer $i$ and the copy of $v$ in layer $i+z$ in the $k$-layer graph. When $u=v$ and $z\not\equiv 0 \mod k$, these are distinct copies of $u$, consistent with the convention above.}
	\end{proof}

\begin{proof}[Proof of Theorem 3.7]
	We first examine the probability measures on the neighborhoods. By construction of the multilayer parametrization, $u'$ in layer $((i-z_1)\mod k)$\footnote{If the residual is $0$, it corresponds to the $k$-th layer} is an in-neighbor of $u$ in layer $i$  in the $k$-layer graph if and only if $u'$ is an in-neighbor of $u$ in the complex-weighted graph and $\varphi_{u'u} = z_1\varphi_0$. Therefore, the measure on the in-neighbors of $u$ in the complex-weighted graph (as defined in Definition \ref{def:orc-complex}), has the same support as the measure on the in-neighbors of $u$ in layer $i$. Meanwhile, $v'$ in layer $((i+z_0+z_2)\mod k)$ is an out-neighbor of $v$ in layer $((i+z_0)\mod k)$ in the $k$-layer graph if and only if $v'$ is an out-neighbor of $v$ in the complex-weighted graph and $\varphi_{vv'} = z_2\varphi_0$. Hence, the measure on the out-neighbors of $v$ in the complex-weighted graph (as defined in Definition \ref{def:orc-complex}), has the same support as the measure on the out-neighbors of $v$ in layer $i+z_0$.

	Then for an arbitrary in-neighbor of $u$, $u'$, with $\varphi_{u'u} = z_1\varphi_0$, and an arbitrary out-neighbor of $v$, $v'$, with $\varphi_{vv'} = z_2\varphi_0$, the shortest-walk distance from $u$ to $v$ of phase $\theta = (z_1+z_0+z_2)\varphi_0$ is equal to the shortest-path distance from $u'$ in layer $((i-z_1)\mod k)$ to $v'$ in $((i+z_0+z_2)\mod k)$, from Lemma \ref{lem:klayer-dist}.
	
	Therefore, $W_1(m_{in, u}, m_{out, v})$ for $\overrightarrow{uv}$ in the complex-weighted graph (as in Eq.~\eqref{equ:W1-complex}) is equal to $W_1(m_{u^{(i)}}, m_{v^{(i+z_0)}})$ in the $k$-layer graph (as in Eq.~\eqref{equ:W1}). Then the ORC of $\overrightarrow{uv}$ in the complex-weighted graph with $\varphi_{uv} = z_0\varphi_0$ (as in Eq.~\eqref{eq:orc-complex}) is equal to the ORC of $\overrightarrow{u^{(i)}v^{(i+z_0)}}$ in the (directed) $k$-layer graph (as in Eq.~\eqref{eq:orc}).
\end{proof}

\begin{proof}[Proof of Corollary 3.8]
	By construction, for each edge $\overrightarrow{uv}\in E$ with $\varphi_{uv} = z_0\varphi_0$ in the complex-weighted graph, there is an edge from $u$ in layer $i$ to $v$ in layer $i+z_0$ in the $k$-layer graph, for $i=1, ..., k$ with convention $k+i = i,\, \forall i$. Suppose there is also an edge $\overrightarrow{vu}\in E$ with $\varphi_{vu} = (k-z_0)\varphi_0$ in the complex-weighted graph, then there is also an edge from $v$ in layer $i+z_0$ to $v$ in layer $i$ in the $k$-layer graph. Therefore, each edge in the $k$-layer graph is bi-directional, and we can ignore the direction. 
\end{proof}

\begin{proof}[Proof of Proposition 4.2]
	Let us denote the complex-weighted graph corresponding to the magnetic Laplacian with parameter $\theta = 2\pi/k$ by $G$. Then by construction, if there is only one directed edge from vertex $u$ to $v$ but not $v$ to $u$ in the digraph $H$, then there is one edge $uv$ with phase $\theta$ and one edge $vu$ with phase $2\pi - \theta$ in $G$; if there are two directed edges from vertex $u$ to $v$ and $v$ to $u$, then there are two edges $uv$, $vu$ with phase $0$ in $G$. 
	
	Suppose we have an undirected cycle $C$, ignoring the direction of edges in $C$, in $H$, and we denote the vertices in $C$ as $j_1, j_2, \dots, j_l$, where $l = |C|$ is the actual length, or the number of edges in $C$. 
	
	Now, we consider the corresponding edges in the $k$-layer graph $G^{(k)}$. (1) Let us start from $j_1$ in layer 1. (2) If $j_1j_2\in E(H),\, j_2j_1\notin E(H)$, then there is an edge from $j_1$ in layer 1 to $j_2$ in layer 2; if $j_1j_2\notin E(H),\, j_2j_1\in E(H)$, then there is an edge from $j_1$ in layer 1 to $j_2$ in layer k; otherwise, there is an edge from $j_1$ to $j_2$ in layer 1. (3) We can continue to examine the edges in $C$ as in step (2), up to the last edge $j_lj_1$. To determine which layer the vertex $j_1$ at the end of the process lies, we need to count how many times we forward versus backward the layers in $G^{(k)}$, i.e., the number of edges oriented in the direction from $j_1$ to $j_2$ minus the number of edges oriented in the opposite direction, denoted by $l_e$. The final $j_1$ is in layer $1+ l_e\mod k$. If $l_e\mod k = 0$, then we have a cycle of length $|C|$ in $G^{(k)}$; otherwise, we can continue the above process (1) to (3) for $j_1$ in layer $1+ l_e\mod k$, and at the end of the process, we arrive at vertex $j_1$ in layer $(2(l_e\mod k)\mod k) + 1$. If $2(l_e\mod k)\mod k=0$, we have a cycle of length $2|C|$ in $G^{(k)}$; otherwise, we can repeat the above step (process (1) to (3)) until we reach the smallest $k_l$ such that  $k_l(l_e\mod k)\mod k=0$, i.e., $k_l(l_e\mod k)$ is the lowest common multiple between $(l_e\mod k)$ and $k$. In this way, we have a cycle of length $k_l|C|$.  
	
	Finally, since the magnetic Laplacian is Hermitian, from Corollary \ref{cor:ORC-k-undi}, we can effectively ignore the direction of edges in the multilayer parametrization. Therefore, following the opposite direction of cycle $C$ in $H$ will give us the same cycle. Therefore, we do not need to specify the direction in calculating $l_e$, and this gives us the effective length $|C|_e$.    
\end{proof}

\begin{proof}[Proof of Proposition 4.3]
		For each undirected cycle $C$ in $H$, the length of the corresponding cycles in $G^{(k)}$ is also $|C|$ if and only if $(|C|_e \hspace{-.5em}\mod k) = 0$, if and only if $(|C|_e(2\pi/k)\hspace{-.5em}\mod 2\pi) = 0$, if and only if the phase of each directed cycle in the complex-weighted graph is a multiple of $2\pi$, if and only if the complex-weighted graph is {structurally balanced}, by definition. 
	\end{proof}

\begin{proof}[Proof of Lemma 4.7]
	We first note that for each directed edge $\overrightarrow{uv}\in E$, the probability measures on the in-neighbors of $u$ and on the out-neighbors of $v$ in the complex-weighted graph $G$ are the same as those in the digraph $\bar{G}$, ignoring the phase part.  
	
	Then, for an arbitrary in-neighbor of $u$, $u'$, and an arbitrary out-neighbor of $v$, $v'$, each walk from $u'$ to $v'$ has the same phase as the walk $u'uvv'$, by the definition of directed structural balance. Hence, the shortest-walk distance from $u'$ to $v'$ in the complex-weighted graph $G$ is the same as the shortest-path distance from $u'$ to $v'$ in the digraph $\bar{G}$, ignoring the phase part. 
	
	Altogether, we can show that ORC in the complex-weighted graph $G$ is equal to the ORC in the digraph $\bar{G}$, ignoring the phase part. 
\end{proof}

\begin{proof}[Proof of Theorem~\ref{thm:directed-combinatorial-bound}]
    The upper bound follows from the same zero-cost overlap argument as in Jost and Liu. A pair $(x,y)\in X\times Y$ can be transported at cost $0$ only if $x=y$. These vertices are precisely $C_3(u,v)=X\cap Y$, and the amount of mass that can be matched at zero cost is
    \[
        \sum_{w\in C_3(u,v)}\min\left\{\frac1a,\frac1b\right\}=\frac{c}{a\vee b}=Z.
    \]
    All remaining mass must move a positive directed distance, hence at least one. Thus
    \[
        W_1(m_{in,u},m_{out,v})\geq 1-Z,
    \]
    and $\kappa(u,v)=1-W_1(m_{in,u},m_{out,v})\leq Z$.

    For the lower bound, we construct a transport plan. First, match the mass supported on $C_3(u,v)$ to itself at cost $0$, moving total mass $Z$. The residual mass is $1-Z$. If $x\in A^-(u,v)$, then for every $y\in Y$ there is a directed path
    \[
        x\to v\to y
    \]
    of length at most $2$. Similarly, if $y\in A^+(u,v)$, then for every $x\in X$ there is a directed path
    \[
        x\to u\to y
    \]
    of length at most $2$. After the zero-cost matching, the residual source mass on $A^-(u,v)$ is $p/a$ and the residual target mass on $A^+(u,v)$ is $q/b$; hence the union of these shortcut-supported rows and columns supports a subcoupling of mass
    \[
        M=\min\left\{1-Z,\frac{p}{a}+\frac{q}{b}\right\}
    \]
    at cost at most $2$. The remaining residual mass can always be transported along the baseline path
    \[
        x\to u\to v\to y,
    \]
    which has length $3$. Hence
    \[
        W_1(m_{in,u},m_{out,v})
        \leq 2M+3(1-Z-M)=3-3Z-M.
    \]
    Since $\kappa(u,v)=1-W_1(m_{in,u},m_{out,v})$, the claimed lower bound follows.
\end{proof}

\begin{proof}[Proof of Theorem~\ref{thm:complex-combinatorial-bound}]
    For the upper bound, zero-cost transport can occur only when the source and target vertices agree and the required phase is $0$. For the edge $\overrightarrow{uv}$, this means precisely that the common vertex belongs to $C_0(u,v)$. Hence at most $Z$ units of mass can be matched at zero cost. The remaining mass $1-Z$ must be transported at positive phase-compatible shortest-walk distance, and by definition, each such positive cost is at least $\ell_+$. Thus
    \[
        W_1(m_{\text{in},u},m_{\text{out},v})\geq \ell_+(1-Z),
    \]
    which gives
    \[
        \kappa(u,v)=1-\frac{W_1(m_{\text{in},u},m_{\text{out},v})}{\delta_{uv}}
        \leq
        1-\frac{\ell_+(1-Z)}{\delta_{uv}}.
\]
    If $\ell_+\geq \delta_{uv}$, this further implies $\kappa(u,v)\leq Z$.

    For the lower bound, we again construct an explicit transport plan. First match the mass on $C_0(u,v)$ to itself at cost $0$, moving total mass $Z$. If $x\in A^-(u,v)$ and $y\in Y$, then the walk $x\to v\to y$ has phase
    \[
        \varphi_{xv}+\varphi_{vy}
        \equiv
        \varphi_{xu}+\varphi_{uv}+\varphi_{vy}
        =
        \theta(xuvy),
    \]
    so it is admissible for the complex-weighted transport cost and has cost $r_{xv}+r_{vy}$. Similarly, if $y\in A^+(u,v)$ and $x\in X$, then the walk $x\to u\to y$ has phase
    \[
        \varphi_{xu}+\varphi_{uy}
        \equiv
        \varphi_{xu}+\varphi_{uv}+\varphi_{vy}
        =
        \theta(xuvy),
    \]
    and has cost $r_{xu}+r_{uy}$. Therefore, after the zero-cost matching, the residual source mass on $A^-(u,v)$ and residual target mass on $A^+(u,v)$ determine a subcoupling of mass $M$ supported on pairs with $x\in A^-(u,v)$ or $y\in A^+(u,v)$. This mass can be transported at a cost of at most $\overline{L}_2$. Every remaining residual pair can be transported along the reference walk $x\to u\to v\to y$, whose cost is at most $L_3$. Hence
    \[
        W_1(m_{\text{in},u},m_{\text{out},v})
        \leq
        M\overline{L}_2+(1-Z-M)L_3.
    \]
    Dividing by $\delta_{uv}$ and using the definition of complex-weighted ORC gives the claimed lower bound.
\end{proof}

\section{Additional experimental results}
We observe a clear curvature gap when the within-community edges are reciprocal and the probability of an edge to appear within communities, $p_{in}$, is much higher than the probability for between communities, $p_{out}$; see Figures \ref{fig:sbm-directed-inboutd} and \ref{fig:sbm-directed-dist-inrouth}. In particular, when the between-community edges are pointing in the same direction (from community 1 to 2), the directed ORC has a clearer curvature gap than the complex-weighted ORC for the magnetic Laplacian, as $p_{out}$ gets closer to $p_{in}$; see Figure \ref{fig:sbm-directed-inboutd}. However, as some between-community edges oppose the direction (from community 2 to 1), the directed ORC has a smaller curvature gap than the pervious case, while the complex-weighted ORC for the magnetic Laplacian has a clearer curvature gap than the previous case. Meanwhile, when the between-community edges are oriented randomly, i.e., equal probability to be from community 1 to 2 and from community 2 to 1, the complex-weighted ORC for the magnetic Laplacian has a clearer curvature gap than the directed ORC; see Figure \ref{fig:sbm-directed-dist-inrouth}. 

\begin{figure}[htbp]
    \centering
    \begin{tabular}{ccc}
        \includegraphics[width=.3\textwidth]{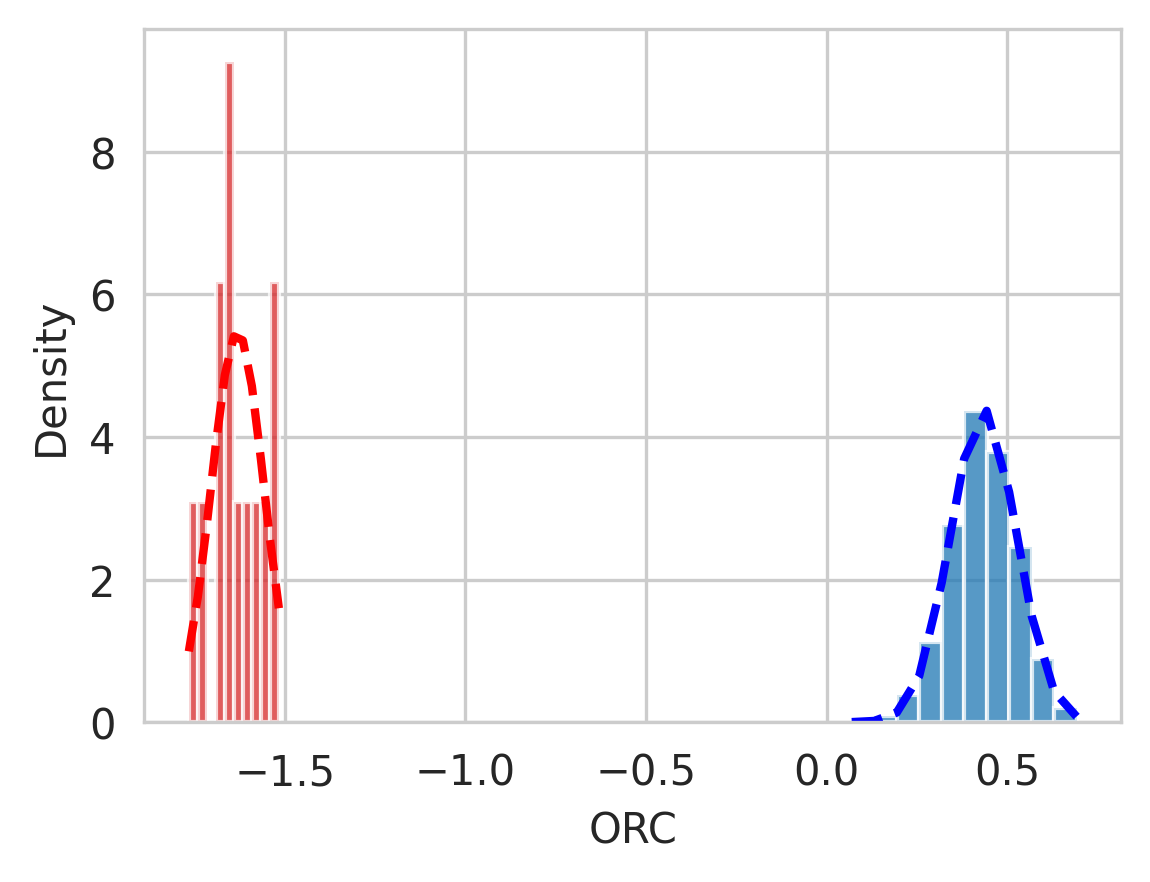} & \includegraphics[width=.3\textwidth]{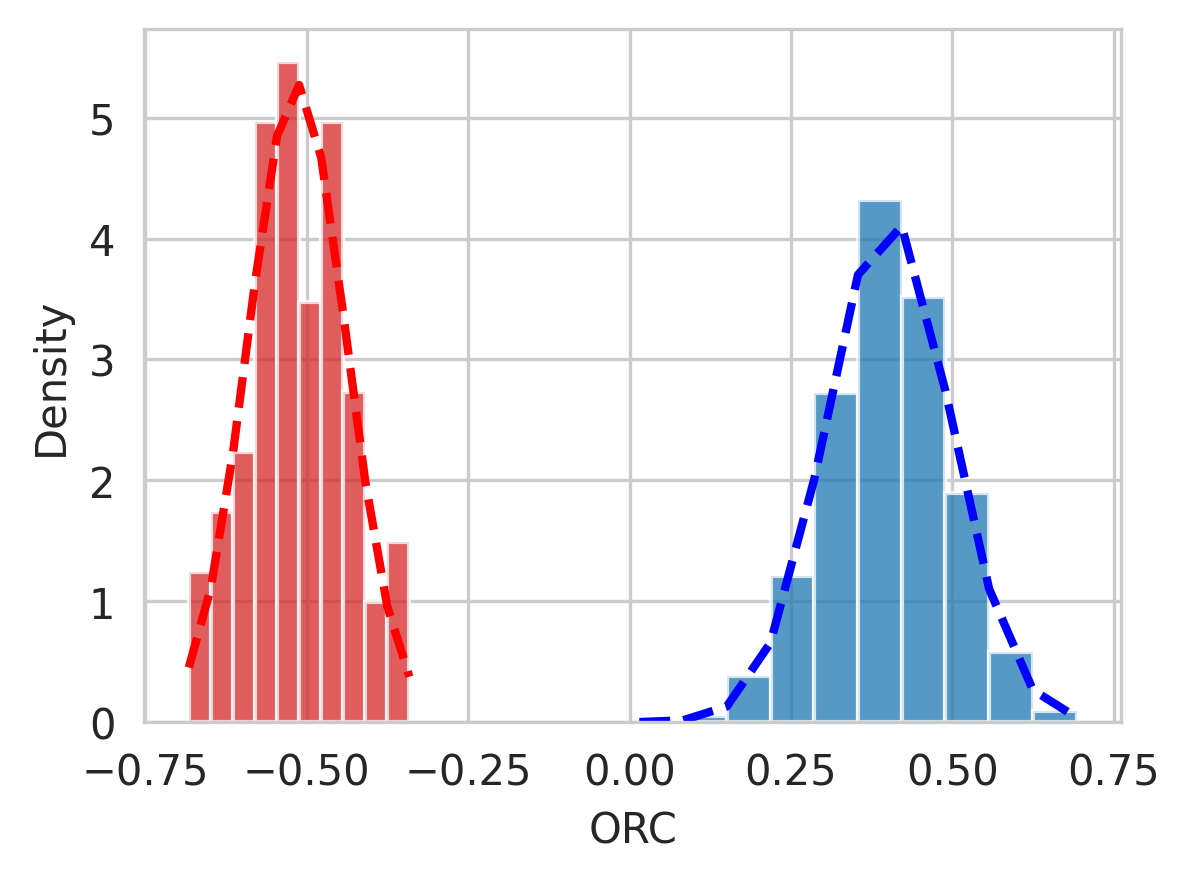} &
        \includegraphics[width=.3\textwidth]{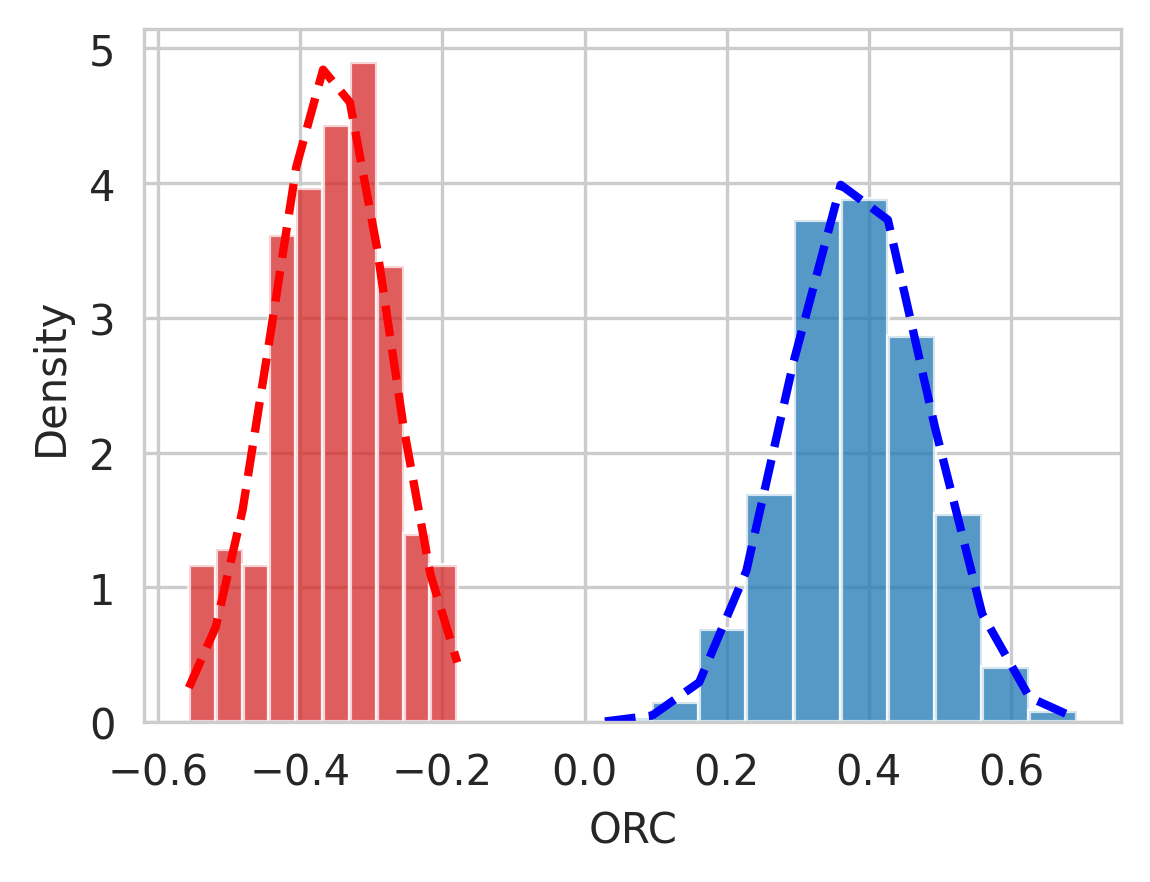} \\
        \includegraphics[width=.3\textwidth]{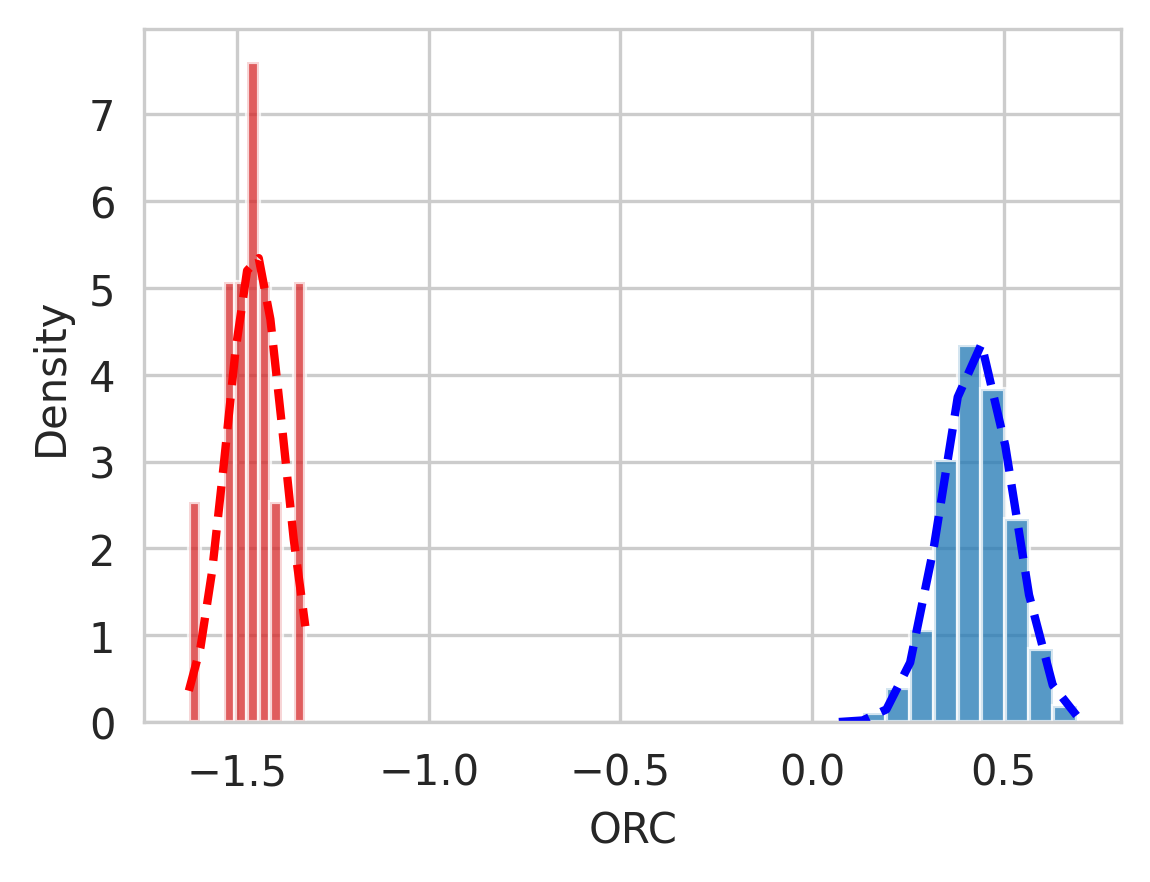} & \includegraphics[width=.3\textwidth]{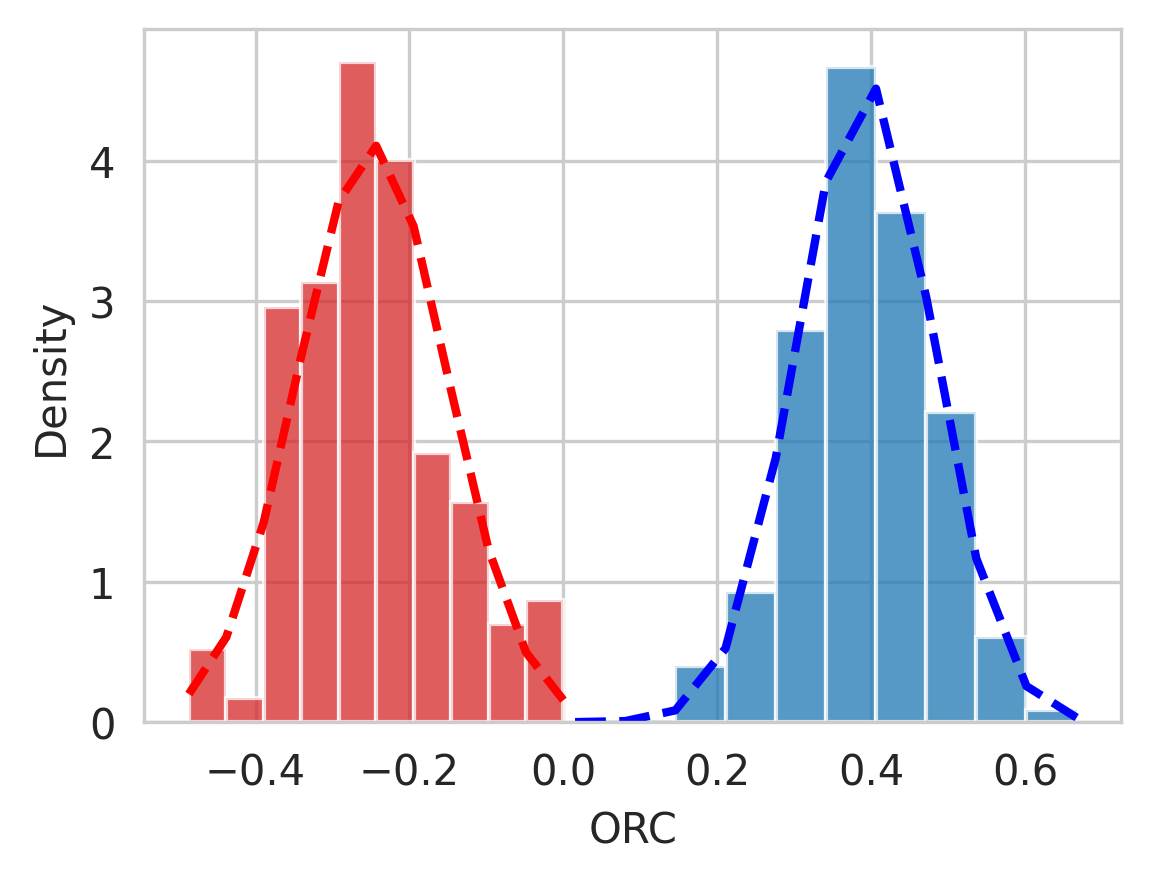} &
        \includegraphics[width=.3\textwidth]{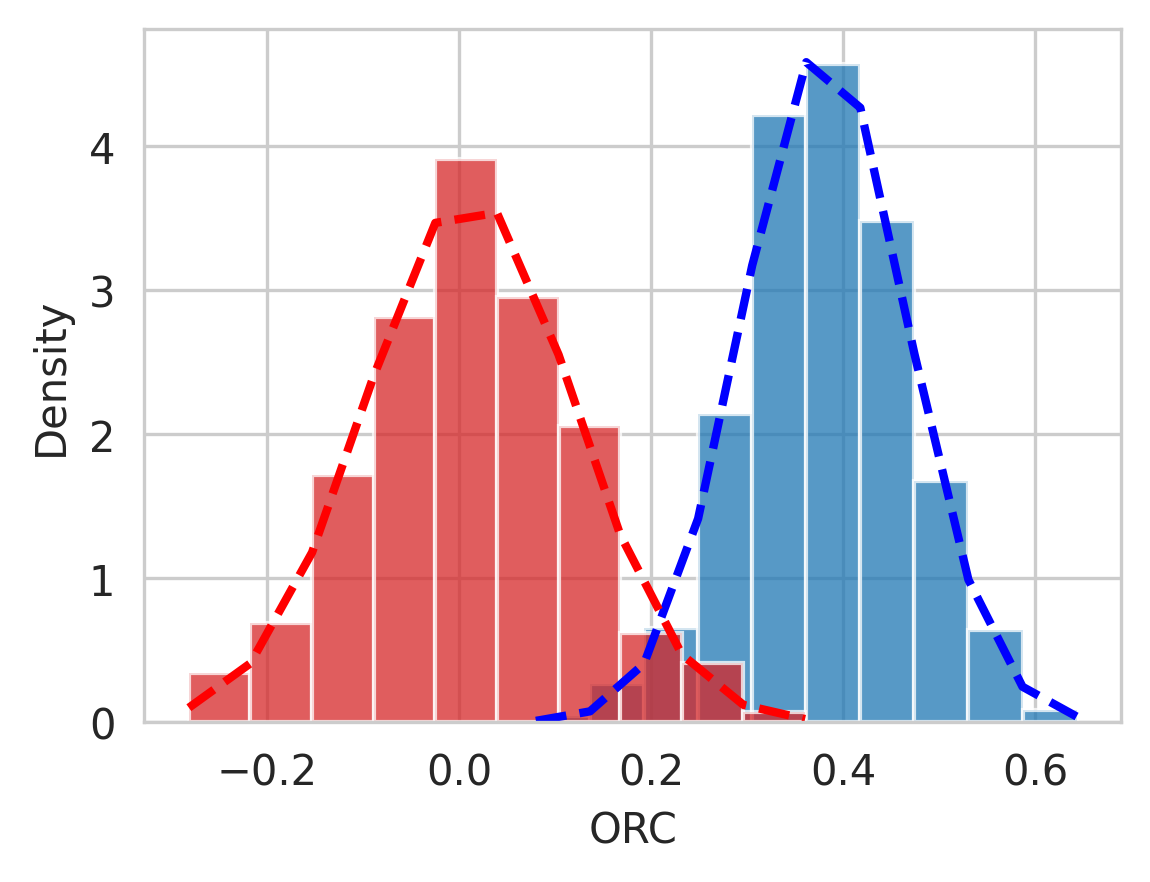}
    \end{tabular}
    \caption{Distribution of directed ORC (top) and complex-weighted ORC for the magnetic Laplacian with $\theta = 2\pi/3$ (bottom) of two-block directed SBMs with size $100$, $p_{\text{in}} = 0.5$ and $p_{\text{out}} = 0.01, 0.1, 0.2$ (left, middle, right, respectively), where within-community edges (blue) are reciprocal and between-community edges (red) are all from community 1 to 2.}
    \label{fig:sbm-directed-inboutd}
\end{figure}

\begin{figure}[htbp]
    \centering
    \begin{tabular}{ccc}
        \includegraphics[width=.3\textwidth]{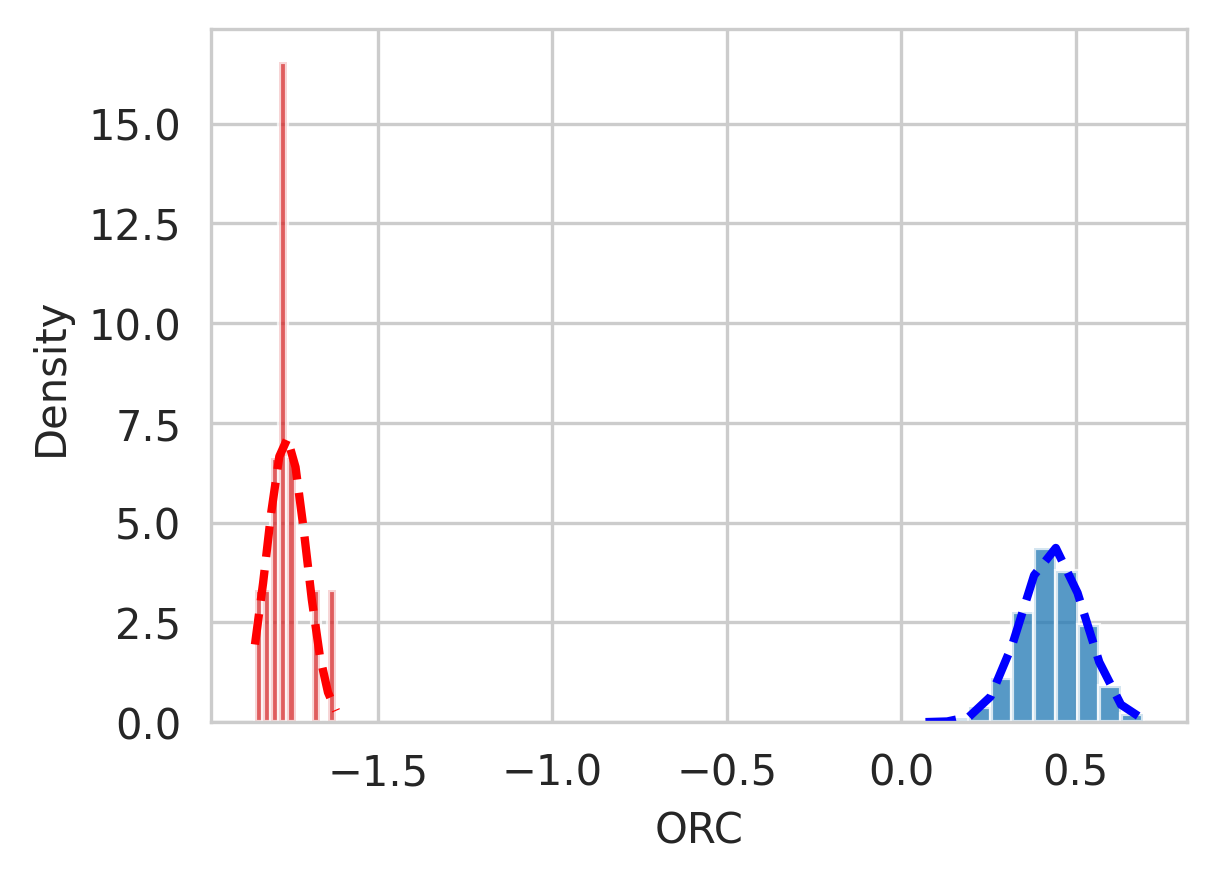} & \includegraphics[width=.3\textwidth]{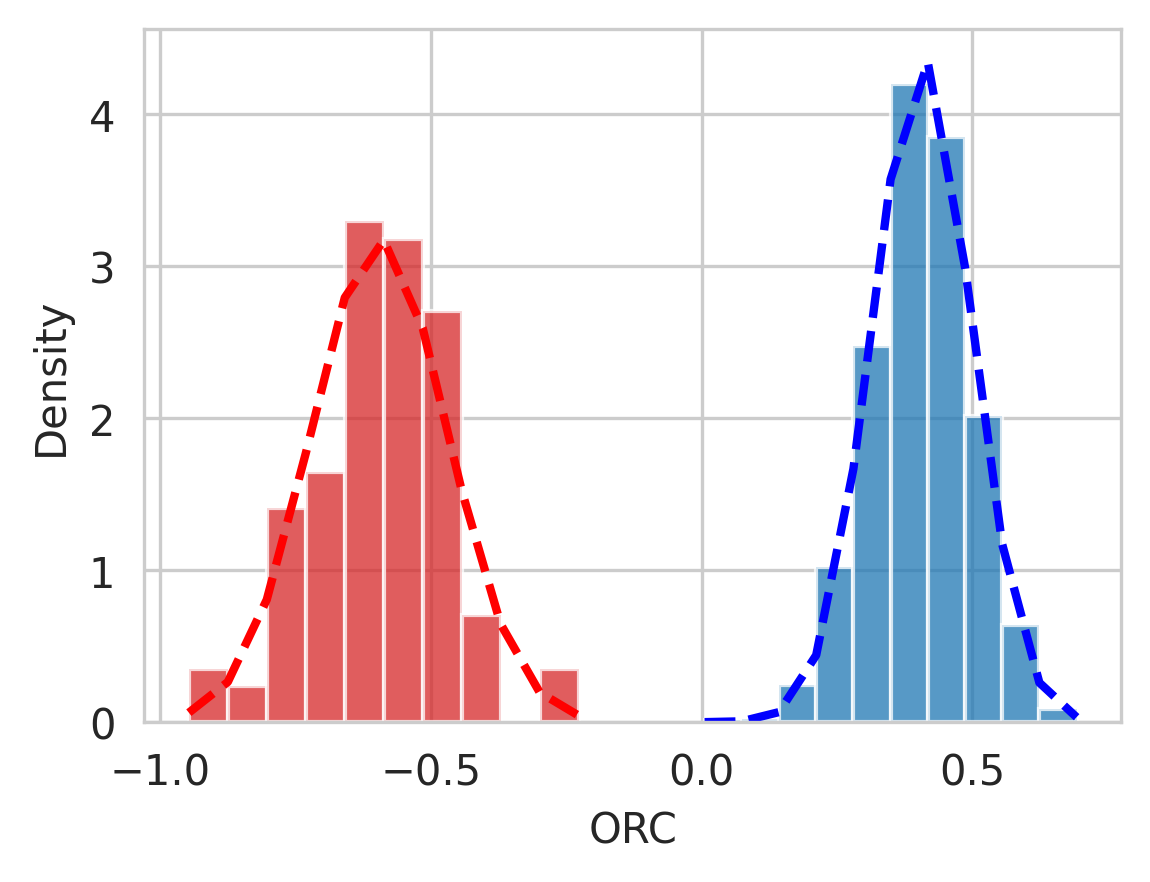} &
        \includegraphics[width=.3\textwidth]{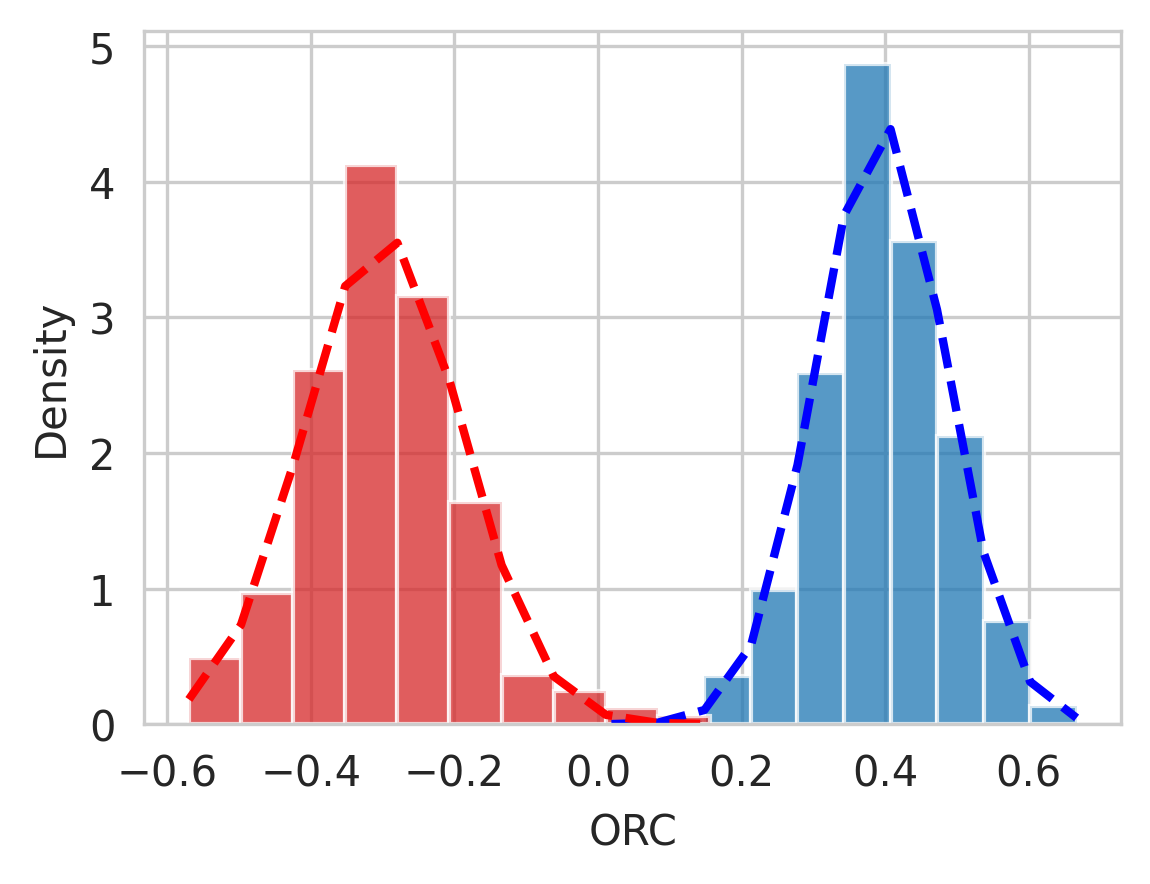} \\
        \includegraphics[width=.3\textwidth]{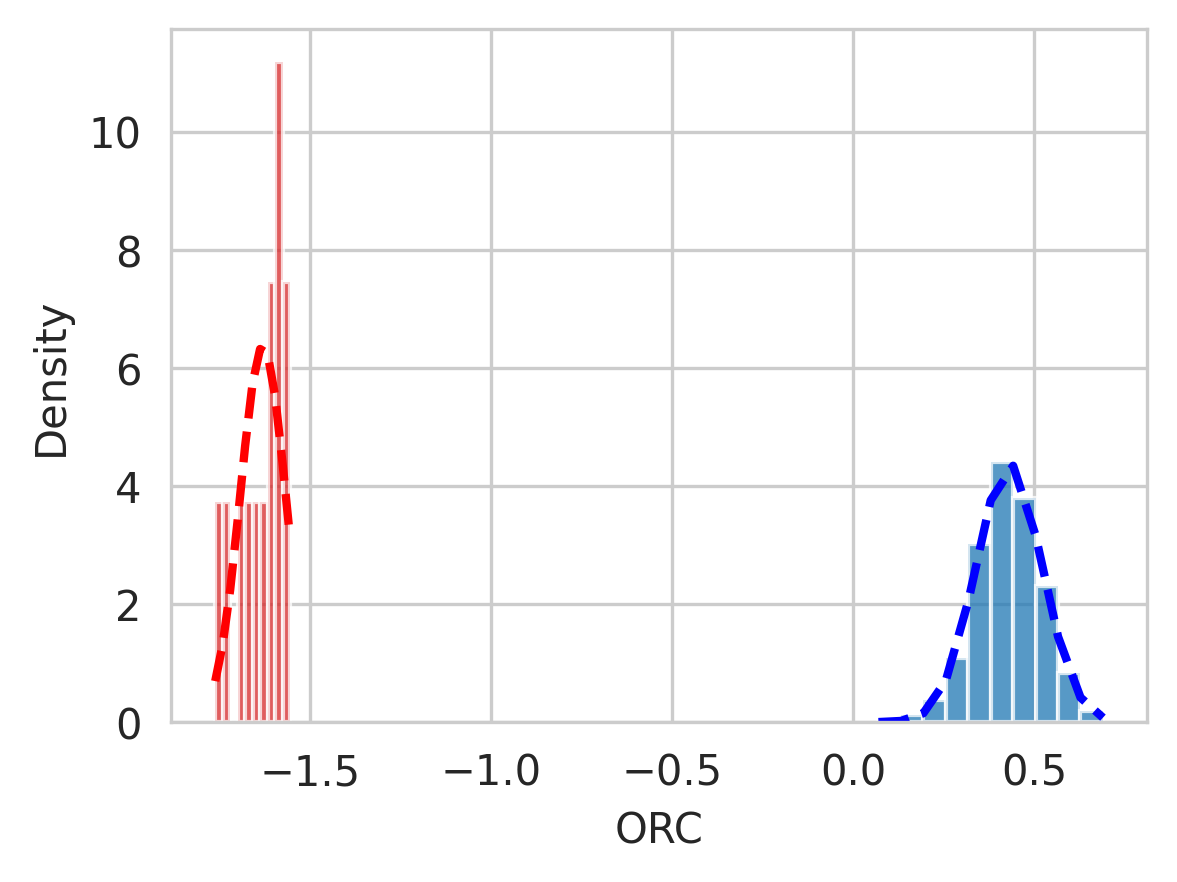} & \includegraphics[width=.3\textwidth]{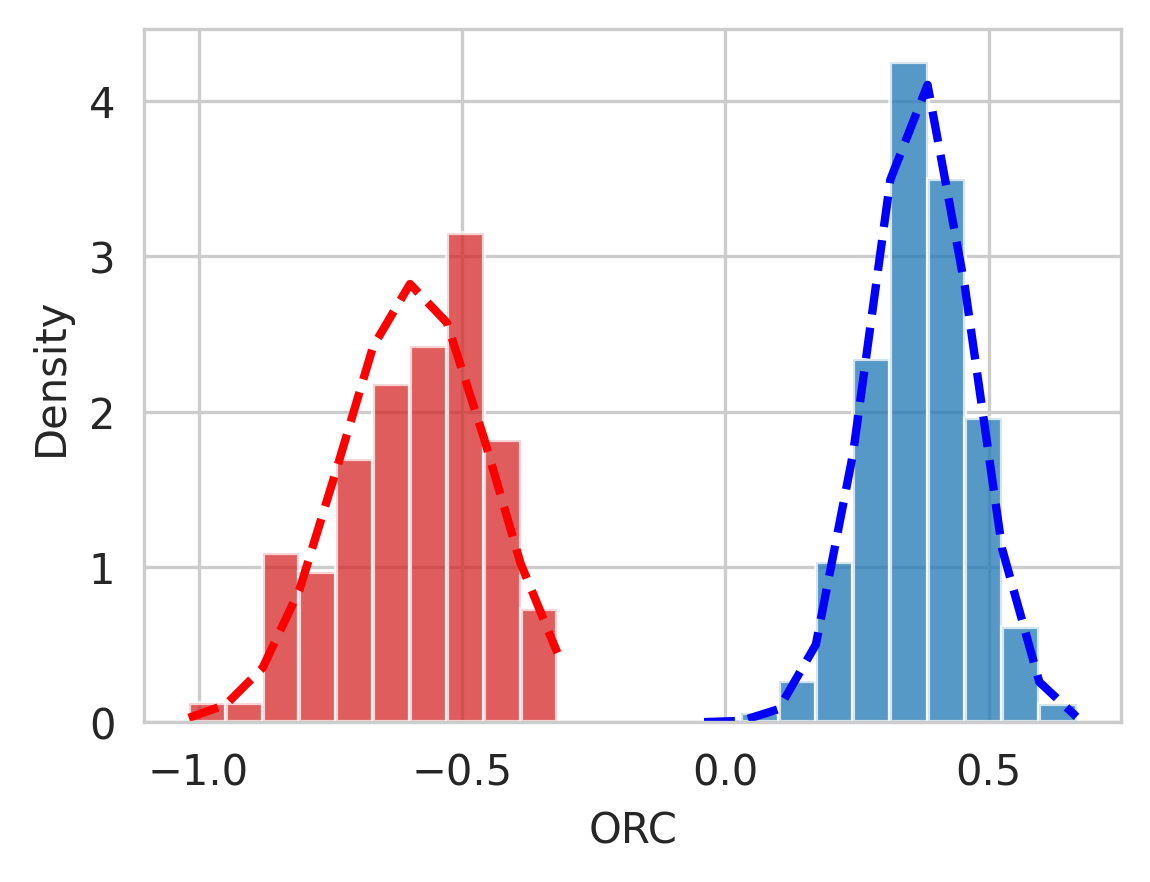} &
        \includegraphics[width=.3\textwidth]{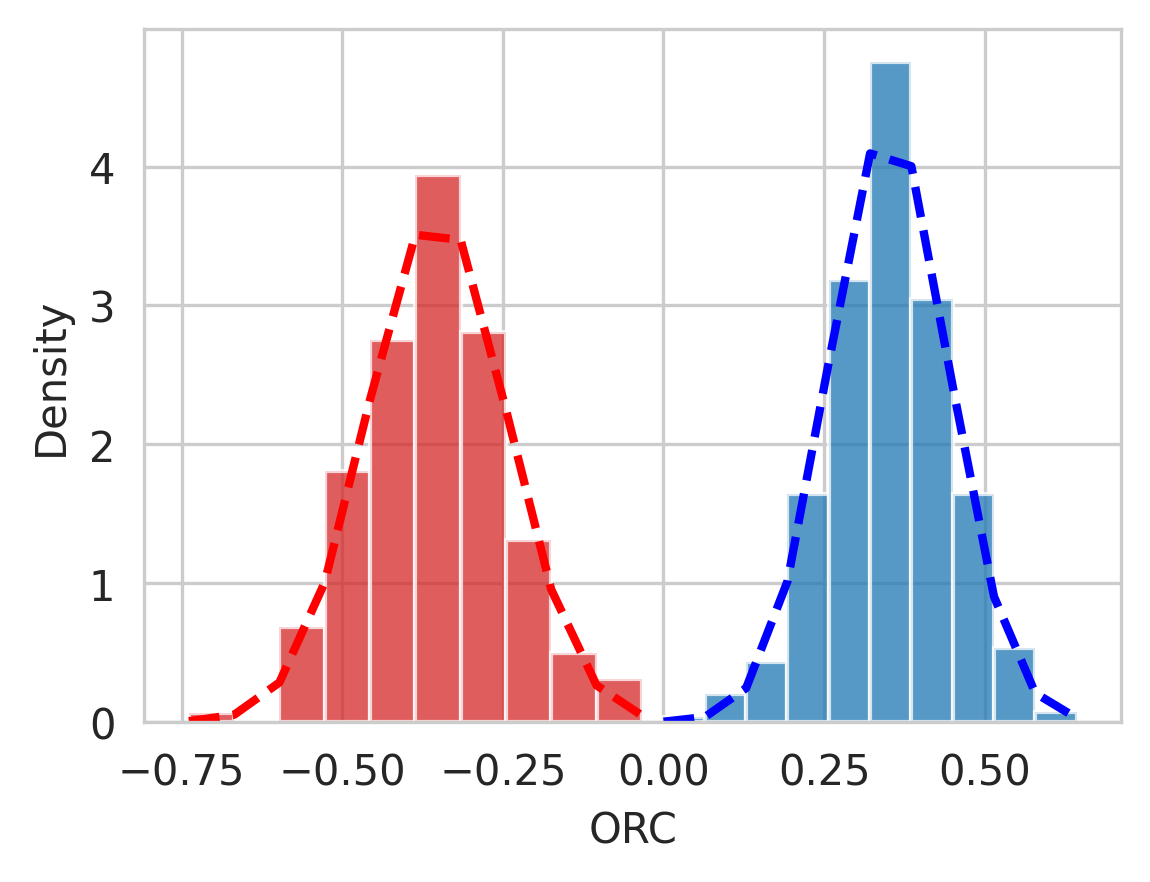}
    \end{tabular}
    \caption{Distribution of directed ORC (top) and complex-weighted ORC for the magnetic Laplacian with $\theta = 2\pi/3$ (bottom) of two-block directed SBMs with size $100$, $p_{\text{in}} = 0.5$ and $p_{\text{out}} = 0.01, 0.1, 0.2$ (left, middle, right, respectively), where within-community edges (blue) are reciprocal and are oriented randomly (equal probability from community 1 to 2 versus from community 2 to 1).}
    \label{fig:sbm-directed-dist-inrouth}
\end{figure}

In the other extreme, when all within-community edges are directed (from the smaller label to the larger one), we still observe a clear curvature gap with the complex-weighted ORC for the magnetic Laplacian when $p_{in}$ is sufficiently larger than $p_{out}$, but not the directed ORC; see Figures \ref{fig:sbm-directed-indoutd} and \ref{fig:sbm-directed-dist-indouth}. In particular, there is still a clear curvature gap in the complex weighted version when $p_{\mathrm{out}} = 0.2$ and between-community edges are oriented randomly, but not the directed version; see Figure \ref{fig:sbm-directed-dist-indouth} (right column). 
This is because, instead of putting the direction information as hard constraints to include it or not completely in the calculation, the magnetic Laplacian incorporates it softly, where edges, paths, and cycles without a desired direction can still be included but are penalized with much longer distances, thus still contributing to the ORC of relevant edges, even though the contribution may be insignificant. 
\begin{figure}[htbp]
    \centering
    \begin{tabular}{ccc}
        \includegraphics[width=.3\textwidth]{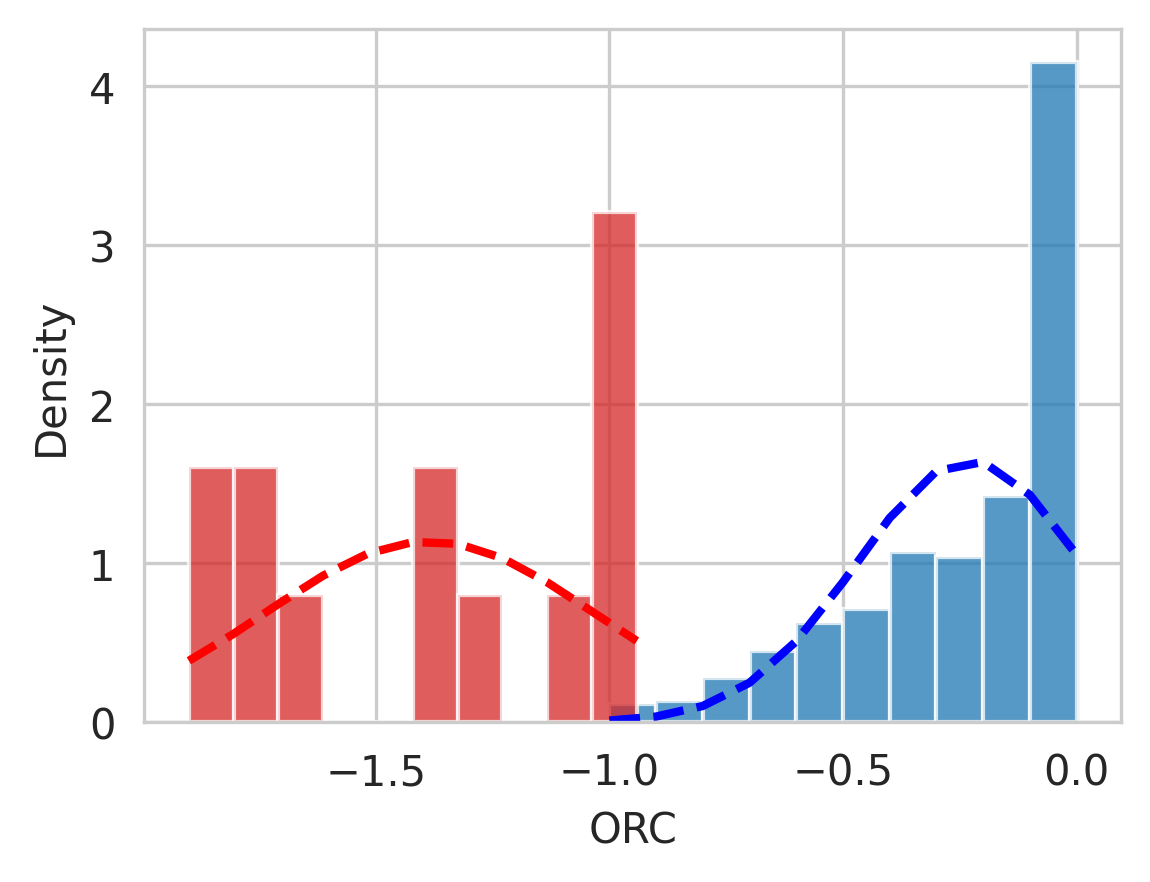} & \includegraphics[width=.3\textwidth]{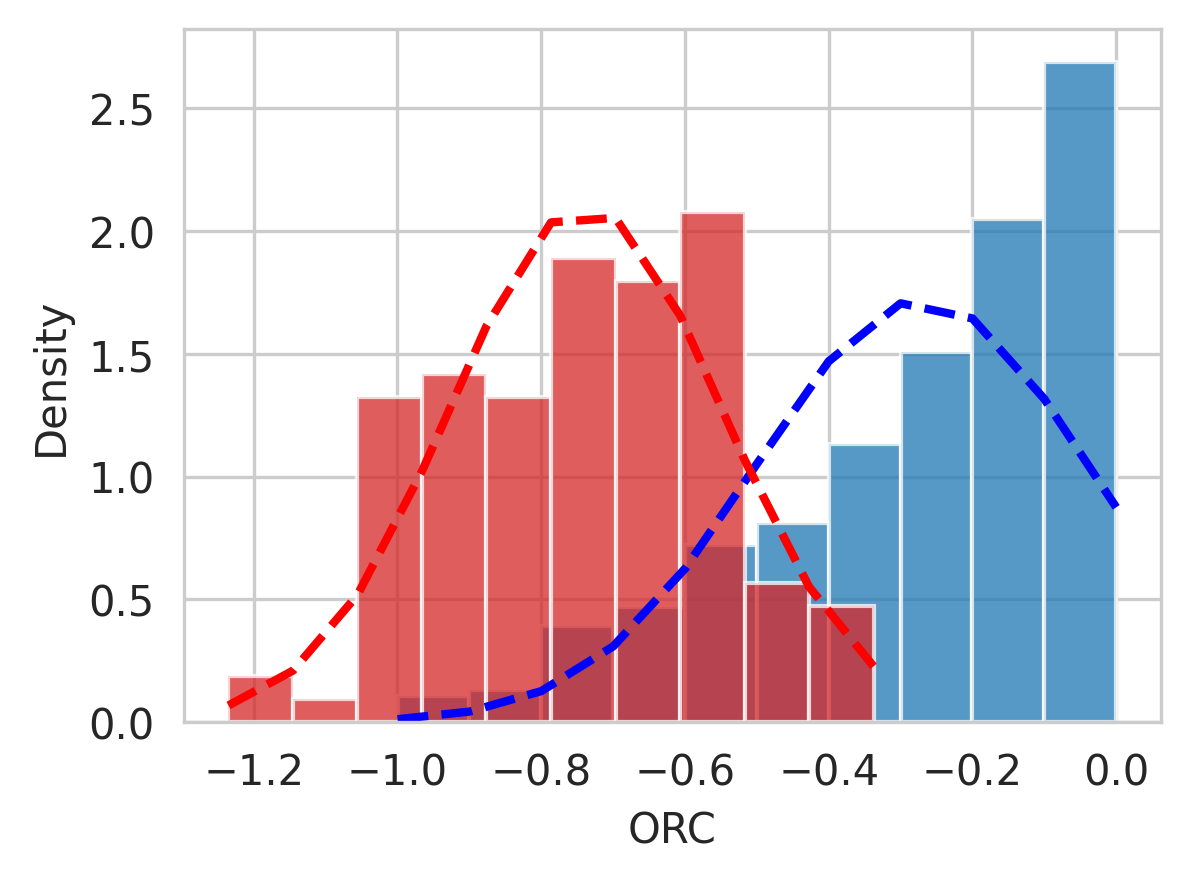} &
        \includegraphics[width=.3\textwidth]{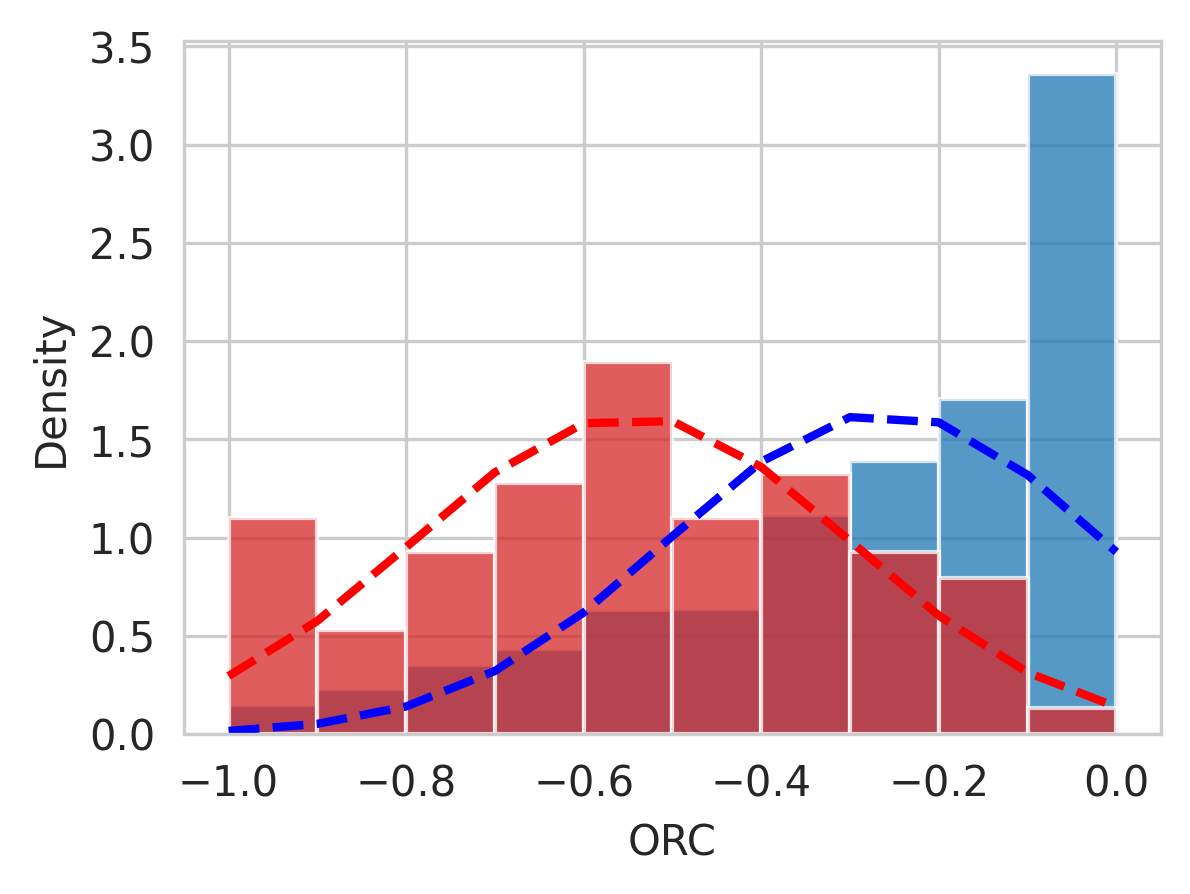} \\
        \includegraphics[width=.3\textwidth]{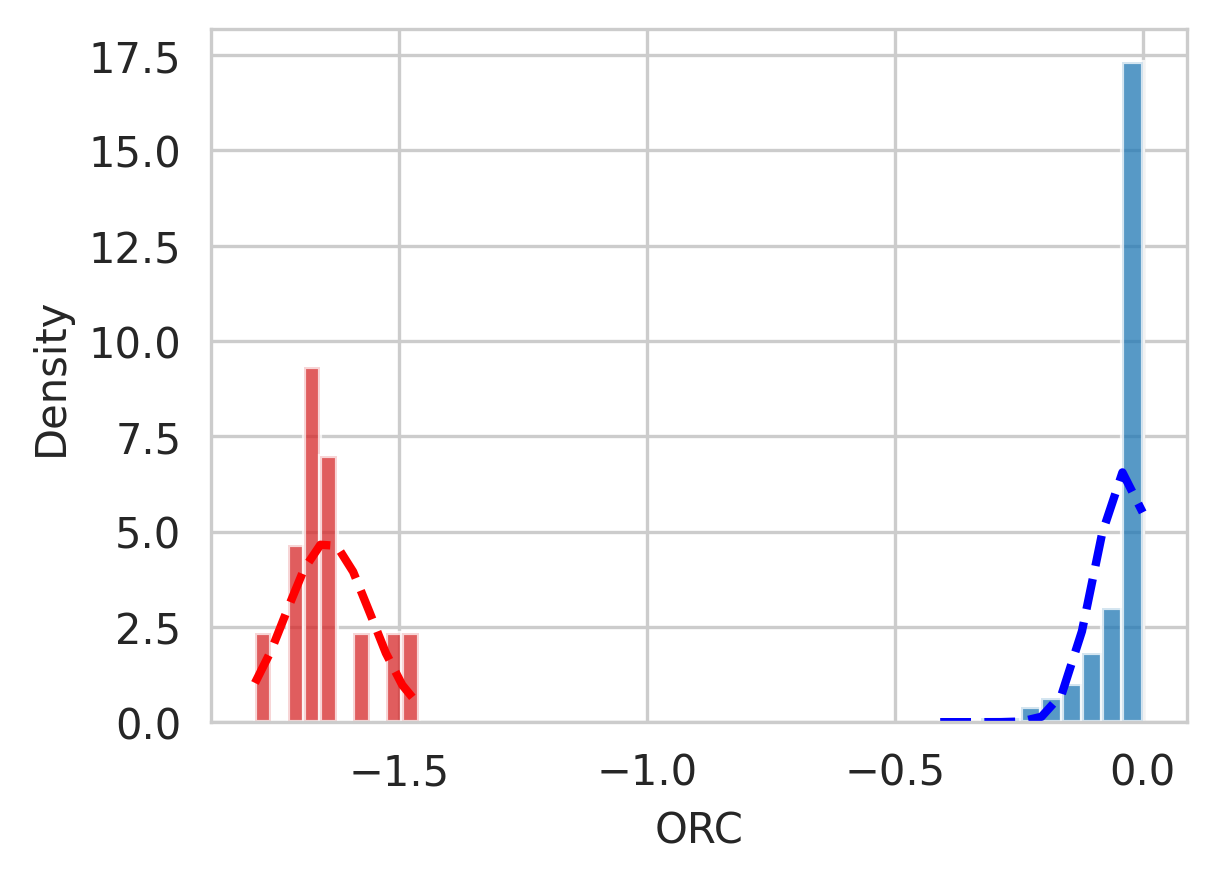} & \includegraphics[width=.3\textwidth]{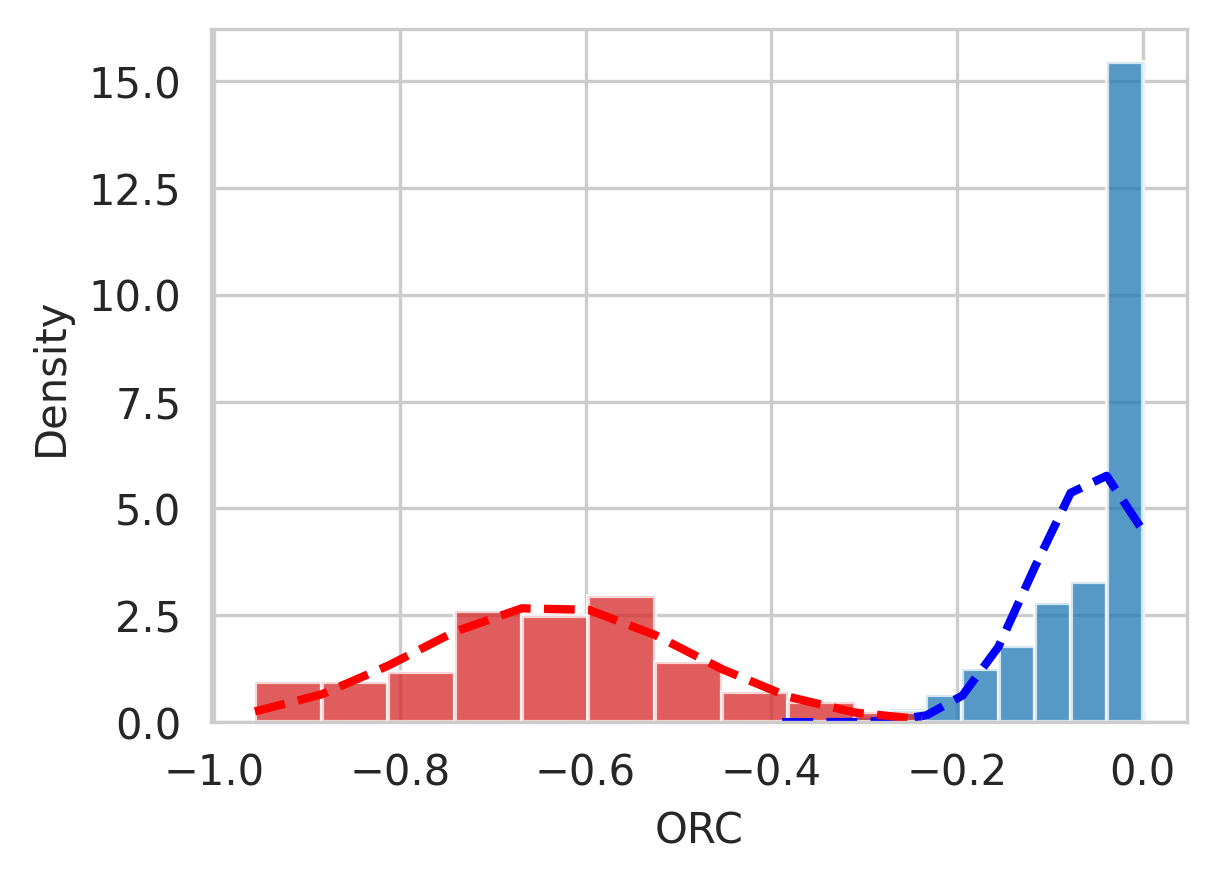} &
        \includegraphics[width=.3\textwidth]{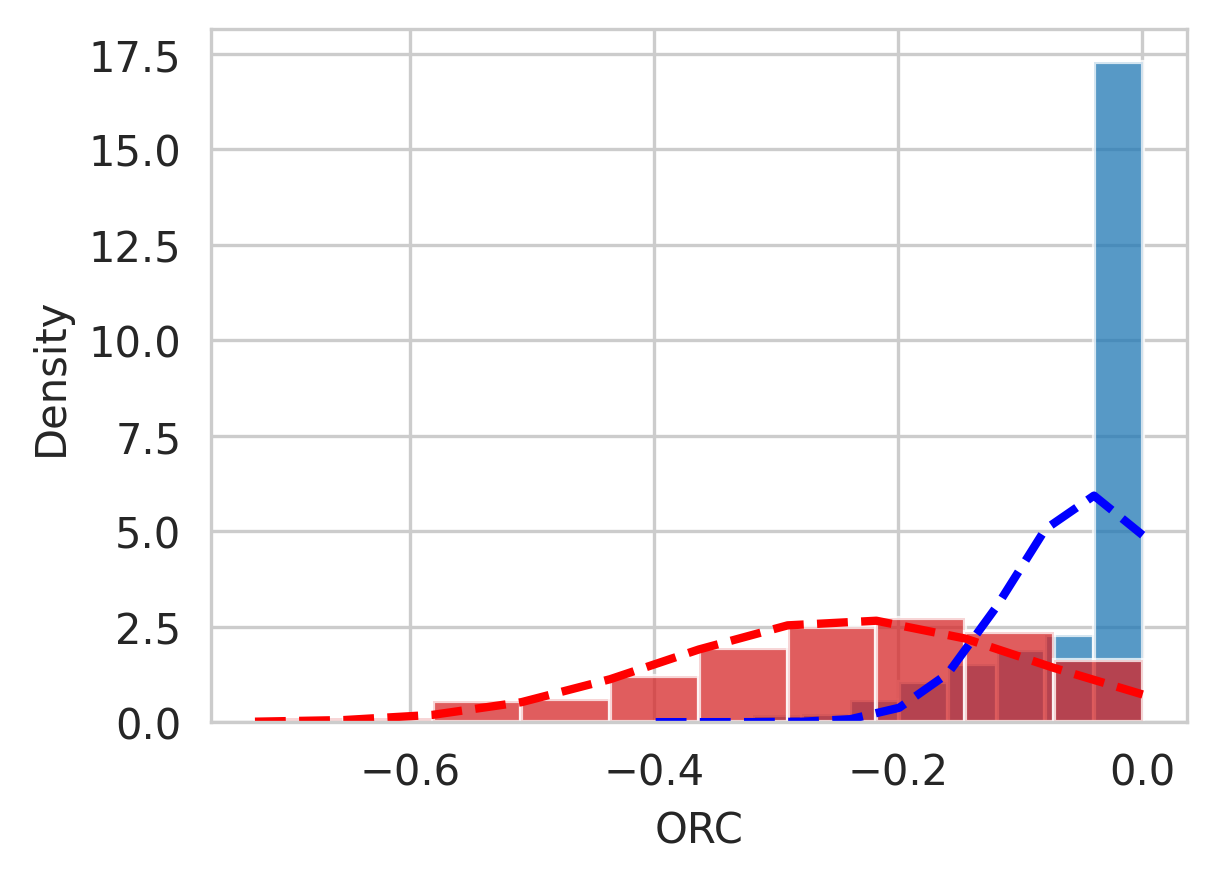}
    \end{tabular}
    \caption{Distribution of directed ORC (top) and complex-weighted ORC for the magnetic Laplacian with $\theta = 2\pi/3$ (bottom) of two-block directed SBMs with size $100$, $p_{\text{in}} = 0.5$ and $p_{\text{out}} = 0.01, 0.1, 0.2$ (left, middle, right, respectively), where within-community edges (blue) are directed (all from smaller label to larger one) and between-community edges (red) are all from community 1 to 2.}
    \label{fig:sbm-directed-indoutd}
\end{figure}

\begin{figure}[htbp]
    \centering
    \begin{tabular}{ccc}
        \includegraphics[width=.3\textwidth]{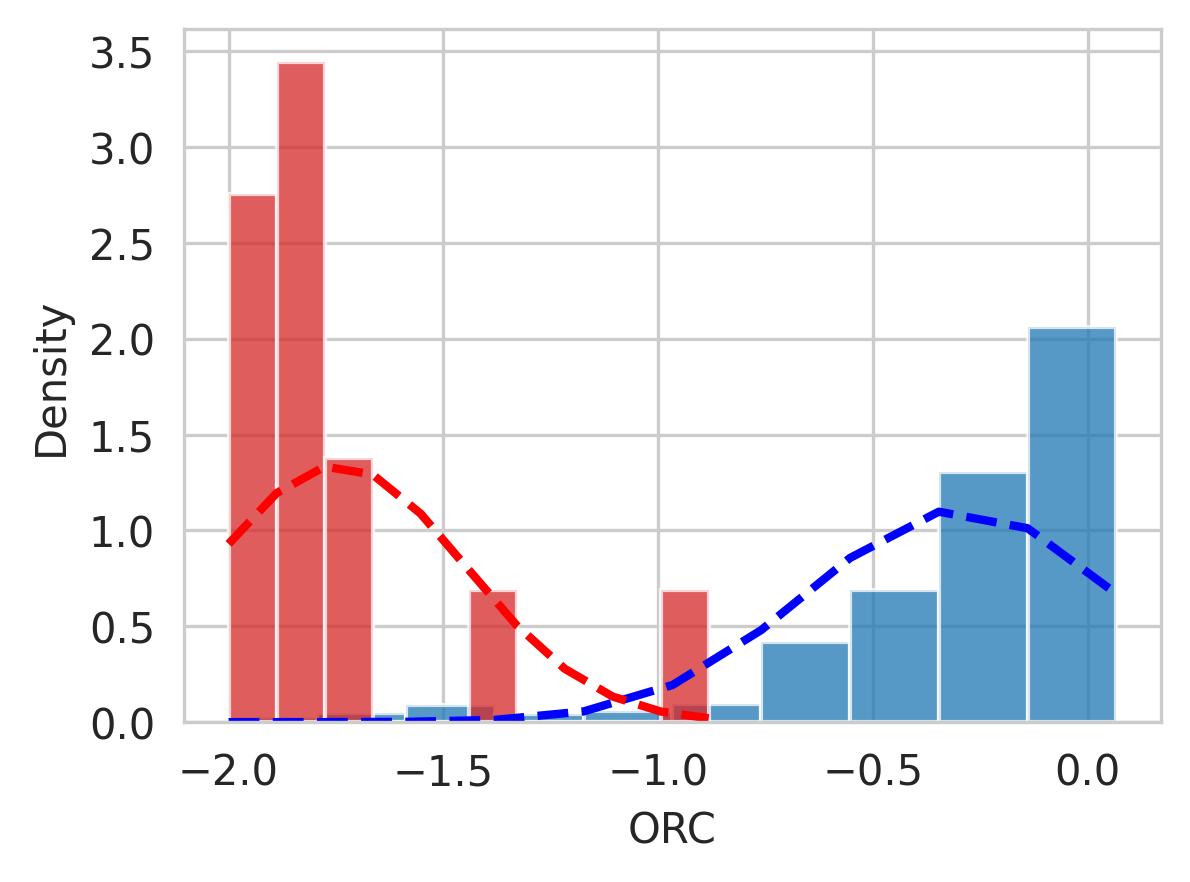} & \includegraphics[width=.3\textwidth]{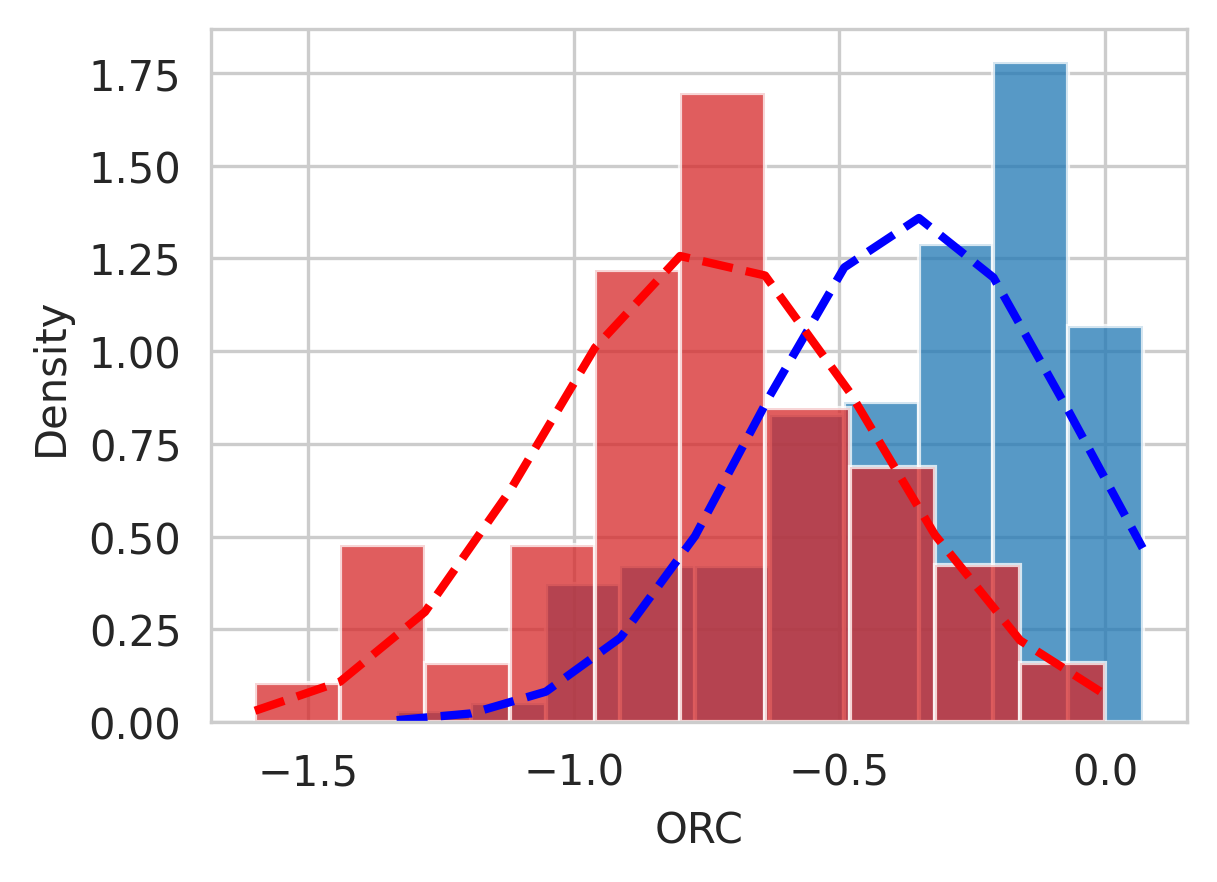} &
        \includegraphics[width=.3\textwidth]{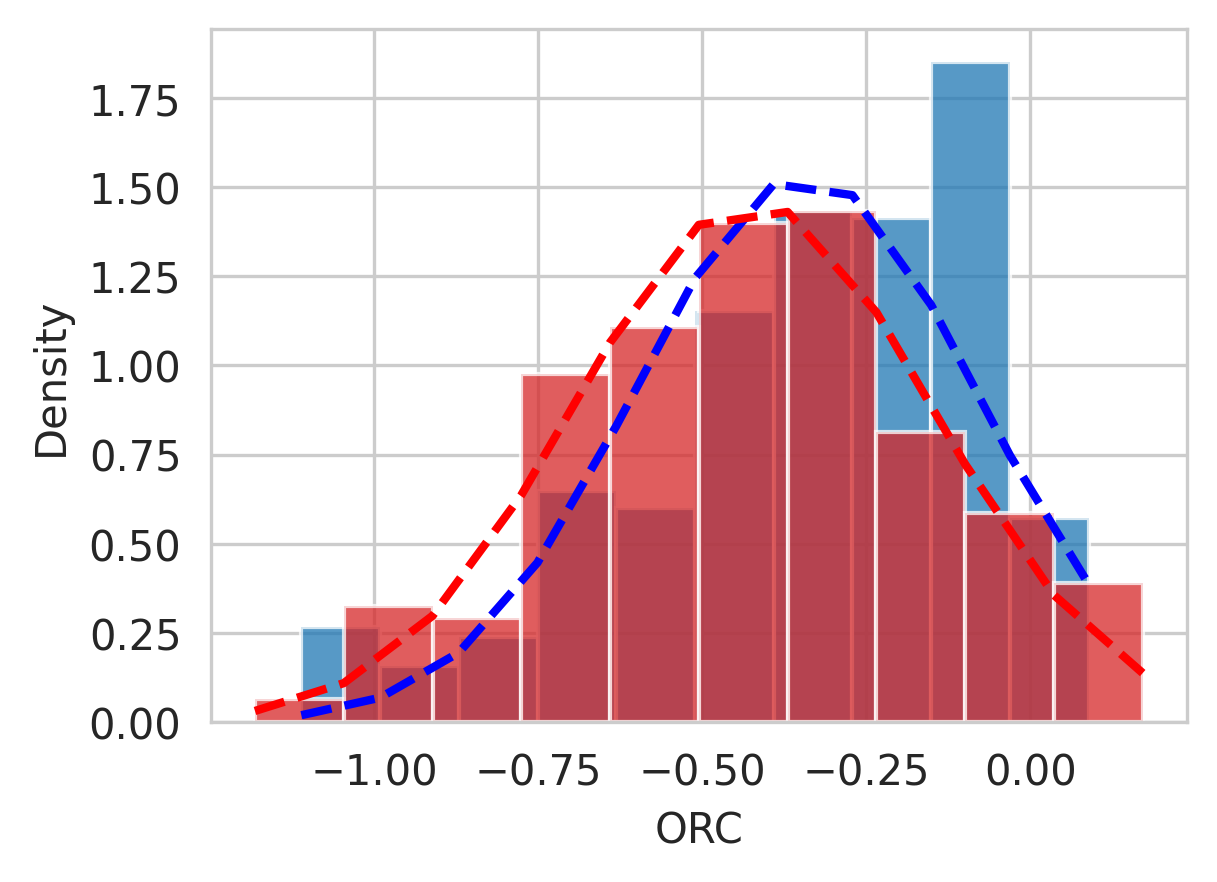} \\
        \includegraphics[width=.3\textwidth]{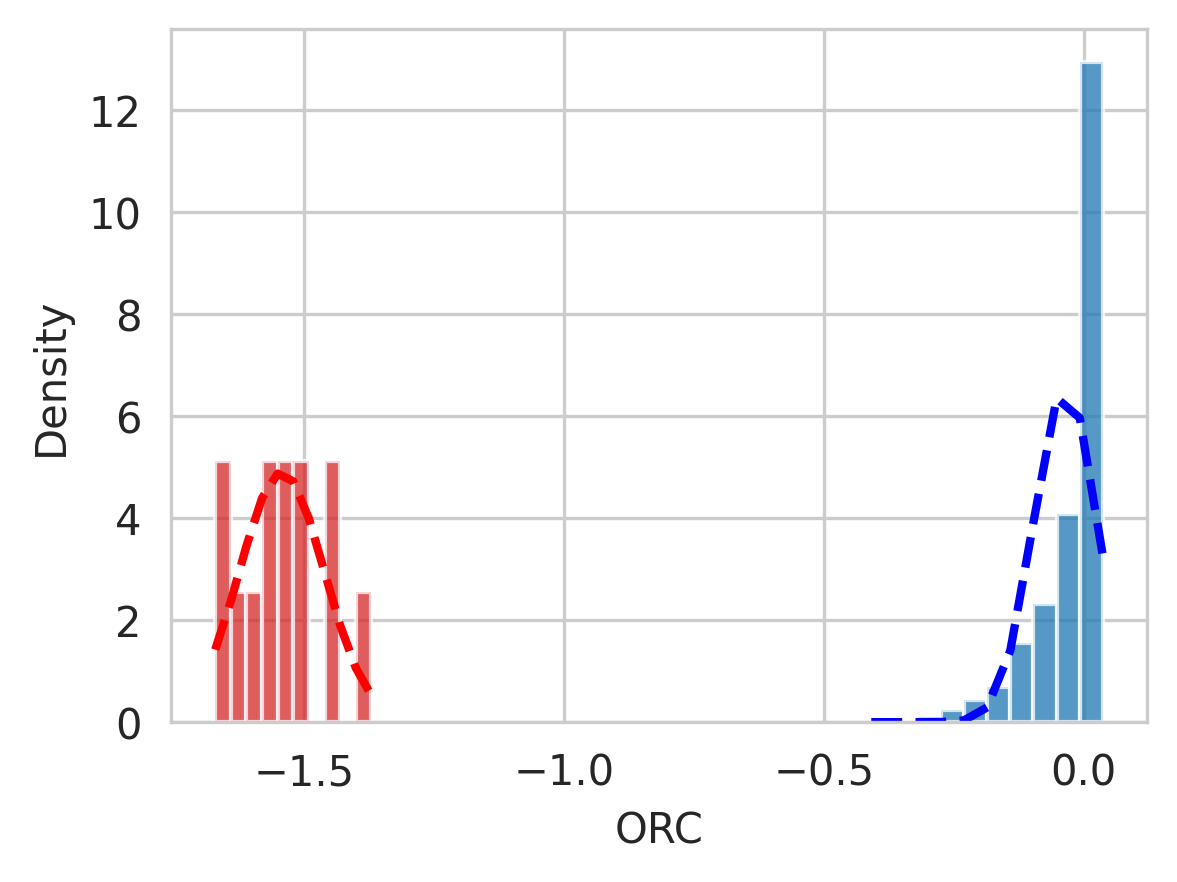} & \includegraphics[width=.3\textwidth]{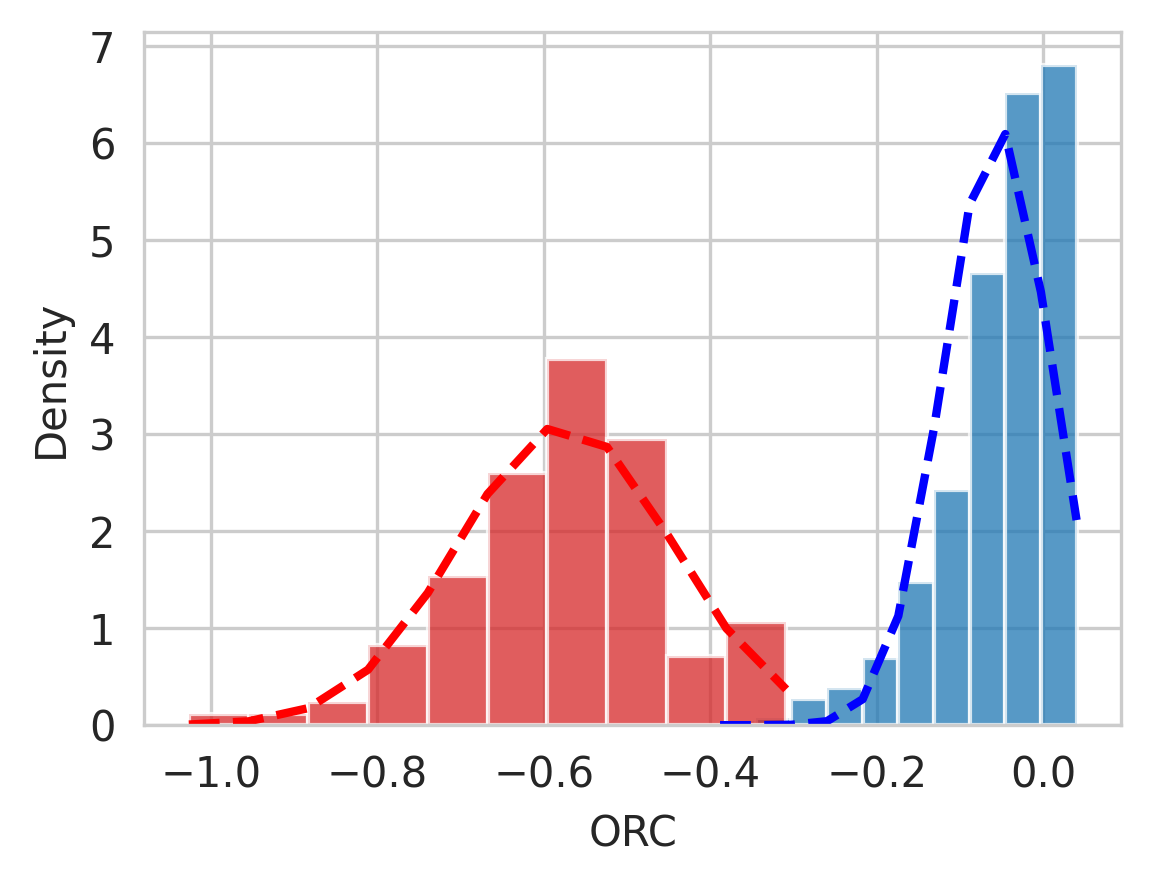} &
        \includegraphics[width=.3\textwidth]{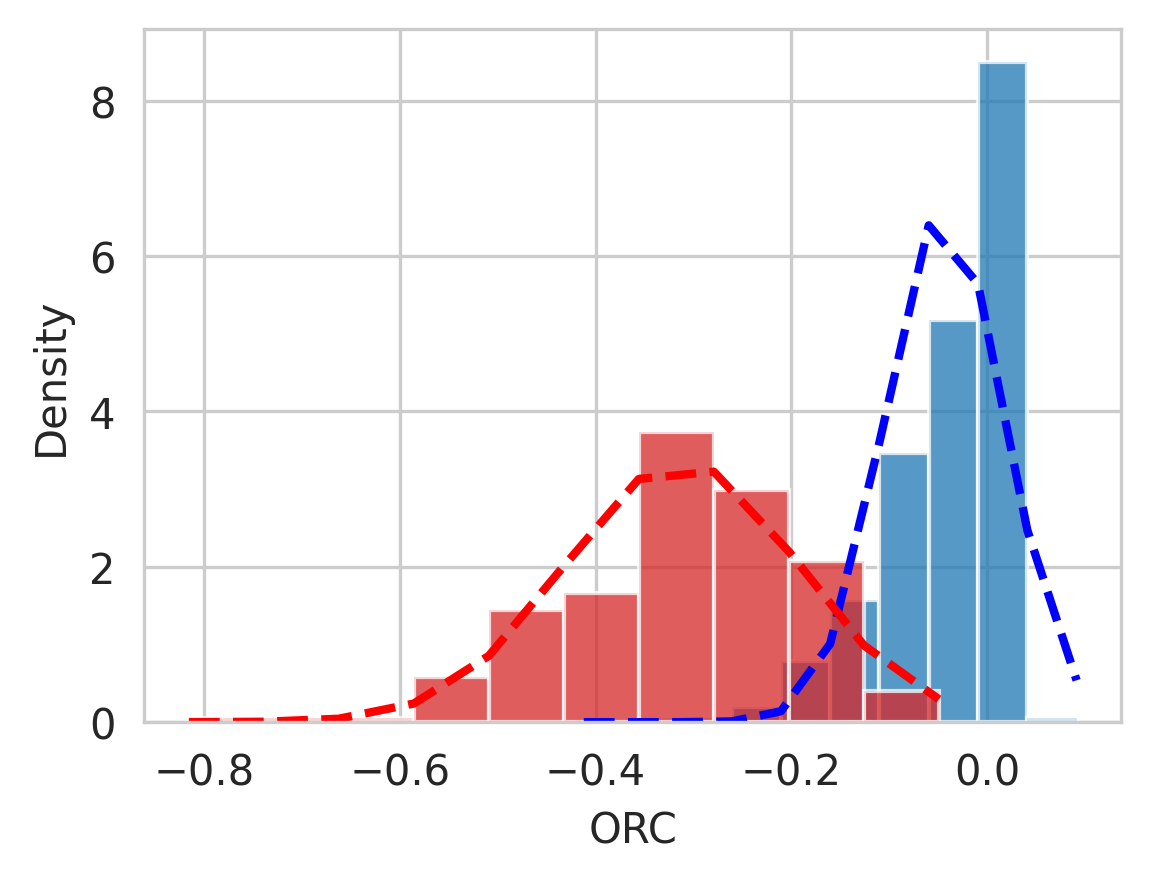}
    \end{tabular}
    \caption{Distribution of directed ORC (top) and complex-weighted ORC for the magnetic Laplacian with $\theta = 2\pi/3$ (bottom) of two-block directed SBMs with size $100$, $p_{\text{in}} = 0.5$ and $p_{\text{out}} = 0.01, 0.1, 0.2$ (left, middle, right, respectively), where within-community edges (blue) are directed (all from smaller label to larger one) and the between-community edges (red) are oriented randomly (equal probability from community 1 to 2 versus from community 2 to 1).}
    \label{fig:sbm-directed-dist-indouth}
\end{figure}

\end{document}